\documentclass[12pt,letter,american,intoc,bibliography=totoc,index=totoc,BCOR10mm,captions=tableheading,nottitlepage,fleqn]{article}

\usepackage{setspace, geometry}
\usepackage{etex}
\usepackage{babel,newcent,textcomp}

\usepackage{scrextend}
\usepackage{booktabs, threeparttable}

\usepackage[latin9]{inputenc}
\usepackage[style=bwl-FU, citestyle=authoryear, backend=biber, maxbibnames=99]{biblatex}
\usepackage{lmodern}

\usepackage[T1]{fontenc}
\usepackage[autostyle]{csquotes}
\usepackage{sgame}
\usepackage{fancyhdr}
\usepackage{amsmath,amssymb,bbm,nccmath}
\usepackage[amsthm]{ntheorem}
\usepackage{dsfont}
\usepackage{tablefootnote}
\usepackage[labelfont=bf]{caption}
\usepackage{enumitem}
\usepackage{geometry}
\usepackage{lscape}
\usepackage[unicode=true,
 bookmarks=true,bookmarksnumbered=true,bookmarksopen=true,bookmarksopenlevel=1,
 breaklinks=false,pdfborder={0 0 0},backref=false,colorlinks=true,linkcolor=black,citecolor=black,anchorcolor=black,urlcolor=blue]
 {hyperref}
\hypersetup{pdftitle={Your title},
 pdfauthor={Your name},
 pdfpagelayout=OneColumn, pdfnewwindow=true, pdfstartview=XYZ, plainpages=false}
 
\usepackage{istgame}

\usepackage{tikz}

\usetikzlibrary{shapes.geometric, arrows, positioning}
\tikzset{
    >=stealth',
    box/.style={
           rectangle,
           rounded corners,
           draw=black, very thick,
           text width=6.5em,
           minimum height=1.5em,
           text centered},
    boxsmall/.style={
           rectangle,
           rounded corners,
           draw=black, very thick,
           text width=6em,
           minimum height=1.5em,
           text centered},
     simple/.style={
           text centered},
    arrow/.style={
           ->,
           thick,
           shorten <=2pt,
           shorten >=2pt,}
}
\usepackage{graphicx}
\graphicspath{{""}}

\usepackage{cleveref}
\usepackage[figure]{hypcap}
\usepackage{calc}

\let\mySection\section\renewcommand{\section}{\suppressfloats[t]\mySection}

\theoremstyle{break}
\newtheorem{theorem}{Theorem}
\newtheorem{lemma}{Lemma}
\newtheorem{proposition}[theorem]{Proposition}

\usepackage{subfigure}

\newenvironment{definition}[1][Definition]{\begin{trivlist}
\item[\hskip \labelsep {\bfseries #1}]}{\end{trivlist}}

\makeatother

\begin{document}

\lhead[]{}

\rhead[]{}

\lfoot[\thepage]{}

\cfoot{}

\rfoot[]{\thepage}

\newpage

\let\endtitlepage\relax
\begin{titlepage}
\title{The Supply of Motivated Beliefs\thanks{\scriptsize{Thank you to Roland B\'enabou, Yan Chen, Ernesto Dal Bo, Christine Exley, Drew Fudenberg, David Laibson, Cesar Martinelli, Matthew Rabin, Mattie Toma, Leeat Yariv, Sevgi Yuksel, seminar participants at Exeter, HEC Paris, Imperial, LMU Munich, Michigan, Nottingham, Princeton, Toronto, UCL, and Warwick, and attendees at the ASSA Annual Meetings, Belief-Based Utility Workshop, NBER SI Political Economy, SITE: Psychology and Economics, and others for helpful comments. I am grateful for funding support from the Pershing Square Fund for the Foundations of Human Behavior. Preregistration plans for the two experiments are available on AsPredicted: (Primary experiment: \url{https://aspredicted.org/blind.php?x=VMH_8HW}; Additional experiment: \url{https://aspredicted.org/blind.php?x=GAQ_ZJF}.) Both experiments were approved by the IRB at Princeton University (13114-04 and 13114-05).}}}
\author{Michael Thaler\thanks{\scriptsize{University College London. Email: \href{michael.thaler@ucl.ac.uk}{michael.thaler@ucl.ac.uk}.}}}
\date{September 2023}  

\maketitle

\abstract{
When people choose what messages to send to others, they often consider how others will interpret the messages. A sender may expect a receiver to engage in motivated reasoning, leading the receiver to trust good news more than bad news, relative to a Bayesian. This paper experimentally studies how motivated reasoning affects information transmission in political settings. Senders are randomly matched with receivers whose political party's stances happen to be aligned or misaligned with the truth, and either face incentives to be rated as truthful or face no incentives. Incentives to be rated as truthful cause senders to be less truthful; when incentivized, senders send false information to align messages with receivers' politically-motivated beliefs. The adverse effect of incentives is not appreciated by receivers, who rate senders in both conditions as being equally likely to be truthful. A complementary experiment further identifies senders' beliefs about receivers' motivated reasoning as the mechanism driving these results. Senders are additionally willing to pay to learn the politics of their receivers, and use this information to send more false messages. 


\vspace{7mm}
\vspace{18mm}
}


\setcounter{page}{0}
\thispagestyle{empty}
\end{titlepage}

\pagebreak \newpage


\section{Introduction}

There has been a proliferation of interest in understanding how people communicate with others and why much of the news marketplace contains inaccurate information. ``Fake news,'' defined by \textcite{AG17} as the intentional reporting of false information, has been shown to be a contributing factor to the undermining of trust in public health, historical misconceptions, and the state of democracy (\cite{L-etal-18}; \cite{OW14}; \cite{PR21}). Given these large societal costs, it is important to better understand what motivates people to send false information to others in order to inform efforts to improve news dissemination and increase trust in credible sources. 

False information has been of particular concern in environments where information comes from unmoderated sources, such as on social media, and when information is about issues that evoke strong motivated beliefs, such as in politics (\cite{AG17}). Meanwhile, there is ample evidence that motivated beliefs affect people's demand for information (e.g. \cite{PI21}; \cite{GLMS22}), but much less is known about their effects on the \textit{supply} of information. 

This paper provides novel experimental evidence to show that motivated reasoning influences the supply of false information, focusing on political settings. It highlights the role of two factors that lead senders to be less truthful: (1) believing that receivers' motivated beliefs are misaligned with the truth, and (2) having incentives to be perceived as truthful. Both factors are prevalent in many political news transmission environments, and can play a significant role on social media, for two reasons. First, when news senders have information about their audience, they often can tailor messages to appeal to their receivers' politics. If receivers prefer to hold beliefs that are farther in the direction that aligns with their political party, they will be motivated to believe that ``good news'' about their party is true while ``bad news'' is false. Second, many social media platforms have users rate other users, incentivizing news senders to be rated well by others. For instance, Facebook implemented a policy in 2015 that enabled users to report news as false, and by 2018 began assigning users a credibility score based on the news they shared (\cite{WaPo18}). These incentives can perversely lead to more disinformation if users send false messages to appeal to receivers' motivated beliefs in order to boost their trustworthiness.


In order to have a controlled environment that is able to cleanly identify the effects of motivated beliefs and incentives, I turn to the (online) laboratory and run two large preregistered experiments. 
The primary experiment recruits self-reported social media users in the United States. 
Subjects send and receive messages about various factual issues that are chosen to evoke politically-motivated beliefs, as described in \Cref{topic-motives}. On each question, receivers report their prior belief about which of two answers is correct. Senders learn the true answer to the question and then choose, as a function of the receiver's prior, whether to send a message to the receiver that corresponds to the true answer or to the false answer. Finally, receivers rate the probability that the message from their sender is truthful. To identify the role that incentives play, senders are randomly assigned to either be paid as a function of receivers' ratings or to have their pay not depend on receivers' ratings. To identify the role that motivated beliefs play, senders are randomly assigned to receivers whose political party is either aligned with the truthful message, whose party is aligned with the false message, or whose party is unknown, and can condition on the receivers' prior beliefs.

\renewcommand\arraystretch{1.3}
\begin{table}
\begin{tabular}{l l l}  
\toprule
\textbf{Topic} & \textbf{Pro-Democrat Motives} & \textbf{Pro-Republican Motives} \\
\midrule
Immigrants' crime rate & Lower than US citizens & Higher than US citizens \\
Racial discrimination & Severe in labor market & Not severe in labor market \\
US crime & Got worse under Trump & Got better under Trump \\
Media bias & Media not mostly Dems & Media mostly Dems \\
COVID-19 restrictions & Mostly stopped spread & Did not mostly stop spread \\
Gun reform & Decreased homicides & Did not decrease homicides \\
Unemployment & Got worse under Trump & Got better under Trump \\
Wages & Grew slower under Trump & Grew faster under Trump \\
Undocumented immigrants & Mostly overstaying visas & Mostly illegally entered US \\
Domestic terrorism & More due to white supremacy & More due to other factors \\
Poverty rates & Got worse under Trump & Got better under Trump \\
Illegal immigration & Not historically high & Historically high  \\
\bottomrule
\end{tabular}
\begin{scriptsize}
\caption{The list of political topics and hypothesized motives in the experiments.}
\label{topic-motives}
\end{scriptsize}
\end{table}

\renewcommand\arraystretch{1}

The main result from the primary experiment is that senders are more likely to send false messages when they face incentives to be rated as truthful and when the receiver's politics are aligned with the false message. When faced with receivers with the same prior belief, senders with these incentives are 7.1 percentage points (subject-level clustered s.e. 3.2 pp) more likely to send a false message when the receiver's party is aligned with the false message than when it is aligned with the true message. Senders are also more likely to send false messages when the receiver's party is aligned with the false message than when the receiver's party is unknown or when the topic is neutral instead of political. Since senders condition on receivers' prior beliefs, these results indicate that senders are tailoring messages to receivers' politics separately from their priors. 

There are many alternative explanations as to why senders may condition messages on receivers' party, such as altruism or innate preferences to be seen as truthful. These factors are not directionally driving results, as without incentives, senders no longer condition on receivers' party. Relatedly, the second main result is that incentives have a negative effect on truthfulness overall. Senders are 7.3 pp (s.e. 2.5 pp) more likely to send false messages when incentivized, an increase of 34 percent. The negative effect of incentives is driven by cases in which the receiver's party is misaligned with the truth. 

Results are similar in an additional treatment in which senders cannot condition on receivers' priors: incentivized senders condition messages on receivers' party even more, but unincentivized senders do not on average, and incentives lead to significantly more false messages being sent. Results are also robust to other specifications and cuts of the data. 

Next, we turn to receivers' behavior. Receivers do not incorporate the negative effects of incentives when rating the truthfulness of senders. Neither receivers' ratings nor beliefs are statistically-significantly affected by senders' incentives, and the point estimates are tightly estimated around zero. Survey evidence confirms these patterns: Senders self-report choosing messages that are more aligned with the receiver's party, and less aligned with the truth, when they are incentivized, but receivers' predictions of sender behavior are not significantly affected by senders' incentives. 

In an additional experiment, I further identify the mechanism of motivated reasoning, defined as the distortion of belief updating to favor preferred states (\cite{K90}; \cite{BT02}; \cite{K16}; \cite{T-WPa}). My preferred explanation from the primary experiment is that senders believe that receivers engage in motivated reasoning, and tailor messages to appeal to this bias. The primary experiment is able to rule out other preference-based explanations by comparing incentivized to unincentivized senders, but higher-order belief explanations are harder to rule out. For instance, if senders believe that receivers expect senders to distort messages in the direction opposed to receivers' party, leading them to trust pro-party messages more (such as in \cite{M01}). To identify the motivated-reasoning mechanism, the additional experiment breaks the interaction between senders and receivers, isolating the beliefs channel directly. 

I identify motivated reasoning among receivers using a version of the design from \textcite{T-WPa}. This design has two main steps. First, each receiver is given a variety of factual questions with numerical answers. On each question, the receiver gives a guess that they think is equally likely to be above or below the correct answer; that is, the median of their belief distribution is elicited. Second, the receiver is given two binary messages from the computer: one message is true and and the other is false. Each message tells them whether the answer was above or below their median. The receiver is not told whether each source is true or false; instead, they are asked to make inferences about its veracity from the message content. Since messages relate the true answer to the receiver's median, a Bayesian would believe that it is equally likely for each source to report either message.\footnote{That is, the receiver has stated that they believe the answer is equally likely to be greater than or less than her median; so, they believe the likelihood that a true message would report that the answer is greater than their median is 1/2, and the likelihood that a false message would report that the answer is less than their median is also equal to 1/2, leading a ``greater than'' or ``less than'' message to be completely uninformative about the veracity of the news source.} On the other hand, a receiver who engages in motivated reasoning will think the news is more likely to be true if it sends a message that aligns more with their preferred (political) beliefs. In \textcite{T-WPa}, I argue that this method is a well-identified and well-powered way to identify motivated reasoning.\footnote{While there are other experimental approaches to identifying motivated reasoning, such as those in \textcite{MNNR-WP} and \textcite{ER11}, those approaches have often found it difficult to cleanly detect motivated reasoning, as identification of the bias from Bayesian updating and other inference biases is noisy (\cite{B19}; \cite{TPR20a}; \cite{TPR20b}).} Indeed, I find significant motivated reasoning among receivers here as well.

After receivers play the game above, each sender is given a matched receiver's median belief and are asked to select either the ``greater than'' or the ``less than'' message of the computer. The two primary treatment arms are the same as those in the primary experiment. First, senders are either incentivized to have the receiver's rating of the message be implemented for payment, or they are unincentivized. Second, senders are matched to receivers whose party is either aligned or misaligned with the truth. The median beliefs provided to the sender are selected such that at least one Democrat and at least one Republican has stated that median. This enables random assignment, for a given prior belief, to a receiver who is a Democrat, is a Republican, or is equally likely to be of either party. 

The main finding from the additional experiment coincides with the main finding from the primary experiment; senders who are incentivized to be perceived as truthful are more likely to choose the false message when it is aligned with the receiver's party, while unincentivized senders are not. Incentivized senders are 14.5 percentage points (s.e. 4.5 pp) more likely to send a false message when the receiver's party is aligned with the false message than when it is aligned with the true message. Unincentivized senders are not directionally affected by this treatment. Incentives again lead to more false messages chosen. Senders are 8.5 pp (s.e. 2.3 pp) more likely to send false messages overall when incentivized, a 34 percent increase, which is driven by the condition in which the receiver's party is misaligned with the truth. This result clarifies the form of politically-motivated reasoning that determines senders' beliefs and choices of false messages. 

The second finding from the additional experiment is that senders demand information about receivers. A majority of senders are willing to pay a positive amount to learn the political party of receivers on political questions, and a majority are not willing to on neutral questions. Senders use receivers' party information in order to choose more false news, suggesting that giving news providers the option to learn about their audience may cause the audience to receive less truthful pieces of news when motivated reasoning is at play. 

In both experiments, I elicit beliefs about others, finding clear evidence that senders expect receivers to use their politically-motivated beliefs, in addition to prior beliefs, when inferring the truthfulness of news. Senders' beliefs are directionally accurate, but on average they overstate the impact of party on inference.

Theoretically, the main experimental findings can be explained as equilibria of a simple sender-receiver game. In this game, the sender prefers to tell the truth and has incentives for the receiver to believe that she is telling the truth, and the receiver prefers accurately rating the sender. There is always a separating Bayesian Nash equilibrium (BNE) in which the sender reports the truth and the receiver believes that all sent messages are truthful. However, when the receiver engages in motivated reasoning, directionally distorting his posteriors in directions he prefers, there may no longer be a separating BNE, as the sender's best response is to bias messages in the receiver-preferred direction.\footnote{More technically, I discuss a solution concept in which the receiver plays a strategy of a BNE in which he has one particular level of bias, and the sender --- who is aware of the receiver's strategy and beliefs about his own bias --- plays a best response to the receiver under the assumption that he has another level of bias.} When the sender and receiver have different beliefs about the receiver's motivated reasoning, and the receiver is unaware of this difference, incentives can lead to a gap between the strategy the sender plays and the strategy the receiver plays a best response to. In particular, the sender's incentives lead her to send less truthful messages but do not lead the receiver to rate her as less truthful, predictions that are consistent with the experimental results. 

This paper adds to several strands of literature. It is among the first papers to unite the growing literature on the demand for motivated beliefs with the large literature on the supply of misinformation and disinformation. The definition of motivated reasoning I use was developed in economics by \textcite{BT02} and was further formalized in a series of subsequent papers by these authors (\cite{BT11}; \cite{B13}; \cite{BT16}).\footnote{The ``optimal beliefs'' framing discussed in my model is also similar to the setups of \cite{BP05} and \cite{MNNR-WP}).} I contribute to the understanding of the relationship between politics and misperceptions has been studied extensively in applied political settings (e.g. \cite{TL06}; \cite{AST18}; \cite{AMS20}). \textcite{B13}, \textcite{L14}, and \textcite{McGee-WP} theoretically study beliefs about others' motivated reasoning, but to my knowledge, my paper is the first to \textit{experimentally} study such higher-order beliefs.\footnote{My model has overlaps with \textcite{L14}. He studies a setting in which a policy-maker exploits the demand for motivated beliefs by voters by choosing whether to supply information to them and finds that there is often a truthful equilibrium as well as non-truthful equilibria. In my setting, I find an even starker setting in which truthful equilibria \textit{cannot} exist.} \textcite{HL-WP} study persuasion in a motivated setting, but focus on senders and receivers with materially-misaligned incentives. The additional experiment further extends the experimental design of \textcite{T-WPa} (which also studies politics), following a set of papers that study motivated reasoning about various topics including gender differences (\cite{T21}), racial discrimination (\cite{E-WP}), and ``good news'' about the world (\cite{T-WPb}). 

These results relate to the experimental cheap talk literature and preferences for truth-telling. As in \textcite{ANR19}, I find that people prefer to send truthful messages but that this preference is malleable. Compared to most cheap-talk and information-design games in the literature (e.g. \cite{CW06}; \cite{FLP-WP}), my game does have a truth-telling BNE, but still finds considerable lying and distrust.\footnote{This result is quite different from \textcite{CW06}, who find excess trust and truth-telling absent motivated beliefs.} A number of other papers experimentally study the determinants of lying (e.g. \cite{EG12}; \cite{SvDP11}; \cite{GKS18}), receivers' behavior in detecting lies in nonpolitical contexts (\cite{SG-WP}); and non-motivated reasons why people may lie to be seen as truthful (\cite{CSC20}; \cite{B-WP}; \cite{SSvdWV-WP}). I add to this literature by emphasizing behavior of senders and relating lying to beliefs about others' motivated reasoning.

This paper also adds to the psychology literature on bias blind spots (\cite{PLR02}; \cite{PGR04}; \cite{P07}) and persuasion of news (\cite{D83}) --- which finds that people see themselves as less susceptible to biases than others --- by considering belief gaps about motivated reasoning.\footnote{In economics, \textcite{F-WP} finds similar results in the context of self-control problems.} It also relates to papers that experimentally study higher-order beliefs about motivated and unmotivated biases in social settings (\cite{OY20}; \cite{BBG-WP}; \cite{G-B-WP}), and theoretically to literatures on higher-order belief biases like bounded strategic rationality (\cite{C03}), cognitive hierarchy (\cite{CHC04}), and cursed equilibrium (\cite{ER05}). \textcite{SS23} find complementary evidence about what people believe about \textit{their own} motivated reasoning. My data provide evidence for bias blind sports about motivated reasoning, in which people act as if \textit{others} treat them as if they were less biased. 

Finally, these findings contribute to the understanding of the spread of disinformation on social media. They show that even small incentives can significantly change behavior, and that disinformation is especially prevalent in contexts that evoke motivated beliefs. Related literature analyzes how news providers distort messages in the direction of their audiences' \textit{current} beliefs (\cite{GS06}; \cite{GS10}; \cite{MS05}); I show that appealing to motivated beliefs is independently important, and may even provide an alternative explanation for previous findings. These findings also provide a belief-driven explanation for pandering, in which senders bias messages towards receivers' \textit{preferences} (\cite{MT04}; \cite{CDK13}; \cite{CHS01}). Here, motivated reasoning impedes truth-telling equilibria that exist even when receivers have differing priors, a gap that does not typically occur in individual decision making (\cite{L-WP}). 

The rest of the paper proceeds as follows: \Cref{theory} develops the sender-receiver model. \Cref{experiment1} presents the design and results of the primary experiment, and \Cref{experiment2} presents the additional experiment. \Cref{conclusion} concludes and proposes directions for future work. Additional results and study materials are in the appendices.

\section{Theory}
\label{theory}

\subsection{A signaling game}

This section introduces a simple sender-receiver game that forms the basis of the primary experiment. I consider a game with one sender (she) and one receiver (he). There is a state of the world which is either high ($\theta_H$) or low ($\theta_L$). The sender observes the true state of the world and the receiver wants to predict the true state of the world. The sender has a preference for telling the truth and may receive an additional benefit for having the receiver think she is telling the truth. The sender may benefit from being perceived as truthful because the receiver wishes to consume more news from her or because the receiver will be more likely to ``like'' or ``share'' the message.

After observing the sender's message, the receiver is asked to rate the probability that the sender was telling the truth, receiving a payoff that depends on the accuracy of his rating. These ratings are functionally equivalent to rating the probability that the state is high or low. The receiver's payoff is a reduced-form interpretation of accuracy motives.


More formally, the timing of the game is as follows: First, Nature chooses whether the state $\theta$ is $\theta_H$ (with probability $\pi$) or $\theta_L$ (with probability $1-\pi$). We will later interpret $\pi$ as reflecting R's prior belief of the state. Next, S learns the state, and then chooses whether to send message $x_H$ or $x_L$ to R. After observing the message, R takes action $a \in [0,1]$. 

The receiver's expected utility is set so that he is incentivized to report his true belief: $u_R = 1-(1-a)^2$ if the sender's message is truthful and $u_R = 1-a^2$ if the message is false. The sender's utility includes two components. She has an intrinsic value to report the truth, receiving $\tau > 0$ for truth-telling. She also receives utility that is linear in the receiver's rating $a$ of her truthfulness. $\gamma \geq 0$ corresponds to the weight put on the receiver's rating. 

The expected utility matrix, as a function of $a$, is as follows:

\begin{center}
\vspace{4mm}
\begin{tabular}{|l|c|c|}
\hline
  & Nature chooses $\theta_H$ & Nature chooses $\theta_L$ \\
 \hline
 Sender chooses $x_H$ & $\gamma a+\tau, 1-(1-a)^2$    & $\gamma a, 1-a^2$ \\
 \hline
 Sender chooses $x_L$ & $\gamma a, 1-a^2$    & $\gamma a+\tau, 1-(1-a)^2$ \\
 \hline
\end{tabular}
\end{center}
\vspace{4mm}

\noindent The first element in each cell is the sender's utility, and the second is the receiver's utility. \Cref{main-game} describes the game in extensive form. 

\begin{figure}
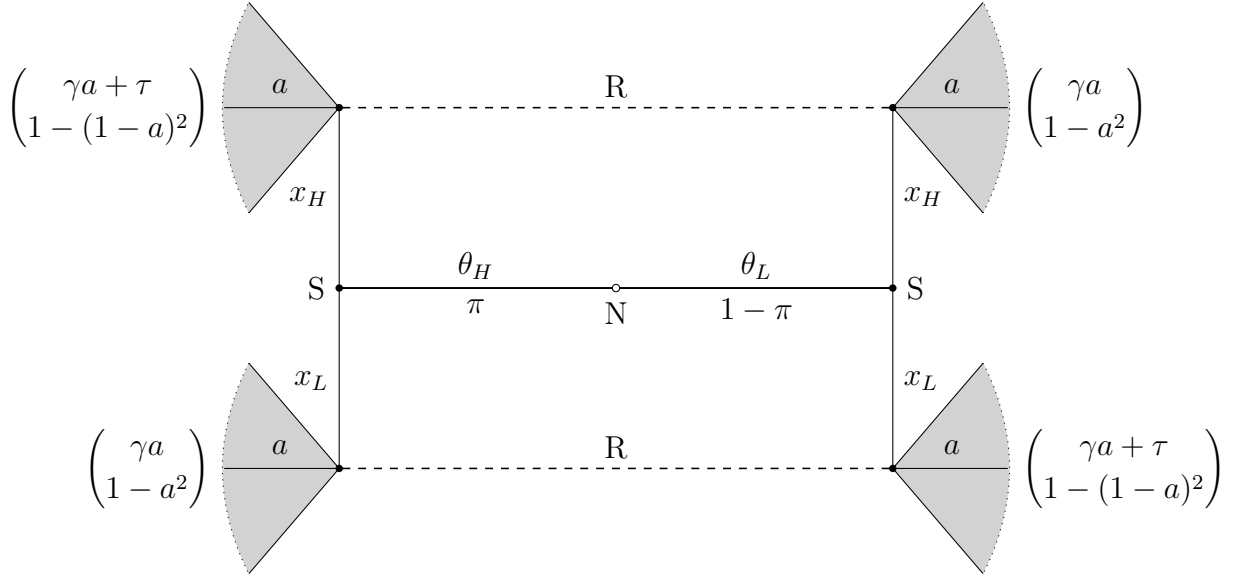

    \centering
\begin{istgame}[scale=2]
\xtdistance{18.4mm}{18.4mm}
\istroot(0)[chance node]<below>{N}
\istB<grow=left>{\theta_H}[a]{\pi}[b]
\istB<grow=right>{\theta_L}[a]{1-\pi}[b]
\endist
\xtdistance{12mm}{20mm}
\istroot(1)(0-1)<180>{S}
\istb<grow=north>{x_H}[l]
\istb<grow=south>{x_L}[l]
\endist
\istroot(2)(0-2)<0>{S}
\istb<grow=north>{x_H}[r]
\istb<grow=south>{x_L}[r]
\endist
\cntmdistance{6mm}{14mm}
\cntmApreset[dotted][gray!35]
\istrootcntmA'[west](a1)(1-1)
\istbA{a}[a]{\dbinom{\gamma a + \tau}{1-(1-a)^2}}
\endist
\istrootcntmA[west](b1)(1-2)
\istbA{a}[a]{\dbinom{\gamma a}{1-a^2}}
\endist
\istrootcntmA[east](b2)(2-1)
\istbA{a}[a]{\dbinom{\gamma a}{1-a^2}}
\endist
\istrootcntmA'[east](a2)(2-2)
\istbA{a}[a]{\dbinom{\gamma a + \tau}{1-(1-a)^2}}
\endist
\xtInfoset[dashed](a1)(b2){R}
\xtInfoset[dashed](b1)(a2){R}
\end{istgame}
    \caption{The extensive-form game. S's payoffs are listed on top; R's payoffs are listed on bottom. Dashed lines denote information sets.}
    \label{main-game}
\end{figure}

We first consider the Bayesian Nash equilibria (BNE) of the game. There is always a full-information separating BNE in which the sender plays a \textit{truthful strategy} --- sending $x_H$ given $\theta_H$ and sending $x_L$ given $\theta_L$ --- and the receiver chooses $a(x_H)=1$ and $a(x_L)=1$. The sender earns $\gamma + \tau$ and the receiver earns 1, and this BNE is Pareto optimal.\footnote{I focus on BNE, instead of equilibrium refinements, for simplicity. The truthful BNE is also a perfect Bayesian equilibrium and a sequential equilibrium.} As is common in coordination games, there may be other equilibria. In this game, the existence of other equilibria depends on $\tau / \gamma$ being sufficiently small.\footnote{Specifically, there is a pooling BNE in which the sender always sends $x_H$ iff $\tau/\gamma \leq \pi$. In such a BNE, the receiver plays $a(x_H) = \pi$ and $a(x_L) = \pi'$ for $\pi' < \pi - \tau/\gamma$. Similarly, there is an $x_L$-pooling BNE iff $\tau/\gamma \leq 1-\pi$. The receiver plays $a(x_H) = 1-\pi''$ and $a(x_L) = 1-\pi$ for $1 - \pi'' < 1 - \pi - \tau/\gamma$. There can also be mixed-strategy equilibria. If $\tau/\gamma \leq 1-\pi$, there is a BNE in which the sender sends $x_H|\theta_H$ with probability 1, and sends $x_L|\theta_L$ with probability $\frac{1-\pi-\tau/\gamma}{(1-\pi)(1-\tau/\gamma)}$. The receiver plays $a(x_H) = 1-\tau/\gamma$ and $a(x_L) = 1$. If $\tau/\gamma \leq \pi$ there is a BNE in which the sender sends $x_L|\theta_L$ with probability 1 and sends $x_H|\theta_H$ with probability $\frac{\pi-\tau/\gamma}{\pi(1-\tau/\gamma)}$. The receiver plays $a(x_L) = 1-\tau/\gamma$ and $a(x_H) = 1$.}

\subsection{Motivated reasoning}

When people receive information, they often distort how they process the information in a non-Bayesian way that favors beliefs they prefer to hold: \textit{motivated reasoning}. I model motivated reasoning here by having the receiver optimally form a posterior that trades off the benefit from believing the state is good with accuracy motives. This model follows a common approach from past literature (\cite{BT02}; \cite{BP05}; \cite{BT11}). 

We formally consider how motivated reasoning affects behavior in the game described above. The receiver either sees an $x_H$ signal or an $x_L$ signal; without loss of generality, suppose that $x_H$ is ``good'' and $x_L$ is ``bad.'' Denote $P(x_H | H) = p_H$ and $P(x_L | L) = p_L$, where $p_H \geq 1-p_L$, so that $x_H$ is weakly more likely in state $H$ and $x_L$ is weakly more likely in state $L$.\footnote{Since $\tau>0$, this will be true in all equilibria of this game.}

A Bayesian receiver would give the following ratings:
\begin{align*}
    a(x_H) = P(H|x_H) &= \frac{\pi p_H}{\pi p_H + (1-\pi) (1-p_L)}, \\
    a(x_L) = P(L|x_L) &= \frac{(1-\pi) p_L}{\pi(1-p_H) + (1-\pi) p_L}.
\end{align*}
Motivated reasoners act as if they receive additional utility for forming a posterior that is consistent with the good state. Specifically, they receive a benefit of $\lambda \cdot a$ towards positively rating signals that indicate that the state is high and $\lambda (1-a)$ towards negatively rating signals that indicate the state is low.\footnote{Predictions are similar if overweighting of good news is less severe than underweighting of bad news (e.g. \cite{BT02}; \cite{MNNR-WP}).}

The updated utility matrix, as a function of $a$, is as follows:

\vspace{5mm}
\begin{tabular}{|l|c|c|}
\hline
  & Nature chooses $\theta_H$ & Nature chooses $\theta_L$ \\
 \hline
 Sender chooses $x_H$ & $\gamma a+\tau, 1-(1-a)^2\pmb{{\color{blue}+ \lambda a}}$    & $\gamma a, 1-a^2\pmb{{\color{blue}+ \lambda a}}$ \\
 \hline
 Sender chooses $x_L$ & $\gamma a, 1-a^2\pmb{{\color{blue}+ \lambda (1-a)}}$    & $\gamma a+\tau, 1-(1-a)^2\pmb{\color{blue}{+ \lambda (1-a)}}$ \\
 \hline
\end{tabular}
\vspace{5mm}

\noindent The blue bold text denotes the added term. The motivated receiver gives ratings that equal:
\begin{align}
    a(x_H) &= \min \left\lbrace \lambda/2 + \frac{\pi p_H}{\pi p_H + (1-\pi) (1-p_L)}, 1 \right\rbrace \nonumber \\
    a(x_L) &= \max \left\lbrace -\lambda/2 + \frac{(1-\pi) p_L}{\pi(1-p_H) + (1-\pi) p_L}, 0 \right\rbrace. \label{eq:br-receiver}
\end{align}
We assume $\lambda \in (-2,2)$ in order to allow for interior ratings to be optimal.



\subsection{Specifying higher-order beliefs}

We now allow for incomplete information over the receiver's type $\lambda$, and particularly do not assume that there is common knowledge over $\lambda$; the sender and the receiver may have different beliefs about $\lambda$ and may have different higher-order beliefs. 

I denote the receiver's first-order belief about $\lambda$ as a probability distribution $\Delta_R(\lambda)$.\footnote{One natural assumption is that the receiver knows $\lambda$, but higher-order beliefs may be wrong.} The receiver's second-order belief reflects his belief about the sender's belief $\Delta_R(\Delta_S(\lambda))$. The receiver's third-order belief reflects his belief about the sender's belief about his belief: $\Delta_R(\Delta_S(\Delta_R(\lambda)))$, and so on.

Similarly, the sender's first-order belief about $\lambda$ is a probability distribution $\Delta_S(\lambda)$. The sender's second-order belief reflects her belief about the receiver's belief: $\Delta_S(\Delta_R(\lambda))$. The sender's third-order belief reflects her beliefs about the receiver's belief about her belief: $\Delta_S(\Delta_R(\Delta_S(\lambda)))$, and so on. 

Formally, I define each player's belief hierarchy in the following recursive manner:\footnote{I assume common knowledge over other game characteristics, such as over the sender's knowledge of the state. An extension could look at cases in which senders may misinfer from Nature's signals.} 

\begin{align*}
    \Delta_{S,k+1}(\lambda) &= \Delta_S(\Delta_{R,k}(\lambda)) \text{ and} \\
    \Delta_{R,k+1}(\lambda) &= \Delta_R(\Delta_{S,k}(\lambda)) \text{ for each }k=1,2,...
\end{align*}

For simplicity, we assume that both the sender and the receiver have point beliefs for each element in their belief hierarchy. I denote $\hat{\lambda}_{i,k}$ to be player $i$'s point belief $\Delta_{i,k}(\lambda)$, and will omit the subscript when $k=1$. In line with the psychological literature on bias blind spots (\cite{PLR02}; \cite{PGR04}; \cite{P07}), we assume that the receiver's belief about $\lambda$ is lower than the sender's belief about $\lambda$: $\hat{\lambda}_{S} > \hat{\lambda}_{R}$. We assume nothing about the accuracy of first-order beliefs, except that $\hat{\lambda}_i \in [0,2)$ for simplicity. One natural possibility is that the receiver knows his $\lambda$, so $\hat{\lambda}_{R} = \lambda$ (and so the sender overstates the receiver's bias).

To close the model, we assume that the receiver projects his belief onto the sender when constructing higher-order beliefs, and that the sender is aware of this projection (e.g. \cite{McGee-WP}):

\begin{align*}
    \hat{\lambda}_{S,k} &= \hat{\lambda}_{R} \text{ and} \\
    \hat{\lambda}_{R,k} &= \hat{\lambda}_{R} \text{ for each }k=2,3,...
\end{align*}

One psychology behind such a formulation is that individuals believe they are unique in their low level of bias. That is, a receiver thinks of himself as relatively unbiased --- in contrast to others --- and projects his perception of his unbiasedness onto others' beliefs.\footnote{An alternative theory is \textcite{C03}, which considers a setting in which sender beliefs about receiver bounded rationality lead to profitable cheap talk in zero-sum games. That setup can lead to similar predictions when receivers are more boundedly rational about strategies than senders are. This asymmetry is hard to rule out, but since roles are randomly assigned, it seems less plausible than having asymmetry in beliefs about ones own versus others' biases.}

\subsection{Motivated equilibrium}

I now consider how beliefs about $\lambda$ affect equilibrium behavior. It is possible for the receiver plays as if his strategy is in a BNE of the game in which $\lambda = \hat{\lambda}_R$ is common knowledge, while the sender plays a best response to the receiver's strategy under her belief that $\lambda = \hat{\lambda}_S$. 

\begin{definition}
The strategy profile $\left(s_S (\theta, \{ \hat{\lambda}_{S,k} \} ), s_R (x, \{ \hat{\lambda}_{R,k} \} ), s_R (x, \lambda ) \right)$ constitutes a \textit{motivated equilibrium} of the game $\Gamma = \left[I = \{S,R\},\{s_i\}_{i \in I}, \{u_i\}_{i \in I}, \Theta, \pi, \lambda, \{\hat{\lambda}_i\}_{i \in I}\right]$ described above if:
\begin{itemize}
    \item $s_S (\theta, \{ \hat{\lambda}_{S,k} \})$ is a best response to $s_R(x, \{\hat{\lambda}_{R,k} \})$;
    \item $s_R(x, \{\hat{\lambda}_{R,k}\})$ is a best response to $s_S (\theta, \{ \hat{\lambda}_{R,k} \})$; and
    \item $s_S (\theta, \{ \hat{\lambda}_{R,k} \})$ is a best response to $s_R(x, \{\hat{\lambda}_{R,k} \})$.
\end{itemize}
\end{definition}

In a motivated equilibrium, the receiver's strategy is a best response (given the true value of $\lambda$) to a hypothetical sender who holds the receiver's belief $\hat{\lambda}_{R}$ about $\lambda$ and who plays a best response to the receiver's strategy. The actual sender, who holds belief $\hat{\lambda}_{S}$ about $\lambda$ and is aware of the receiver's higher-order beliefs, plays a best response to the receiver's strategy. 

\begin{proposition}
    There exists a motivated equilibrium for this game (possibly in mixed strategies). 
\end{proposition}
Existence of a motivated equilibrium follows immediately from existence of BNE. A BNE exists in this game for all $\hat{\lambda}_R$, since strategy spaces are compact in $\mathbb{R}$ and utilities are continuous (\cite{G52}). This describes R's strategy $s_R(x, \{\hat{\lambda}_{R,k}\})$ in a motivated equilibrium. Likewise, there exists a best response for S to R's strategy for $\hat{\lambda}_S$, which is $s_S (\theta, \{ \hat{\lambda}_{S,k} \})$. 

We can strengthen this proposition by calculating the pure-strategy equilibria.

\begin{proposition} \label{prop:list-eq}
There exists a motivated equilibrium for this game in pure strategies. \Cref{motivated-equilibria} enumerates the six possible motivated equilibria in pure strategies. The existence of each equilibrium depends on an inequality relating $\tau / \gamma$ to functions of other parameters.

\renewcommand\arraystretch{1.3}
\begin{table}[ht]
    \centering
\begin{tabular}{l l l}  
\toprule
\textbf{R plays BR to $s_S(\cdot)$} & \textbf{S's strategy $s_S(\cdot)$} & \textbf{Conditions} \\
      \hline
      $x_H | \theta_H$ and $x_L | \theta_L$  &  $x_H | \theta_H$ and $x_L | \theta_L$  
      & $\tau/\gamma \geq \hat{\lambda}_S/2$  \\
      \hline
      $x_H | \theta_H$ and $x_H | \theta_L$  &  $x_H | \theta_H$ and $x_H | \theta_L$ 
      & $\tau/\gamma \leq \min\{\pi + \hat{\lambda}_R/2,1\} - a(x_L)$   \\
      \hline
      $x_L | \theta_H$ and $x_L | \theta_L$  &  $x_L | \theta_H$ and $x_L | \theta_L$ 
      & $\tau/\gamma \leq \max\{1-\pi - \hat{\lambda}_S/2,0\} - a(x_H)$   \\
      \hline
      $x_H | \theta_H$ and ${\pmb{\color{teal}x_L | \theta_L}}$  &  $x_H | \theta_H$ and ${\pmb{\color{teal}x_H | \theta_L}}$ 
      & $\tau/\gamma \in [\hat{\lambda}_R/2, \hat{\lambda}_S/2]$  \\
      \hline
      ${\pmb{\color{teal}x_L | \theta_H}}$ and $x_L | \theta_L$  &  ${\pmb{\color{teal}x_H | \theta_H}}$ and $x_L | \theta_L$ 
      & $\tau/\gamma \leq   \max\{1-\pi - \hat{\lambda}_R/2,0\} - a(x_H)$, \\
      & & $\tau/\gamma \geq \max\{1-\pi - \hat{\lambda}_S/2,0\} - a(x_H)$,    \\
      & & $\tau/\gamma \geq -(\max\{1-\pi - \hat{\lambda}_S/2,0\} - a(x_H))$    \\
      \hline
      ${\pmb{\color{teal}x_L | \theta_H}}$ and ${\pmb{\color{teal}x_L | \theta_L}}$  &  ${\pmb{\color{teal}x_H | \theta_H}}$ and ${\pmb{\color{teal}x_H | \theta_L}}$ 
      & $\tau/\gamma \leq \max\{1-\pi - \hat{\lambda}_R/2,0\} - a(x_H)$, \\
      & & $\tau/\gamma \leq -(\max\{1-\pi - \hat{\lambda}_S/2,0\} - a(x_H))$    \\
\bottomrule
\end{tabular}
\begin{scriptsize}
\caption{The list of pure-strategy motivated equilibria and their corresponding conditions. Teal, bold text: S's motivated equilibrium strategy differs from the strategy that R plays a best response to.}
\label{motivated-equilibria}
\end{scriptsize}
\end{table}

\renewcommand\arraystretch{1}
\end{proposition}

Motivated equilibria can be categorized into three strategy profiles in which the receiver strategy is a best response to the sender's strategy as per \Cref{eq:br-receiver}, and three strategy profiles in which the receiver's strategy is not a best response. I discuss the equilibria below and how their existence depends on the parameters in the game, and the details of the proof are relegated to \Cref{appendix-a}.

The first category includes a truthful equilibrium, an $x_H$-pooling equilibrium, and an $x_L$-pooling equilibrium, just as in the BNE case. Here, the truthful equilibrium only occurs when $\gamma \cdot \hat{\lambda}_S$ is sufficiently small. This is because the sender expects the receiver to give a higher rating to $x_H$ than $x_L$, so she faces a tradeoff between honesty $\tau$ and incentives $\gamma$. When $\gamma$ is sufficiently small, the truthful equilibrium exists and is unique, but when $\gamma$ is sufficiently large, no truthful equilibrium exists.

Motivated equilibria in the second category are possible for intermediate values of $\gamma$: The receiver plays a best response to a sender who believes $\lambda$ is $\hat{\lambda}_R$, while the sender --- who believes $\lambda$ is $\hat{\lambda}_S > \hat{\lambda}_R$ --- plays a different strategy. In each of these equilibria, the receiver plays a best response to a sender who plays $x_L$ in some state of the world, while the sender actually plays $x_H$ in that state. This form of naivete has a similar flavor to cognitive hierarchy (\cite{CHC04}) and cursed equilibrium (\cite{ER05}). 

Of these three equilibria, two involve the receiver playing a best response to a sender who plays an $x_L$-pooling strategy, and one involves the receiver playing a best response to a sender who plays a truthful strategy. The latter equilibrium is particularly of interest because of its psychological plausibility; the receiver expects the sender to play a truthful strategy because he thinks she thinks he is not very biased, while the sender instead always sends the $x_H$ message. An example of possible motivated equilibria using values estimated in the experiment, and where higher-order beliefs are varied, is presented in \Cref{fig:eq-viz}.

I now summarize the predictions of the model in order to generate testable hypotheses for the primary experiment. 

\textbf{Prediction 1:} Increasing the sender's incentives from $\gamma = 0$ to $\gamma \in (2\tau / \hat{\lambda}_S, 2\tau / \hat{\lambda}_R)$ leads to a motivated equilibrium in which fewer truthful messages are sent by the sender, but there is no change in strategy by the receiver. 

\textbf{Prediction 2:} In such an equilibrium, senders asymmetrically ``deviate'' in the high direction. That is, the receiver rates $x_H$ messages as being more likely to be truthful than $x_L$ messages, but the sender is actually \textit{less} truthful when sending $x_H$. 

\textbf{Prediction 3:} When the sender has no incentives $\gamma=0$ to be perceived as truthful, she does not condition messages on $\hat{\lambda}_S$. 

It is worth noting that senders condition on $\pi$ in some ME as well, and are weakly more likely to send $x_H$ as $\pi$ increases. When $\gamma$ is low, equilibria do not depend on $\pi$; however, when $\gamma$ is intermediate, the existence of an $x_H$-pooling ME strategy depends on $\pi$ being sufficiently large and the existence of an $x_L$-pooling ME strategy depends on $\pi$ being sufficiently small. Note also that the specification of higher-order beliefs is not without loss of generality.\footnote{For instance, another possibility is that $\hat{\lambda}_{S,k} = \hat{\lambda}_{S} \text{ and } \hat{\lambda}_{R,k} = \hat{\lambda}_{S} \text{ for each }k=2,3,....$ In this case, both $S$ and $R$ would believe $S$ will play the strategy of the BNE in which $\lambda = \hat{\lambda}_{S}$, while $R$ expects himself to actually play as if $\lambda = \hat{\lambda}_{R}$, leading to different ME than in the model.}

\section{Primary Experiment}
\label{experiment1}

\subsection{Design}\label{expt1-design}
Below, I outline the timing, treatment arms, and main hypotheses of the game, which follows the setup of the model. Screenshots for the pages subjects see are in \Cref{screenshots1}. 

Subjects are randomly assigned to be senders or receivers. Receivers state their prior likelihood that the answer to a factual question is greater or less than a target number. Senders learn whether the answer is actually greater or less and choose (via the strategy method) whether to send a message that says ``The answer is greater than [the target]'' or ``The answer is less than [the target].'' For each message, receivers state (via the strategy method) how likely it is that their sender's answer is truthful.

To fix ideas, consider the following question that subjects see in the experiment: 

\vspace{5mm}
\begin{addmargin}[1cm]{1cm}
\textit{The U.S. has seen a sharp rise in the share of undocumented immigrants over the past several years. Some people believe that undocumented immigrants are more likely to commit violent crimes, while others believe that undocumented immigrants are less likely to commit violent crimes.}

\textit{Texas is the only state that directly compares crime rates for US-born citizens to undocumented immigrants, and provided felony data from 2012-2018. During this time period, the felony violent crime rate was 213 per 100,000 U.S. citizens.}

\textit{This question asks about the felony violent crime rate for undocumented immigrants. Do you think it is more likely that this rate was greater or less than 213 per 100,000?}
\end{addmargin}
\vspace{5mm}

\noindent Republicans and Democrats disagree about the answer to this question, and subjects may expect Republicans to be motivated to believe the crime rate is higher, and Democrats to be motivated to believe the crime rate is lower.\footnote{Indeed, the two parties have differing prior beliefs. Republican receivers' average prior is 56 percent (subject-level clustered s.e. 3.5 pp) and Democratic receivers' average prior is 35 percent (s.e. 3.1 pp).} As such, I code the question as one on which Republicans are motivated to believe ``greater than'' is more likely than ``less than'' to be correct, and Democrats are motivated to believe ``less than'' is more likely than ``greater than'' to be correct.\footnote{If this coding is incorrect, it will work against most of the main treatment effects.} For a full detailing of topics and hypothesized motives, see \Cref{topic-motives}; for the full text of each question, see \Cref{1-question-wordings}.

For each question in the main treatment block (Block 1), the following data are elicited from senders and receivers. Details on incentives are described later.

\begin{enumerate}
\item \textbf{Receivers: Prior beliefs.} Receivers are asked and incentivized to guess the percent chance that the answer to questions like the one above is ``greater than'' or ``less than'' a target number using a scale from 0 to 100. Their reports are restricted to be in multiples of ten: 0 percent, 10 percent, 20 percent ..., 100 percent. They are incentivized to state their true belief.

\item \textbf{Senders: Message choice.} Senders are provided with the correct answer to the question and are asked to choose one of two messages to send to a receiver. The message they can send either says ``the answer is \textbf{greater than} [the target number]'' or ``the answer is \textbf{less than} [the target number].'' Senders make choices via the strategy method, choosing one message \textit{for each possible prior belief that the receiver can report}. That is, senders choose 11 messages for each question. They choose one message if the receiver has a prior of 0 percent, one message if the receiver has a prior of 10 percent, and so on.

\item \textbf{Receivers: Message rating.} Receivers are asked to rate the percent chance that each message was truthful using the strategy method and a scale from 0 to 100. That is, they rate two messages for each question. They are incentivized to state their true belief.
\end{enumerate}
\noindent For senders, the main outcome of interest is whether they choose to send the true message or the false message. For receivers, the main outcome of interest is how they rate the truthfulness of the senders' messages.

There is also a supplementary treatment block (Block 2), which has a similar structure but with two differences. First, receivers are now guessing the percent chance that a given quote from Joe Biden or Donald Trump is ``accurate'' or ``inaccurate'' (and senders choose a message that says ``the statement is \textit{accurate}'' or ``the statement is \textit{inaccurate}''). Second, senders cannot condition on, and do not learn, the receiver's prior. Senders choose one message for each question.

\subsubsection*{Treatments}

Senders and receivers are each randomized into several treatments. The main treatment arms are:
\begin{enumerate}
    \item \textbf{Topic arm:} For both senders and receivers, the topics are varied within subject. They are either political or neutral.
    \item \textbf{Information arm:} For senders, information about the receiver's party is varied within subject. Either senders know the receiver's party or they do not. In addition, senders are randomly matched with Democratic and Republican receivers, so the matched receiver's political motives are randomly either aligned with the true message or aligned with the false message. 
    \item \textbf{Incentives arm:} The senders' incentives are randomly varied between subjects. In the main treatment arms, the sender is incentivized to have their message be rated as truthful by the receiver or the sender is unincentivized. Receivers are honestly told what the sender's incentives are. 
\end{enumerate}
For senders, this is a 2x3x2 design; for receivers, this is a 2x2x2 design. There are two additional treatment arms that are not used for primary analyses (as specified in the pre-analysis plan): (1) There is a competition treatment in which receivers compare the truthfulness of two senders over the course of the experiment, and senders are incentivized to be rated as more truthful than their competitor; (2) There is a receiver-knowledge arm in which the receiver randomly either knows or does not know what party the sender is.

Subjects play the rounds in the order as described in \Cref{tab:timing-r}. Within a block, the questions are presented in a random order.
\renewcommand\arraystretch{1.5}

\begin{table}
\centering
\begin{tabular}{l l l}  
\toprule
\textbf{Round} & \textbf{Receivers' Block} & \textbf{Senders' Block} \\
\midrule
   Pre-randomization  &  Practice, demographics 1 & Practice, demographics 1  \\
   1-7  &  Prior beliefs Block 1 & Send messages Block 1  \\
   1-7  &  News rating Block 1 & ---  \\
   8  &  Attention check & Attention check  \\
   9-12  &  Prior beliefs Block 2 & Send messages Block 2 \\
   9-12  &  News rating Block 2 & ---   \\
   End of experiment  &  Belief elicitations, survey questions & Belief elicitations, survey questions \\
   End of experiment  &  Demographics 2, solutions & Demographics 2, solutions  \\
\bottomrule
\end{tabular}
\begin{scriptsize}
\caption{The timing of treatments in the primary experiment.}
\label{tab:timing-r}
\end{scriptsize}
\end{table}
\renewcommand\arraystretch{1}
Before their role is revealed, senders and receivers give their demographic information, including political party preference.\footnote{Other demographics are age, gender, politics, education, and race. They are also asked to do a three-item cognitive reflection task (modified from \cite{F05}).} Subject select a party or party-lean. Their party is defined as Republican (including leaner) if they state that they are Republican or an independent who leans towards the Republican party, and their party is defined as Democratic (including leaner) otherwise. Next, all subjects play a practice round, playing in both the sender and receiver roles for the given incentives arm they are in. After the practice round, subjects' roles are revealed. The timing of the practice round was chosen to ensure that subjects would be familiar with both roles in the experiment instead of only focusing on their specific role. 

In Block 1 (rounds 1-7), senders can condition on receivers' prior beliefs. In Block 2 (rounds 9-12), senders cannot condition on receivers' prior beliefs. Since, senders may infer receivers' prior beliefs from their party, Block 2 will overestimate the role that party plays relative to priors. Thus, as pre-specified, the main analyses restrict the sample to Block 1.\footnote{As discussed later, main effects are larger when Block 2 is included as well.}

At the end of the experiment, participants' beliefs about the behavior of other players are elicited. This is done first in an incentivized manner (senders predict receivers' ratings and receivers predict senders' truthfulness), and with unincentivized Likert scale questions (both about oneself and about players in the other role). After all sections are completed, correct answers and source links are provided to all participants.

\subsubsection*{Incentives}
\label{1_incentives}

I first describe the mapping of ``points'' into payoffs and then describe how subjects earn points. All subjects are incentivized using a version of the binarized scoring rule (\cite{HO13}).\footnote{This earnings system is a version of the most broadly incentive-compatible one discussed in \textcite{ACH18}. However, as discussed by \textcite{BHL-WP}, it may still lead to some ambiguity-averse hedging.} Throughout the experiment, subjects earn between 0 and 100 points for each incentivized response (e.g. an answer to a question, rating of a message, or choice of a message) they give. At the end of the experiment, ten subjects are randomly selected to receive a bonus payment. If a subject is selected, they either receive \$10 or \$100 depending on their responses. One response is randomly selected to determine payment; the percent chance that they win the bonus is equal to their points earned by this response. When a response is chosen for payment, a sender and receiver are randomly matched for the relevant question, and only the relevant choices in the strategy method are used. 

Receivers' prior beliefs and ratings are incentivized by the quadratic scoring rule. For each question, they report a prior $\pi \in [0,1]$ about the answer. If the answer is ``greater than,'' they earn $100(1-(1-\pi)^2)$ points; if the answer is ``less than,'' they earn $100(1 - \pi^2)$ points. A similar scoring rule is used to incentivize their ratings about their matched senders' messages; receivers who state that a message is truthful with probability $a$ earn $100(1-(1-a)^2)$ points if the message is truthful and $100(1 - a^2)$ points if the message is false.\footnote{Receivers in the competition treatment who state that Sender X is more truthful than Sender Y earn $100(1-(1-a)^2)$ points if Sender X has chosen more truthful messages over the course of the experiment and earn $100(1 - a^2)$ if Sender Y has.} Receivers maximize expected points by stating the closest multiple of 0.1 to their true belief. They are given a table with the points earned as a function of each rating and news type and are told that providing honest ratings is the best way to maximize expected earnings. 

Senders in the incentivized condition are incentivized based on the receivers' ratings; the points they earn for a given choice equals their matched receiver's rating of the percent chance their message was truthful. Senders in the unincentivized condition do not have this round chosen for payment.\footnote{Instead, unincentivized senders are solely paid based on their responses to the beliefs questions.} Subjects in all roles and treatments are incentivized to give accurate answers to the beliefs questions (whose answers are between 0 and 100). They are incentivized using a quadratic scoring rule. If they guess $g$ and the correct answer is $c$, they earn max($0, 100-(c-g)^2$) points.

This design makes substantial use of the strategy method. The strategy method has two clear advantages and two clear disadvantages in this design.\footnote{\textcite{BC11} is a classic paper that discusses differences in behavior between direct elicitation and the strategy method in other games, generally finding that such differences are modest.} The first advantage is that it removes much of the role that other subjects' perceptions play for reasons outside the model. For instance, a sender may otherwise send messages that they would like the receiver to see because they want the receiver to find her entertaining or likeable (e.g. \cite{SG-WP}).\footnote{Relatedly, with this version of the strategy method, subjects will not use others' behavior in previous rounds to predict messages in the current round, limiting learning effects.} The second advantage is statistical power in detecting effects; senders choose eleven messages on each question (instead of one). The data can also estimate, within-subject and within-question, the effect of priors on messages sent. The first disadvantage is that the strategy method does not reflect how people send messages in practice, but it is not clear why this would affect the role of hypothesized mechanisms. The second disadvantage is choice overload, which may lead subjects to attend less to any particular choice of message. Choice overload would likely push the rate of false news closer to 1/2, and would dampen treatment effects, if senders behave more randomly.\footnote{Effects do not noticeably change over the course of the experiment, suggesting that choice overload is not a primary confound.}



\subsection{Data}

750 subjects were recruited from Prolific (\url{prolific.co}) in September 2021 and passed a simple attention check.\footnote{13 subjects (2 percent) were excluded for incorrectly answering the attention check. Results are robust to the inclusion of these subjects.} Prolific is an online platform that was designed by social scientists in order to attain more representative subject samples; it has been shown to perform well relative to other subject pools (\cite{GRW-WP}). The subject pool was restricted to regular social media users: Prolific asks platform users which websites they use ``on a regular basis (at least once a month),'' and the study was only available to Prolific users who say they regularly use Facebook, Twitter, Instagram, Reddit, or LinkedIn. Subjects were recruited so that 375 were Democrats or Independents who lean towards the Democratic Party and the other 375 were Republicans or Independents who lean towards the Republican Party. Subjects were required to have had prior experience on the platform.\footnote{All subjects needed to have completed 100 prior studies and have at least a 90-percent approval rating. Most subjects were required to have registered for Prolific prior to July 2021. These specifications were preregistered; however, due to an unexpectedly slow sign-up rate, 19 percent of subjects were eventually recruited from a larger sample that included registrations after July 2021.}

The subjects were split evenly into 375 senders and 375 receivers. Within each role, subjects were randomly chosen to be in each of the three incentives treatments. Overall, there were 254 subjects (34 percent) in the unincentivized treatment, 248 (33 percent) in the main incentives treatment, and 248 (33 percent) in the competition treatment.

Over the course of the experiment, senders made 30,340 choices of messages.\footnote{This is 99.9 percent of the targeted number of 30,375 choices. The remaining 0.1 percent are instances where the subject did not select an answer.} Most analyses restrict to subjects in the unincentivized treatment (10,107 choices) and the incentivized treatment (10,036 choices). Note that for analyses in this paper, I include all messages \textit{chosen} by senders in the data. That is, I include each strategy-method choice in analyses: \{Send $x$ if R has prior of 0/10, send $x$ if R has prior of 1/10, send $x$ if R has prior of 2/10, ...\}. An alternative approach would be to restrict analyses to messages that correspond to the receiver's actual prior or to focus on specific priors. This approach has the drawback of limiting the number of observations substantially (typically by a factor of 11), but robustness checks indicate that results look similar. 

Receivers gave 4,125 prior beliefs and 8,237 news ratings.\footnote{Priors constituted 100 percent of the targeted number of 4,125. News ratings constituted 99.8 percent of the targeted number of 8,250.} Most analyses focus on the unincentivized treatment (2,832 ratings) and the incentivized treatment (2,726 ratings). As with senders, I include all messages rated by receivers in the data. That is, I include both \{Rate truthfulness of message if it says $x_H$\} and \{Rate truthfulness of message if it says $x_L$\}. I analyze treatment balance for the main incentives treatment for senders (Appendix \Cref{balance-table-s}) and receivers (Appendix \Cref{balance-table-r}). These tables show no evidence that there are significant differences by treatment.

\subsection{Main results}

The effects of the receiver's prior and party on senders are evident from the raw data. \Cref{1-raw-data} shows that both the receiver's party and prior affect incentivized senders' truthfulness in the predicted direction. Consistent with Prediction 2, senders send fewer false messages when the receiver's prior is aligned with the truth and send fewer false messages --- conditional on the receiver's prior --- when the receiver's party is aligned with the truth.

\begin{figure}[htb!]
    \caption{The Effect of the Receiver's Prior and Party on Sending False News}
    \label{1-raw-data}
    \centering
\includegraphics[width=.92\textwidth]{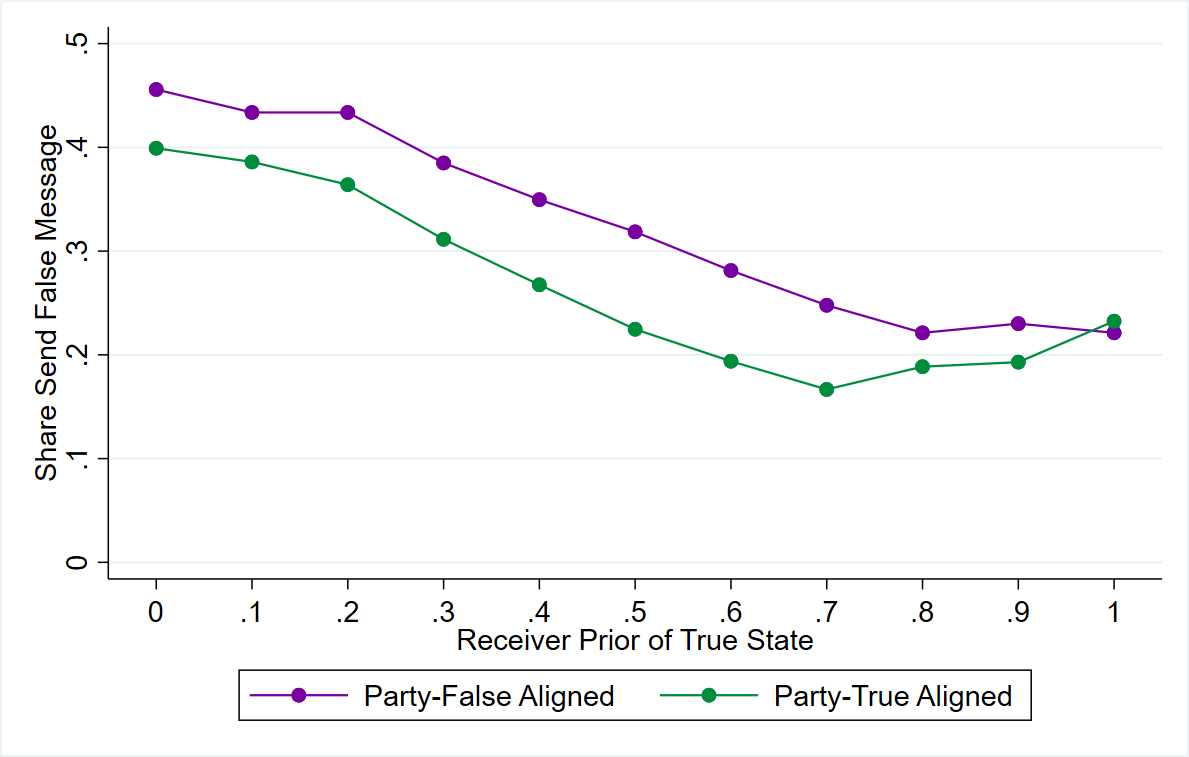}
\begin{threeparttable}
\begin{tablenotes}
\begin{scriptsize}
\vspace{-4mm}
\item \textbf{Notes:} Receiver Prior of True State: the receiver's belief that the true state is correct. Party-True Aligned: indicator for the receiver's party being revealed and aligned with the true message. Party-False Aligned: indicator for the receiver's party being revealed and aligned with the false message. Observations restricted to senders who are incentivized and learn the party of the receiver.
\end{scriptsize}
\end{tablenotes}
\end{threeparttable}
\end{figure}

To causally identify the effect of the receiver's party and prior on the senders' behavior, I run a set of within-subject regressions. The main specification regresses an indicator for sending the false message on an indicator that equals one if the false message is aligned with the receiver's party (Party-False Aligned) for each subject $i$, question topic $q$, and round $r$. I control for the receiver's prior belief and include fixed effects for $i$, $q$, and $r$. 
\begin{equation*}
\text{\textit{SendFalse}}_{iqr} = \alpha + \beta_1 \cdot 1(\text{\textit{PartyFalse}})_{iqr} + \beta_2 \cdot \text{\textit{PriorFalse}}_{iqr} + \nu FE_i + \delta FE_q + \zeta FE_r + \epsilon_{iqr}
\end{equation*}
\Cref{among-incentivized} shows, consistent with \Cref{1-raw-data}, that incentivized senders are more likely to send the false message when good news for the receiver's party is false than in other treatments, controlling for the receiver's prior over the true state.

The specifications differ in the group that Party-False Aligned is compared to. The first column shows that subjects send more false messages when Party-False is aligned than when Party-True is aligned on political topics. The second column shows that subjects send more false messages when Party-False is aligned than when the receiver's party is unknown. The third column suggests that subjects send more false messages when Party-False is aligned than they do on neutral topics.\footnote{Note that since I am comparing political to neutral topics, the last specification does not include question-level fixed effects.}

\begin{center}
\begin{threeparttable}[htbp!]
\begin{footnotesize}
\caption{Factors that lead incentivized subjects to send false messages}
{
\def\sym#1{\ifmmode^{#1}\else\(^{#1}\)\fi}
\begin{tabular}{l*{3}{c}}
\hline\hline
                    &\multicolumn{1}{c}{Vs. Party-True Aligned}         &\multicolumn{1}{c}{Vs. No Info}         &\multicolumn{1}{c}{Vs. Neutral Topics}         \\
\hline
Party-False Aligned &    0.071\sym{**} &    0.072\sym{***}&    0.049\sym{*}  \\
                    &  (0.032)         &  (0.023)         &  (0.029)         \\
Prior-False Aligned &    0.248\sym{***}&    0.229\sym{***}&    0.226\sym{***}\\
                    &  (0.053)         &  (0.059)         &  (0.060)         \\
Question FE         &$\checkmark$         &$\checkmark$         &                  \\
Subject FE          &$\checkmark$         &$\checkmark$         &$\checkmark$         \\
Round FE            &$\checkmark$         &$\checkmark$         &$\checkmark$         \\
Vs. Party-True Aligned\hspace{10mm} &$\checkmark$         &                  &                  \\
Vs. No Info         &                  &$\checkmark$         &                  \\
Vs. Neutral Topics  &                  &                  &$\checkmark$         \\
\hline
Observations        &     4990         &     4705         &     4055         \\
Subjects            &      124         &      123         &      124         \\
Mean                &    0.296         &    0.297         &    0.322         \\
\hline\hline
\multicolumn{4}{l}{\footnotesize Subject-level clustered standard errors in parentheses}\\
\multicolumn{4}{l}{\footnotesize \sym{*} \(p<0.10\), \sym{**} \(p<0.05\), \sym{***} \(p<0.01\)}\\
\end{tabular}
}

\label{among-incentivized}
\end{footnotesize}
\begin{tablenotes}
\begin{scriptsize}
\vspace{-2mm}
\item \textbf{Notes:} OLS, errors clustered at subject level. Dependent variable: indicator for sender choosing the false message. Each column represents a different subset of the data, and Vs. lines indicate the comparison group. Party-False Aligned: indicator for the receiver's party being revealed and aligned with the false message. Prior-False Aligned: the receiver's prior belief that the incorrect answer is true. Party-True Aligned: indicator for the receiver's party being revealed and aligned with the true message. No Info: indicator for the receiver's party not being revealed. By chance, one subject happened to always learn the receiver's party. Only includes questions in Block 1.
\end{scriptsize}
\vspace{5mm}
\end{tablenotes}
\end{threeparttable}
\end{center}

\Cref{among-incentivized} also shows that incentivized senders are more likely to send the false message when receivers' priors are false. The first column shows that the treatment effect of having the receiver's party aligned with the false state rather than the true state is equivalent to the treatment effect of the receiver's prior changing by 0.071 / 0.248 = 29 pp.

The effect that senders expect receivers' party to have on inference is similar when beliefs are explicitly elicited. Senders are asked to predict the average rating of receivers whose prior is 1/2 when they receive good news for their party or bad news for their party. Incentivized senders estimate that receivers will have an average gap of 30 pp (s.e. 3 pp).

Next, we turn to the effect that incentives have on the truthfulness of messages. \Cref{incentives-effect} compares behavior of the incentivized senders to behavior of the unincentivized senders using a between-subject specification. This specification regresses $\text{\textit{SendFalse}}_{iqr}$ on an indicator for the incentives treatment and a set of demographic and treatment controls ($X_i$) as well as fixed effects for $q$ and $r$.\footnote{Controls include gender, race age, own party, education, score on a cognitive reflection task (CRT), and whether the receiver knows the sender's party.}
\begin{equation*}
\text{\textit{SendFalse}}_{iqr} = \alpha + \beta \cdot 1(\text{\textit{Incentivized}})_{i} + \eta X_i + \delta FE_q + \zeta FE_r + \epsilon_{iqr}.
\end{equation*}

Consistent with Prediction 1, \Cref{incentives-effect} shows that there is a negative effect of incentives on message truthfulness. In particular, the first column of \Cref{incentives-effect} shows that senders are more likely to send the false message when they are incentivized to be perceived as truthful. As shown by the second column, this treatment effect is more pronounced when constrained to the questions on which receivers' party is misaligned with the truth. The final column shows that incentivized senders send more false messages when receivers' party is misaligned with the truth (as in \Cref{among-incentivized}), but that unincentivized senders are not directionally affected by this condition. Next, we explore this unincentivized group further.

\begin{center}
\begin{threeparttable}[htbp!]
\begin{footnotesize}
\caption{The effect of incentives on sending false messages}
{
\def\sym#1{\ifmmode^{#1}\else\(^{#1}\)\fi}
\begin{tabular}{l*{3}{c}}
\hline\hline
                    &\multicolumn{1}{c}{All Pol. Questions}        &\multicolumn{1}{c}{Party-False Aligned}         &\multicolumn{1}{c}{Interaction}         \\
\hline
Incentivized        &       0.073\sym{***}&    0.107\sym{***}&    0.051\sym{*}  \\
                    &  (0.026)           &  (0.037)          &  (0.028)         \\
Party-False Aligned x Incentivized  &            &                   &    0.062\sym{**} \\
     &                  &                    &  (0.026)         \\
Party-False Aligned x Unincentivized &                  &                &   -0.003         \\
   &                  &                  &   (0.027)         \\
Question FE         &$\checkmark$           &$\checkmark$         &$\checkmark$         \\
Round FE            &$\checkmark$           &$\checkmark$         &$\checkmark$         \\
Subject controls   &$\checkmark$     &$\checkmark$       &$\checkmark$         \\
All Questions       &$\checkmark$         &                &$\checkmark$         \\
Only Party-False Aligned &                   &$\checkmark$         &                  \\
\hline
Observations        &    14248         &      4702        &    14248         \\
Subjects            &      249              &      220          &      249         \\
Mean                &    0.250         &    0.272       &    0.250         \\
\hline\hline
\multicolumn{4}{l}{\footnotesize Subject-level standard errors in parentheses}\\
\multicolumn{4}{l}{\footnotesize \sym{*} \(p<0.10\), \sym{**} \(p<0.05\), \sym{***} \(p<0.01\)}\\
\end{tabular}
}

\label{incentives-effect}
\end{footnotesize}
\begin{tablenotes}
\begin{scriptsize}
\vspace{-2mm}
\item \textbf{Notes:} OLS, errors clustered at subject level. Dependent variable: indicator for sender choosing the false message. Subject controls: Gender, race, age, own party, education, CRT score, and whether the receiver knows the sender's party. Only includes observations in Block 1.
\end{scriptsize}
\vspace{5mm}
\end{tablenotes}
\end{threeparttable}
\end{center}



\subsection{Unincentivized senders and own-party effects}

An alternative explanation for the results among incentivized senders is that senders have an innate preference to be rated as truthful. Another concern is that the way the experiment is designed pushes senders to condition on receivers' party. For instance, perhaps the highly politicized topics, the varying of receivers' party and priors within subject, and the salience of the text may all push senders towards the observed treatment effects. 

However, these issues arise across all the incentive treatments, and treatment effects are entirely driven by the explicit incentives. In particular, the main treatments have little effect on unincentivized subjects, as shown in Appendix \Cref{among-unincentivized}, which replicates \Cref{among-incentivized} among the unincentivized senders. Consistent with Prediction 3, there is no statistically significant effect of receivers' party in any specification, and estimates are close to zero. The effect of receivers' priors is positive, but significantly lower than the effect for the incentivized senders. This result demonstrates the important role that incentives play; unincentivized senders do not seem to inherently value aligning their messages with the receiver's motivated beliefs.

\begin{center}
\begin{threeparttable}[htbp!]
\begin{footnotesize}
\caption{Factors that lead unincentivized senders to send false messages}
{
\def\sym#1{\ifmmode^{#1}\else\(^{#1}\)\fi}
\begin{tabular}{l*{3}{c}}
\hline\hline
                    &\multicolumn{1}{c}{Vs. Party-True Aligned}         &\multicolumn{1}{c}{Vs. No Info}         &\multicolumn{1}{c}{Vs. Neutral Topics}         \\
\hline
Party-False Aligned &    0.005         &   -0.038         &   -0.011         \\
                    &  (0.032)         &  (0.030)         &  (0.024)         \\
Prior-False Aligned &    0.063         &    0.074\sym{*}  &    0.071         \\
                    &  (0.039)         &  (0.040)         &  (0.045)         \\
Question FE         &$\checkmark$         &$\checkmark$         &                  \\
Subject FE          &$\checkmark$         &$\checkmark$         &$\checkmark$         \\
Round FE            &$\checkmark$         &$\checkmark$         &$\checkmark$         \\
Vs. Party-True Aligned\hspace{10mm} &$\checkmark$         &                  &                  \\
Vs. No Info         &                  &$\checkmark$         &                  \\
Vs. Neutral Topics  &                  &                  &$\checkmark$         \\
\hline
Observations        &     4636         &     4619         &     3920         \\
Subjects            &      125         &      125         &      125         \\
Mean                &    0.211         &    0.215         &    0.203         \\
\hline\hline
\multicolumn{4}{l}{\footnotesize Subject-level clustered standard errors in parentheses}\\
\multicolumn{4}{l}{\footnotesize \sym{*} \(p<0.10\), \sym{**} \(p<0.05\), \sym{***} \(p<0.01\)}\\
\end{tabular}
}

\label{among-unincentivized}
\end{footnotesize}
\begin{tablenotes}
\begin{scriptsize}
\item \textbf{Notes:} OLS, errors clustered at subject level. Dependent variable: indicator for sender choosing the false message. Each column represents a different subset of the data, and Vs. lines indicate the comparison group. Party-False Aligned: indicator for the receiver's party being revealed and aligned with the false message. Prior-False Aligned: the receiver's prior belief that the incorrect answer is true. Party-True Aligned: indicator for the receiver's party being revealed and aligned with the true message. No Info: indicator for the receiver's party not being revealed. 
\end{scriptsize}
\end{tablenotes}
\end{threeparttable}
\end{center}

However, there is still a non-negligible share of unincentivized senders who choose false messages, at 21 percent. Part of this 21 percent could be due to randomness: for instance, sometimes senders click randomly, do not read the answer correctly, or misclick. In addition, even pure coordination games may lead to communication problems if players disagree about the meaning of messages (e.g. \cite{FR96}). Finally, some senders may send false information because of expressive preferences: they prefer to send news that aligns with their own party, even when it is false. \Cref{1-ownparty} provides supporting evidence for this expressive-preferences mechanism, showing that unincentivized senders (column 2) send false messages significantly more often when it aligns with their own party.\footnote{Meanwhile, incentivized senders (column 1) are not statistically significantly affected by their own party, and are (as described above) instead affected by the other subject's party. However, these estimates are not precisely measured and should not be overinterpreted.} In the raw data, unincentivized senders are 38 percent more likely to send false messages when the truth is misaligned with their party than when it is aligned.

\begin{center}
\begin{threeparttable}[htbp!]
\begin{footnotesize}
\caption{The effect of senders' own party on messages}
{
\def\sym#1{\ifmmode^{#1}\else\(^{#1}\)\fi}
\begin{tabular}{l*{3}{c}}
\hline\hline
                    &\multicolumn{1}{c}{Incentivized}         &\multicolumn{1}{c}{Unincentivized}         &\multicolumn{1}{c}{Interaction}         \\
\hline
Own Party-False Aligned    &    0.010         &    0.053\sym{**} &    0.054\sym{**} \\
             &  (0.025)         &  (0.021)         &  (0.021)         \\
Other's Party-False Aligned &    0.092\sym{***}&    0.005         &    0.008         \\
             &  (0.028)         &  (0.028)         &  (0.027)         \\
Own Party-False Aligned  x Incentivized    &                  &                  &   -0.039         \\
 &                  &                  &  (0.033)         \\
Other's Party-False Aligned x Incentivized &                  &                  &    0.088\sym{**} \\
 &                  &                  &  (0.040)         \\
Question FE         &$\checkmark$         &$\checkmark$         &$\checkmark$         \\
Subject FE          &$\checkmark$         &$\checkmark$         &$\checkmark$         \\
Round FE            &$\checkmark$         &$\checkmark$         &$\checkmark$         \\
Incentivized subjects &$\checkmark$         &                  &$\checkmark$         \\
Unincentivized subjects\hspace{30mm} &                  &$\checkmark$         &$\checkmark$         \\
\hline
Observations        &     5486         &     5136         &    10622         \\
Subjects            &      124         &      125         &      249         \\
Mean                &    0.296         &    0.212         &    0.255         \\
\hline\hline
\multicolumn{4}{l}{\footnotesize Standard errors in parentheses}\\
\multicolumn{4}{l}{\footnotesize \sym{*} \(p<0.10\), \sym{**} \(p<0.05\), \sym{***} \(p<0.01\)}\\
\end{tabular}
}

\label{1-ownparty}
\end{footnotesize}
\begin{tablenotes}
\begin{scriptsize}
\vspace{-2mm}
\item \textbf{Notes:} OLS, errors clustered at subject level. Dependent variable: indicator for sender choosing the false message. Own Party-False Aligned: indicator for the sender's party being aligned with the false message. Other's Party-False Aligned: indicator for the receiver's party being revealed and aligned with the false message. 
\vspace{5mm}
\end{scriptsize}
\end{tablenotes}
\end{threeparttable}
\end{center}

In other exploratory analyses, I also find evidence that false news effects are stronger when senders and receivers are of the same party. On political questions in the unincentivized and incentivized groups, senders send false news to their own party 23.0 percent of the time (s.e. 1.7 pp) and false news to the opposing party 28.0 percent of the time (s.e. 1.7 pp); that is, they send 5.0 pp more true messages to their own party (s.e. 1.9 pp, $p=0.009$). Similarly, receivers rate political messages that come from their own party as being 6.0 pp more likely to be truthful (s.e. 0.8 pp, $p<0.001$). Such a finding is consistent with participants having in-group pro-social biases (such as in \cite{CL09}).\footnote{Since who the sender is matched with is randomized within subject, previous results do not qualitatively change when controlling for whether the two players are of the same party.}

\subsection{Receivers' behavior and higher-order beliefs}

Next, we turn to receivers' actual ratings of these messages and discuss both players' higher-order beliefs. Consistent with Prediction 2, receivers who are matched with incentivized senders rate pro-party news to be true with probability 60.3 percent (subject-level clustered s.e. 1.4 pp) and anti-party news to be true with probability 48.8 percent (s.e. 1.5 pp). This gap is 11.4 pp (s.e. 2.3 pp) and statistically significant ($p < 0.001$).\footnote{These numbers restrict to questions in Block 1, where senders can condition on receivers' priors.}

It is worth noting that this gap does not necessarily rule out Bayesian updating since there are differences in receivers' priors and receivers may have different beliefs about senders' behavior. Across all treatments, receivers believe that the pro-party state is true with probability 58.7 percent (s.e. 0.6 pp), which is statistically significantly larger than 50 percent ($p<0.001$).\footnote{Appendix \Cref{1-priors} shows that prior beliefs favor the pro-party state on essentially every topic. Once priors are controlled for, the gap in ratings is positive but significantly reduced.}
In addition, the sender-receiver setting may confound identification of motivated reasoning from strategic considerations. Because of these confounds, \Cref{experiment2} uses a different experimental strategy that is able to more cleanly identify receivers' motivated reasoning from Bayesian updating.

Next, we consider receivers' beliefs about senders' strategies. Perhaps surprisingly, receivers' ratings are not directionally affected by senders' incentives. This result occurs despite the fact that receivers have played in the role of sender in a practice round and are fully informed about their sender's incentives in every round. \Cref{s-r-incentives} shows the treatment effect on senders and receivers graphically. Receivers who are matched with incentivized (versus unincentivized) senders give ratings that are 0.3 pp (s.e. 1.2 pp) higher, even though incentivized senders are 7.3 pp \textit{less} likely to be truthful. In the condition where treatment effects are largest (when the truth or message is misaligned with the receivers' politics), senders' incentives lead receivers to decrease ratings by 1.1 pp (s.e. 1.8 pp), while they lead senders to be 10.7 pp less likely to tell the truth. Each of these differences are statistically significant, and sizeable in magnitude. 

\begin{figure}[ht]
    \caption{Senders' Incentives Affect Messages Sent but Not How They Are Rated}
    \label{s-r-incentives}
    \centering
\includegraphics[width=.8\textwidth]{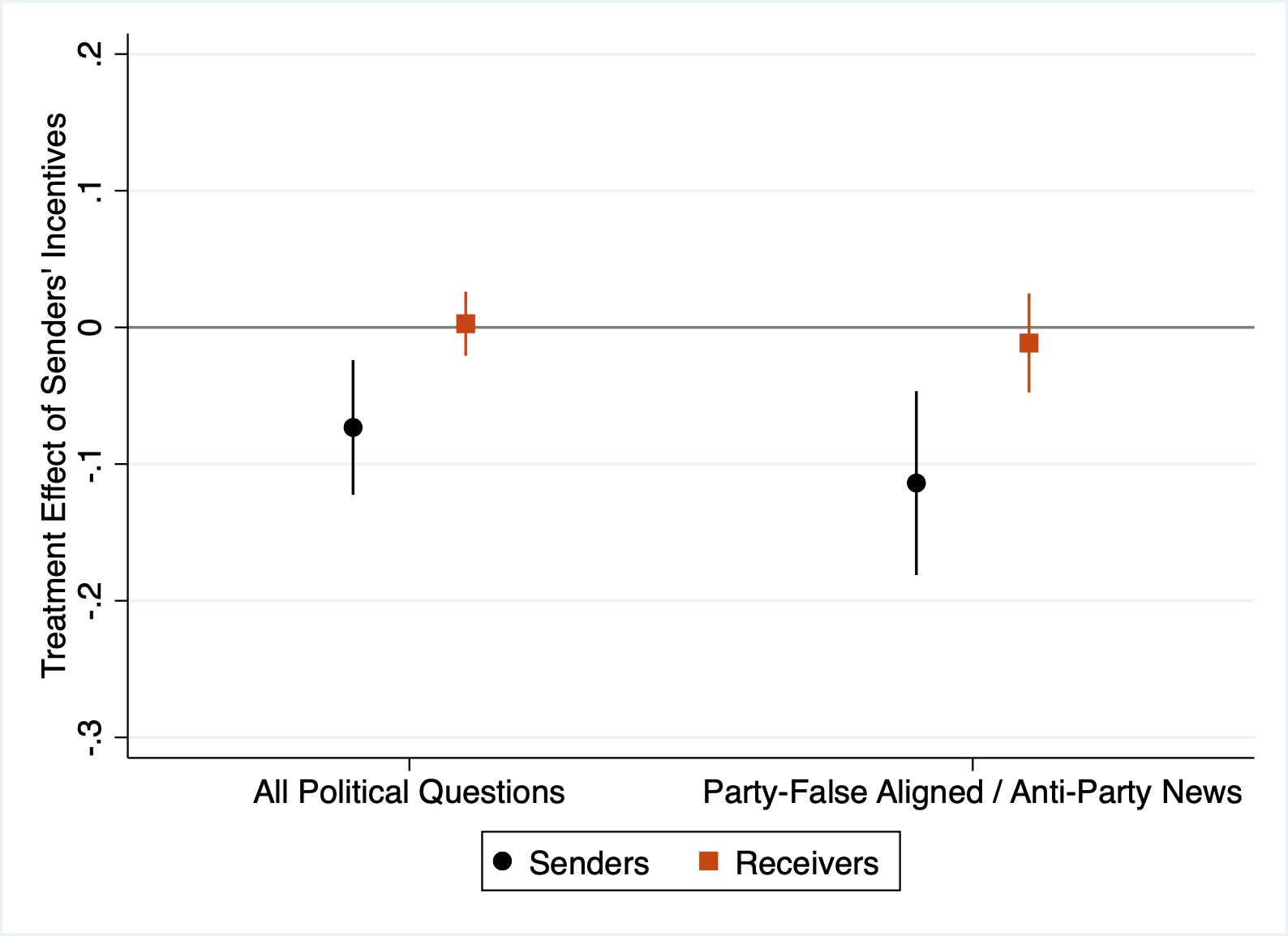}
\begin{threeparttable}
\begin{tablenotes}
\begin{scriptsize}
\vspace{-4mm}
\item \textbf{Notes:} OLS regression coefficients, errors clustered at subject level. Coefficients are from a regression of message truthfulness (for senders) or ratings about message truthfulness (for receivers) on being in the S-incentivized treatment. Controls for age, race, gender, education, CRT score, own party, and fixed effects for question and round number are included. Also included is a control for whether the sender (receiver) knows the party of the receiver (sender) they are matched with. This figure shows that senders choose more false messages when incentivized, while receivers do not anticipate more false messages when the sender is incentivized. Error bars correspond to 95 percent confidence intervals. 
\end{scriptsize}
\end{tablenotes}
\end{threeparttable}
\end{figure}

These results are consistent with receiver naivete with respect to higher-order beliefs. This evidence is consistent with the motivated equilibrium of the model in which receivers believe their motivated reasoning parameter $\lambda$ is low, expect senders to believe $\lambda$ is low, and so on, but senders believe that $\lambda$ is high. We saw in \Cref{theory} that a motivated equilibrium exists in which raising the sender's incentives from low to moderate would lead her to be less truthful but not affect the receiver's play.

An alternative explanation is that receivers believe that incentivized senders are less truthful, but are being altruistic towards them by giving higher ratings. As such, I use exploratory analyses on survey questions to study what senders and receivers believe determines the behavior of each player in the game. Appendix \Cref{effects-survey} shows that incentivized senders report that they rely less on the truth and more on the party of the receiver, as compared to unincentivized senders. However, consistent with naivete, receivers do not state significantly different reports when they have been faced with incentivized or unincentivized senders. Neither senders nor receivers respond significantly differently about receivers' behavior when they are incentivized versus unincentivized, though the effects are noisy, weakly suggesting that both players believe receivers will not respond much to senders' incentives.

An alternative explanation for the null effect of senders' incentives on receivers is that receivers are simply ignoring information they see about senders. Contrary to this hypothesis, receivers do take into account \textit{other} information they have about senders, namely their party. As mentioned above, receivers give significantly higher ratings to senders of their own party than to senders of the opposing party. Instead, it seems that receivers are specifically insensitive to the \textit{incentives} that senders face. 

Next, I find systematic differences between receivers' behavior and senders' beliefs about receivers' behavior, as senders overstate the relative role of politics versus priors in receivers' inference. On average, senders are asked to state their beliefs about Republican and Democratic receivers' ratings of pro-Republican and pro-Democratic messages when the receivers have a prior of 1/2. Senders believe that the gap between the pro-party and the anti-party ratings will be 30 pp. These beliefs are predictive of the treatment effects of incentives, as shown in Appendix \Cref{1-incentives-beliefs}. 

Senders treat the effect of the receiver's party as being similar to an effect of a change of 29 pp in the receiver's prior, but their beliefs are substantially larger than the true effect of party on receivers. While it is not possible to determine what the optimal \textit{level} of truth-telling is for a sender, these results suggest that, given receivers' behavior, senders would be better off by being \textit{relatively} more sensitive to priors and less sensitive to politics.

In the context of the model, measured higher-order beliefs are consistent with the model in which $\hat{\lambda}_S > \lambda \approx \hat{\lambda}_R$. That is, receivers may engage in motivated reasoning to a small extent, and they are aware of this, but senders overstate the bias. When senders' incentives go from $\gamma=0$ to an intermediate level $\gamma > 0$, receivers do not realize that senders overstate the bias, leading to the form of motivated equilibrium described in \Cref{theory} in which this increase in incentives affect senders but not receivers.

\subsection{Discussion and Robustness}

In this section, I further explore a number of additional drivers of the experimental results. First, I consider whether results are driven by senders' capability to condition on receivers' priors. Next, I look at individual-level distributions. Then, additional types of incentives are discussed. Robustness exercises are conducted to argue that the set of results seen are not likely to be explained by senders believing that receivers systematically misreport their priors, the details of the strategy method specification, or social signaling. I argue that some forms of higher-order beliefs are more challenging to rule out, however, so an additional experiment is helpful. 

The previous results have focused on Block 1, where senders can condition on receivers' prior beliefs. One concern is that senders who are asked to answer via the strategy method focus more on differences in receiver features. However, in Block 2, senders can no longer condition on receivers' priors, and results look similar. Appendix \Cref{by-block-by-inc} compares the main effects by incentives and by block. For incentivized senders, the effect of receivers' party on senders' messages is larger, indicating that senders want to condition on receivers' prior beliefs, and infer their priors from their party. For unincentivized senders, there is again a small and statistically insignificant effect of receivers' party on senders' messages. When the data from Block 1 and Block 2 are combined, we see similar results. Appendix \Cref{among-incentivized-noprior} repeats the analysis in \Cref{among-incentivized} for incentivized senders, but includes data from both Block 1 and Block 2, and does not control for the receiver's prior. Appendix \Cref{among-unincentivized-noprior} repeats the analysis in Appendix \Cref{among-incentivized} for unincentivized senders. In each of these cases, the main results look similar. 

We have so far seen aggregate-level results, so we next explore differences at the individual level. \Cref{1-indiv-data} plots these data with CDFs by incentivized condition and by the alignment of the receiver's party. \Cref{1-indiv-data} shows that incentive effects are primarily driven by the more-truthful part of the sender distribution: 43 percent of unincentivized senders never send a false message, while only 21 percent of incentivized senders are always truthful. Meanwhile, the effect that the receiver's party has on senders is more evenly dispersed. One potential explanation for this difference is that some senders do not think that receivers' motivated beliefs are the ones hypothesized in \Cref{topic-motives}.

There are no distinguishable effects of the receiver's party on unincentivized senders at the individual level. Appendix \Cref{1-indiv-data-unin} plots the same CDF as the one in the second panel of \Cref{1-indiv-data} but restricts to observations from unincentivized senders. In Appendix \Cref{1-indiv-data-unin}, the CDFs lie on top of each other, indicating that there are little distributional differences for unincentivized senders by the receiver's party-truth alignment. In addition, the median share of false messages in each condition is zero, indicating that the majority of unincentivized senders never or rarely send false messages.

We can use other types of incentives to better unpack what drives senders' behavior. The theory posits that senders choose messages to send in order for those particular messages to be rated as truthful. An alternative theory could have senders differ in their fundamental propensity to tell the truth, and senders then choose messages in order to signal that they are truthful \textit{in general}. In order to test between these theories, I consider an additional treatment with \textit{competition} incentives, in which receivers rate which of two senders they believe was more truthful over the course of the experiment. Each sender is incentivized to be rated as more truthful overall than her competitor.\footnote{In particular, receivers predict the likelihood that each sender is more truthful if the two senders choose the same message and if the two senders choose opposite messages. Mirroring the main incentives treatment, senders earn points equal to the receiver's rating of the percent chance they are more truthful.} I find that these incentives do not have a significant impact on the overall truthfulness of messages (Appendix \Cref{1-competition-effect}), and the receiver's party plays less significant of a role, though priors still have some affect (Appendix \Cref{among-competition}). While the null effects are noisy, these results suggest that incentives for having messages appear as truthful are more impactful than incentives for appearing truthful overall. 

Senders may believe that receivers systematically misreport their priors or report their priors with noise. If senders believe that receivers' priors are accentuated more in ways that favor their party, the effects of prior may be overestimated and the effects of party may be understated for senders. On the other hand, senders may expect a Democratic and Republican receiver who each give the same prior to have different beliefs because the prior is rounded to the nearest 0.1. This would affect estimates in the direction that overestimates the party effect. However, the magnitude of the party effect is large enough that, to fully account for these effects, priors would need to be biased by 15 pp.\footnote{In particular, I run the main specification from \Cref{incentives-effect}, but replace the prior $\pi$ with $\min\{\pi+x,1\}$ for party-true aligned and with $\max\{\pi-x,0\}$ for party-false aligned beliefs. The coefficient of the party effect necessarily decreases in $x$, and crosses zero at $x = 0.15$.} Similarly, if senders believe that receivers' priors are stated with noise, then they may infer something about the receiver from his party. In either case, to fully explain these results, senders would need to expect, on average, that a Democratic and Republican receiver who each state the same $\hat{\pi}$ actually hold priors $\pi_D$ and $\pi_R$ that differ by 29 pp. Senders may also believe that receivers' ratings do not reflect their true beliefs, but rather a form of expressive preferences. In this case, the results can be interpreted as saying that senders cater to their beliefs about receivers' expressive beliefs. 

Results are not driven by which control variables are included or by the strategy-method design that senders face. I show in Appendix \Cref{1-alt-specs} that the main treatment effects do not significantly move when alternative specifications are used. For instance, when specifications include no control variables, or additional sender-level control variables, the point estimates hardly change. This latter specification includes controls for whether the sender's party is aligned with the truth and whether the sender's party is aligned with the receiver's party, both of which affect senders. Next, another approach to averaging across all of senders' choices is to weight choices based on the likelihood that they would be matched with a receiver.\footnote{For instance, a sender in a given treatment who is twice as likely to be matched with a receiver with prior 7/10 than a receiver with prior 3/10 has observations that are weighted twice as high for their choice conditional on a 7/10 receiver prior than their choice conditional on a 3/10 receiver prior.} In \Cref{1-alt-specs}, the second set of spikes shows no notable difference when such a weighting is used. In addition, I test whether restricting to senders' choices when faced with a receiver who has a prior of 1/2 affects the estimates. In each of these alternate specifications, point estimates and statistical significance are similar to those in the main specifications.

Without incentives, senders do not condition on receivers' party, but it is still possible that senders care about what others think of them and socially signal, as has been shown in other domains (e.g. \cite{DLMR17}). The effects are ambiguous in this experiment, as implicit reputation incentives could either serve to increase truthtelling overall (to signal honesty to the experimenter) or increase sending particular falsehoods. Importantly, if social signaling were driving these patterns, their effects would have to be present both when senders face incentives (where they incorporate receivers' party to send false information), and when they do not face incentives (where they do not incorporate receivers' party). 

Lastly, while the main results provide evidence that senders use the direction of receivers' political beliefs in determining what messages to send them, it is difficult to isolate what form of motivated reasoning is the root cause. For instance, senders may believe that receivers are more \textit{accurate} at rating senders' messages when they are aligned with the receivers' party. If senders have other-regarding preferences, caring about the accuracy of receivers' ratings, senders would distort in the direction of the receivers' party. (However, this cannot fully explain why incentives have an effect on senders.) In addition, there may be higher-order belief explanations for these results; for instance, senders may believe that receivers believe that senders distort how they send messages towards being misaligned with the receivers' party. Such a mechanism is similar to the one discussed in \textcite{M01}.

Therefore, in order to isolate senders' beliefs about the particular bias that receivers have in processing information, the additional experiment complements the above analyses by (1) identifying receivers' motivated reasoning using ratings of messages sent by a computer and (2) studying the behavior of senders choosing which computer messages to be paid for.

\section{Additional Experiment}
\label{experiment2}

\subsection{Design}
The additional experiment is designed to isolate the role that beliefs about others' motivated reasoning play in choosing messages. Receivers are asked to rate the truthfulness of messages from a computer in a task that is designed to elicit their motivated reasoning. Senders decide which of the two computers' messages they want to choose; the messages affect their payoff through the receivers' ratings, but do not affect what the receivers see or how many points they score. That is, senders are essentially betting on messages. Breaking the interaction between senders and receivers allows for identification of motivated reasoning among receivers, which in turn enables us to disentangle the strategic effects present in \Cref{experiment1} from the motivated-reasoning effects of senders' choice of news. The game is described below. Screenshots for the pages subjects see are in \Cref{screenshots2}. 

Receivers play a simplified version of the game in \textcite{T-WPa}. They play the game in two steps: 
\begin{enumerate}
\item \textbf{Median belief:} Receivers are given questions like the one in \Cref{expt1-design}. The only difference is that, instead of being asked for the probability that the answer is greater or less than a particular number, they are asked to guess what the number actually is. I elicit their median beliefs; receivers report a guess $\mu$ such that they believe that the answer is equally likely to be greater than or less than $\mu$. (Details on incentives are below.)
\item \textbf{News rating:} After stating their median beliefs, receivers are given two messages from a computer. One message says ``the answer is \textbf{greater than} [median]'' and the other message says ``the answer is \textbf{less than} [median].'' One of these messages is true and one is false, and receivers do not know which one is true. They are asked whether they believe that the first message is more likely to be true, the second message is more likely to be true, or believe that both messages are equally likely to be true.\footnote{This setup has two differences from the setup in \textcite{T-WPa}, both of which serve to simplify the environment. First, in this experiment I elicit beliefs about the relative likelihood of the two messages instead of beliefs about only one of the two messages. Second, in this experiment I only ask receivers to choose which message is more truthful instead of the probability that a given message is truthful.}
\end{enumerate}

\noindent Senders are asked to choose one of the two computer messages on each question. Specifically, on each question senders are matched with two receivers, one Democrat and one Republican, who have the same median belief. Senders are told the true answer and then asked to choose either: ``the answer is \textbf{greater than} [receiver's median]'' or ``the answer is \textbf{less than} [receiver's median].''

\subsubsection*{Treatments}
The main treatment arms for senders are similar to those in the primary experiment. 
\begin{enumerate}
    \item \textbf{Topic arm:} The topics are varied within subject. They are either political or neutral.
    \item \textbf{Information arm:} Information about the receiver's party is varied within subject. Either senders know the receiver's party or they do not. In addition, senders are randomly matched with Democratic and Republican receivers (with the same median belief), so the matched receiver's political motives are randomly either aligned with the true message or aligned with the false message.
    \item \textbf{Incentives arm:} The senders' incentives are randomly varied between subjects. The sender is either incentivized to choose a computer message that is rated as truthful by the receiver or the sender is unincentivized.
\end{enumerate}
This is a 2x3x2 design. Senders play the rounds in the order as described in \Cref{tab:timing-s2}, and see questions in a random order.

\renewcommand\arraystretch{1.5}

\begin{table}[ht]
\centering
\begin{tabular}{l l}  
\toprule
\textbf{Round} & \textbf{Senders' Block}  \\
\midrule
   1  &  Sample question in role of Receiver  \\
   2-7  &  Choose messages  \\
   8-11  &  Demand for information \\
   12  &  Attention check \\
   End of experiment &  Belief elicitation  \\
   End of experiment &  Demographics and solutions \\
\bottomrule
\end{tabular}
\begin{scriptsize}
\caption{The timing of treatments for senders in the additional experiment.}
\label{tab:timing-s2}
\end{scriptsize}
\end{table}
\renewcommand\arraystretch{1}

After senders play six rounds in which they choose messages, they play four rounds in which they choose whether to ``purchase information'' about the receiver by conditioning their message on the receiver's political party. In these rounds, senders see the receiver's question and are asked to choose one of the following two options: (1) Be able to condition their message choice on the receivers' party with probability 1/2 and receive a slightly-higher payoff, or (2) Be able to condition their message choice on the receivers' party with probability 1 and receive a slightly-lower payoff. They are asked the information-purchasing questions after the main treatment block in order to enable them to have a chance to determine how much they value party information.

In two of the four rounds, senders can condition on the receiver's party regardless of their choice; in the other two rounds, they can condition on the receiver's party if and only if they have purchased the information.



\subsubsection*{Incentives}

Subjects earn points in the experiment; points translate into earnings using the same binarized scoring rule as in \Cref{experiment1}. The only difference is that in this experiment, five subjects win a bonus instead of ten.

Receivers' median beliefs are incentivized by a linear scoring rule. For each question, they give a guess $\mu$ about the answer. They earn $\max\{100 - |\mu - \text{answer}|, 0\}$ for their guess. Receivers' news ratings are incentivized by a simple concave scoring rule. Receivers who guess that Message X is more likely to be true than Message Y earn 100 points if they are correct and 0 points if they are incorrect; receivers who guess that both messages are equally likely to be true earn 55 points regardless of the true answer. They maximize points by guessing that X is true iff they believe P(X true) $\geq 0.55$, by guessing that Y is true iff they believe P(X true) $\leq 0.45$, and by guessing that they are equally likely iff P(X true) $\in [0.45,0.55]$. A Bayesian, whose belief remains at 1/2, would guess that they are equally likely. Systematic differences in news source ratings are attributed here to motivated reasoning.

Senders in the incentivized condition are incentivized to choose the one of the two messages that the receiver was more likely to think is true. They earn 100 points if the receiver guesses their message is true, 50 points if the receiver guesses both messages are equally likely, and 0 points if the receiver guesses the other message is true. All senders are given these incentives in the demand-for-information rounds. Subjects in both treatments are incentivized to give accurate answers to the beliefs questions (whose answers are between 0 and 100). They are incentivized using a quadratic scoring rule. If they guess $g$ and the correct answer is $c$, they earn max($0, 100-(c-g)^2$) points.

\subsubsection*{Comparing the two experimental designs}

The main difference in the experiments is that, while in the primary experiment the sender and receiver both affect each others' payoffs, in the additional experiment the sender does not impact the receiver. The primary experiment is able to identify the role of higher-order beliefs and receivers' beliefs about senders more broadly. There are a few additional differences. While the primary experiment uses fixed target numbers, the additional experiment uses median beliefs. As such, while the primary experiment elicits senders' beliefs about receivers' belief updating, the additional experiment elicits senders' beliefs about receivers' motivated reasoning. Only the primary experiment studies receivers' beliefs about senders. More subtly, the primary experiment varies, within topic, the effect of answers being too Dem or too Rep, while the additional experiment only does this between topics. Lastly, only the additional experiment looks at senders' demand for information about receivers.

Given the relative contributions of the two studies, the emphasis in the analysis of the additional experiment is on senders' behavior. For a deeper discussion on receivers in a similar context, and what the requirements are for this design to be able to identify motivated reasoning, see \textcite{T-WPa}, which uses an expanded version of this design with a sample of approximately 1,000 receivers. What is important to know for this paper is that, assuming that receivers report their true median beliefs, a Bayesian would always say that the two messages were equally likely to be true.

\subsection{Data}

550 subjects were recruited from Prolific (\url{prolific.co}) in May-June 2021 and passed a simple attention check. The subject pool included the general United States population, and the pool was restricted to subjects who had prior experience on the platform.\footnote{All subjects needed to have completed 100 prior studies and have at least a 90-percent approval rating.} To emphasize that the focus is on sender behavior, the sample consists of 500 senders and 50 receivers. Receivers participated first. After receivers took the experiment, on each question I chose a median belief that both a Democratic receiver and a Republican receiver stated.\footnote{On each question, such a median belief existed. I would not have included a question if there was no median that both a Democratic and a Republican receiver had stated.} Half of median beliefs were too far in the Democratic direction, and the other half were too far in the Republican direction. These were the median beliefs that were presented to senders. \Cref{screenshots2} includes the list of topics, median beliefs, and which messages would be truthful. 

48 receivers (96 percent) and 492 senders (98 percent) stated a preference or lean for one party versus the other.\footnote{Unlike in the primary experiment, the party question included an option to select ``Independent (do not lean towards either party).''} Of the receivers, 26 (52 percent) were Republicans and 22 (44 percent) were Democrats. Of the senders, 244 (49 percent) were Republicans and 248 (50 percent) were Democrats. Analyses are restricted to receivers with a party preference; these receivers made 383 news ratings on political topics. In the total sample of senders, there were 2,999 messages chosen.\footnote{The targeted sample was 3,000; one subject did not choose a message on one question.} Senders are split into the incentives treatment and the unincentivized treatment, with 245 (49 percent) being incentivized. In Appendix \Cref{balance-table-s2}, I show the balance table for the incentives treatment among senders. I find modest political differences, but no other substantial differences, by treatment. 

\subsection{Main results}

Senders in the additional experiment choose computer messages in a similar manner to how senders in the primary experiment choose messages to send to receivers. The nearly-exact replication of these results suggests that beliefs about motivated reasoning are an important determinant in understanding the results in the full sender-receiver game.

\Cref{2-among-incentivized} tests the impact of the receiver's party on the truthfulness of senders' choices. The specification is identical to the specification in Appendix \Cref{among-incentivized-noprior}. It is similar to \Cref{among-incentivized}, but there are no receiver priors to control for since senders see the same median belief regardless of the receiver's party.

\begin{center}
\begin{threeparttable}[htbp!]
\begin{footnotesize}
\caption{Factors that lead incentivized senders to choose false computer messages}
{
\def\sym#1{\ifmmode^{#1}\else\(^{#1}\)\fi}
\begin{tabular}{l*{3}{c}}
\hline\hline
                    &\multicolumn{1}{c}{Vs. Party-True Aligned}         &\multicolumn{1}{c}{Vs. No Info}         &\multicolumn{1}{c}{Vs. Neutral Topics}         \\
\hline
Party-False Aligned &    0.145\sym{***}&    0.098\sym{**} &    0.090\sym{**} \\
                    &  (0.045)         &  (0.041)         &  (0.041)         \\
Question FE         &$\checkmark$         &$\checkmark$         &                  \\
Subject FE          &$\checkmark$         &$\checkmark$         &$\checkmark$         \\
Round FE            &$\checkmark$         &$\checkmark$         &$\checkmark$         \\
Vs. Party-True Aligned\hspace{20mm} &$\checkmark$         &                  &                  \\
Vs. No Info         &                  &$\checkmark$         &                  \\
Vs. Neutral Topics  &                  &                  &$\checkmark$         \\
\hline
Observations        &      779         &      789         &      638         \\
Subjects            &      229         &      223         &      206         \\
Mean                &    0.336         &    0.335         &    0.352         \\
\hline\hline
\multicolumn{4}{l}{\footnotesize Standard errors in parentheses}\\
\multicolumn{4}{l}{\footnotesize \sym{*} \(p<0.10\), \sym{**} \(p<0.05\), \sym{***} \(p<0.01\)}\\
\end{tabular}
}

\label{2-among-incentivized}
\end{footnotesize}
\begin{tablenotes}
\begin{scriptsize}
\vspace{-2mm}
\item \textbf{Notes:} OLS, errors clustered at subject level. Dependent variable: indicator for sender choosing the false message. Each column represents a different subset of the data, and Vs. lines indicate the comparison group. Party-False Aligned: indicator for the receiver's party being revealed and aligned with the false message. Prior-False Aligned: the receiver's prior belief that the incorrect answer is true. Party-True Aligned: indicator for the receiver's party being revealed and aligned with the true message. No Info: indicator for the receiver's party not being revealed. 
\end{scriptsize}
\vspace{5mm}
\end{tablenotes}
\end{threeparttable}
\end{center}

\Cref{2-among-incentivized} shows that the effects of party-false alignment are qualitatively identical to the effects in the primary experiment; party-false alignment leads to more false computer messages chosen in each comparison group. 

Appendix \Cref{2-among-unincentivized} shows that, as in Appendix \Cref{among-unincentivized}, unincentivized senders are not affected by these treatments. Subjects are not statistically-significantly more likely to send false computer messages when the messages are aligned with receiver's party in any comparison. As in the primary experiment, there is still a non-negligible share of unincentivized senders who choose false computer messages, at 24 percent. Appendix \Cref{2-ownparty} shows that part of this effect is driven by a form of expressive preferences in which senders prefer to select the option of their own party. As in \Cref{1-ownparty}, this effect plays a significant role in the unincentivized condition and does not play a significant role in the incentivized condition.


\Cref{2-incentives-effect} shows that, as in \Cref{incentives-effect}, the incentives treatment causes senders to choose more false computer messages. This effect is again largely driven by the condition in which party and false messages are aligned.

\begin{center}
\begin{threeparttable}[htbp!]
\begin{footnotesize}
\caption{The effect of incentives on choosing false computer messages}
{
\def\sym#1{\ifmmode^{#1}\else\(^{#1}\)\fi}
\begin{tabular}{l*{3}{c}}
\hline\hline
                    &\multicolumn{1}{c}{All Pol. Questions}         &\multicolumn{1}{c}{Party-False Aligned}         &\multicolumn{1}{c}{Interaction}         \\
\hline
Incentivized        &    0.078\sym{***}&    0.167\sym{***}&    0.035         \\
                    &  (0.022)         &  (0.036)         &  (0.025)         \\
Party-False Aligned x Incentivized &                  &                  &    0.135\sym{***}\\
      &                  &                  &  (0.032)         \\
Party-False Aligned x Unincentivized &                  &                  &    0.006         \\
    &                  &                  &  (0.025)         \\
Question FE         &$\checkmark$         &$\checkmark$         &$\checkmark$         \\
Round FE            &$\checkmark$         &$\checkmark$         &$\checkmark$         \\
Subject controls    &$\checkmark$         &$\checkmark$         &$\checkmark$         \\
All Questions       &$\checkmark$         &                  &$\checkmark$         \\
Only Party-False Aligned &                  &$\checkmark$         &                  \\
\hline
Observations        &     2421         &      822         &     2421         \\
Subjects            &      500         &      429         &      500         \\
Mean                &    0.274         &    0.318         &    0.274         \\
\hline\hline
\multicolumn{4}{l}{\footnotesize Standard errors in parentheses}\\
\multicolumn{4}{l}{\footnotesize \sym{*} \(p<0.10\), \sym{**} \(p<0.05\), \sym{***} \(p<0.01\)}\\
\end{tabular}
}

\label{2-incentives-effect}
\end{footnotesize}
\begin{tablenotes}
\begin{scriptsize}
\vspace{-2mm}
\item \textbf{Notes:} OLS, errors clustered at subject level. Dependent variable: indicator for sender choosing the false message. Subject controls: Gender, race, age, own party, education, and CRT score. Own party takes 0.5 if subject does not lean towards either party.
\end{scriptsize}
\vspace{5mm}
\end{tablenotes}
\end{threeparttable}
\end{center}

\subsection{Receivers' behavior and senders' beliefs}
Receivers rate pro-party news to be more likely than anti-party news to be true 38.4 percent of the time (subject-level clustered s.e. 2.5 pp), less likely to be true 29.2 percent of the time (s.e. 2.3 pp), and equally likely to be true 32.4 percent of the time (s.e. 2.4 pp). The difference between pro-party and anti-party news is statistically significant, and the point estimate is comparable to that of the substantially-larger sample in \textcite{T-WPa} (9.1 pp; s.e. 4.2 pp; $p=0.036$). Deviations in the opposite direction may reflect noise in answers or indicate that the predicted direction of motivated reasoning is incorrect. 

Senders are asked to predict each of these three percentages. They predict that receivers will rate pro-party news to be more likely to be true true 42.3 percent of the time (s.e. 0.8 pp), less likely to be true 29.8 percent of the time (s.e. 0.7 pp), and equally likely 27.9 percent of the time (s.e. 0.7 pp). The difference between pro-party and anti-party news is statistically significant (12.5 pp; s.e. 1.4 pp; $p<0.001$), and the point estimate is suggestively larger than that of receivers' behavior. The results are qualitatively similar, but not as stark, as those in the primary experiment. Differences may be due to senders' different perceptions about the two tasks or due to different behavior from receivers.

\subsection{Demand for information}
The majority of senders choose to pay to condition on the party of the receivers on each of the political questions. There are no sizeable differences between topics; on every political topic, between 55.6 percent and 63.1 percent of senders choose to pay. Since there is no effect on other parts of the experiment, this result indicates that senders value this information for instrumental reasons. Meanwhile, less than half of senders choose to pay on either of the neutral topics, indicating that senders particularly value this information on political topics. The gap in the share demanding information between political and neutral topics is 17.7 pp (s.e. 2.7 pp; $p<0.001$).

Senders' information choices are predictive of their choice of messages. Recall that in two of the four rounds, senders are allowed to condition on the receiver's party regardless of their choice to the demand-for-information question. Comparing behavior from high-demand and low-demand senders, there is a correlation between demand for information and willingness to choose false computer messages. Senders who demand the information choose false messages 40.6 percent of the time (s.e. 1.2 pp), and senders who do not demand the information choose false messages 30.1 percent of the time (s.e. 1.7 pp). The difference is large and statistically significant (10.5 pp; s.e. 2.0 pp; $p<0.001$). Suggestively, senders who do not demand information on political topics choose false messages a similar amount to subjects who send messages on neutral topics (30.8 percent). 

In addition to these correlations, there is causal evidence from the main treatment block that speaks to this relationship. Incentivized senders who randomly receive the receiver's party information are 6.5 pp more likely to choose false messages (s.e. 2.9 pp; $p=0.026$). These results suggest that senders causally condition on the party of their receiver to strategically choose more of the false computer messages.

\section{Conclusion}
\label{conclusion}

Understanding the root causes of disinformation is critical in determining how best to combat it. These experiments have demonstrated two complementary reasons that senders intentionally transmit false information: Incentives to be perceived as truthful by receivers, and a belief that receivers engage in politically-motivated reasoning. Across many different versions of this experiment --- when senders can or cannot condition on receivers' prior beliefs, when senders bet on messages sent by a computer to receivers instead of playing a game with them directly, and when senders just answer questions about their strategies --- the results consistently show three patterns: (1) When senders have incentives to be rated as truthful by receivers, they are less truthful; (2) when senders have these incentives, they condition messages on their receiver's party; (3) and when senders do not have these incentives, they do not systematically condition messages on their receiver's party. In addition, the primary experiment shows that receivers are naive about senders' strategic messaging. 

These findings suggest a number of potential avenues for future work. The experimental designs provided in this paper can be portable across domains, allowing applied researchers to test whether senders believe that receivers motivatedly reason on any issue, and then to test whether senders asymmetrically send receivers false information. They can also be used to study how these effects play out in field settings on social media or traditional media. Estimating persuasion rates on motivated topics (e.g. \cite{DG10}) and unpacking the welfare effects of different incentive structures would be fruitful next steps. 

Understanding why people knowingly send false information is an important issue on social media, but further work could extend these analyses to consider the tradeoffs senders face between informativeness and slant. For instance, in many traditional media environments, media outlets often choose which information to report, rather than whether to outright disinform. One valuable extension would be to study whether selective reporting is also affected by incentives and beliefs about motivated reasoning in the same manner.

In addition, it would be natural to extend these predictions from a static to a dynamic environment. That is, what happens when incentives to be rated as truthful are replaced with repeated interactions in a reputation game? The effects are ambiguous: on the one hand, dynamic incentives may lead senders to continue to tailor messages to receivers' motivated beliefs, but on the other hand, the impact of naive higher-order beliefs may be corrected if receivers learn about past messages from senders. 

Lastly, these results suggest potential levers for reducing disinformation in politicized environments, either via changes to the structure of incentives for news suppliers or via changes to the perceived impact that motivated reasoning plays. Even when \textit{incentives} are fixed in news environments, these results suggest that \textit{beliefs} may be malleable. Increasing receiver sophistication or reducing their bias in a way that changes senders' higher-order beliefs may reduce both trust in, and the supply of, false information. Furthermore, while senders correctly forecast that receivers will trust political ``good news'' more than ``bad news,'' these results suggest that they often overstate this gap. Correcting senders' misperceptions of receivers may reduce how frequently they supply false information.

\vspace{10mm}


\nocite{*} 
\AtNextBibliography{\small}
{\singlespace\printbibliography} 

\newpage 

\appendix

\section{Additional Theoretical Results} \label{appendix-a}

\subsection{Proposition 2}

To prove \Cref{prop:list-eq}, we first calculate the BNE for arbitrary $\lambda$ in \Cref{lem:bne}, and then construct the motivated equilibria (ME) generated from each BNE. 
\begin{lemma}
\label{lem:bne}
In a game with parameter $\lambda$, parameters allow for three possible BNE in pure strategies:
\begin{enumerate}
    \item A separating BNE where S plays $x_H|\theta_H$ and $x_L|\theta_L$, R plays $a(x_H) = 1$ and $a(x_L) = \max\{1-\lambda/2, 0\} = 1-\lambda/2$ (since $\lambda < 2$ by assumption).

    \item A pooling BNE where S plays $x_H|\theta_H$ and $x_H|\theta_L$, R plays $a(x_H) = \min\{\pi + \lambda/2,1\}$ and some suitable off-path $a(x_L)$.

    \item A pooling BNE where S plays $x_L|\theta_H$ and $x_L|\theta_L$, R plays $a(x_L) = \max\{1-\pi - \lambda/2,0\}$ and some suitable off-path $a(x_H)$.
\end{enumerate}
\end{lemma}

\begin{proof}[\textbf{Proof of \Cref{lem:bne}}]
We show when the above strategy profiles are in a BNE: 
\begin{enumerate}
    \item In the separating BNE: S will never profitably deviate to $x_L|\theta_H$ and will not profitably deviate to $x_H|\theta_L$ iff $\tau + \gamma(1-\lambda/2) \geq \gamma$, i.e. when $\tau/\gamma \geq \lambda / 2$.

    \item In the $x_H$-pooling BNE: If S can profitably deviate to $x_L|\theta_H$, then she can also profitably deviate to $x_L|\theta_L$, so we only need to consider the latter deviation. This deviation will not be profitable iff $\gamma \cdot \min\{\pi + \lambda/2,1\} \geq \tau + \gamma a(x_L)$, i.e. when $\tau / \gamma \leq \min\{\pi + \lambda/2,1\} - a(x_L)$. 

    \item In the $x_L$-pooling BNE: If S can profitably deviate to $x_H|\theta_L$, then she can also profitably deviate to $x_H|\theta_H$, so we only need to consider the latter deviation. This deviation will not be profitable iff $\gamma \cdot \max\{1-\pi - \lambda/2,0\} \geq \tau + \gamma a(x_H)$, i.e. when $\tau / \gamma \leq \max\{1-\pi - \lambda/2,0\} - a(x_H)$. 
\end{enumerate}

No separating BNE exists where S plays $x_L|\theta_H$ and $x_H|\theta_L$. R would play $a(x_H) = \lambda/2$ and $a(x_L) = 0$. S would have a profitable deviation to $x_H | \theta_H$. As such, we have computed all the possible BNE. 
\end{proof}

\begin{proof}[\textbf{Proof of \Cref{prop:list-eq}}]

First, we observe that R's strategy in an ME involves the strategies described in \Cref{lem:bne}, but with $\hat{\lambda}_R$ replacing $\lambda$. 

Next, we consider S's behavior in an ME. S will calculate her best response to the above strategies but with motivated reasoning parameter $\hat{\lambda}_S$. 

\begin{enumerate}
\item In the ME generated from the separating BNE: 
\begin{itemize}
    \item Given $\theta_H$, S will have a best response of $x_H$ iff $\tau + \gamma \geq \gamma(1-\hat{\lambda}_S/2)$, which is always true.

    \item Given $\theta_L$, S will have a best response of $x_H$ iff $\gamma \geq \tau + \gamma(1-\hat{\lambda}_S/2)$, i.e. when $\tau / \gamma \leq \hat{\lambda}_S/2$.

    \item For a ME in which R plays a best response to $x_H | \theta_H$ and $x_L | \theta_L$, and S plays $x_H | \theta_H$ and $x_L | \theta_L$, we need for $\tau / \gamma \geq \hat{\lambda}_R/2$ and $\tau / \gamma \leq \hat{\lambda}_S/2$. Since $\hat{\lambda}_S > \hat{\lambda}_R$, this condition reduces to: 
    \[ \tau / \gamma \leq \hat{\lambda}_S/2.\]
    
    \item For a ME in which R plays a best response to $x_H | \theta_H$ and $x_L | \theta_L$, but S plays $x_H | \theta_H$ and $x_H | \theta_L$, we need for $\tau / \gamma \geq \hat{\lambda}_R/2$ and $\tau / \gamma \leq \hat{\lambda}_S/2$. That is: 
    \[ \tau / \gamma \in [\hat{\lambda}_R/2,\hat{\lambda}_S/2].\]
\end{itemize}

\item In the ME generated from the $x_H$-pooling BNE: 
\begin{itemize}
    \item Given $\theta_H$, S will have a best response of $x_H$ iff $\tau + \gamma \cdot \min\{\pi + \hat{\lambda}_S/2,1\} \geq \gamma a(x_L)$, i.e. when $\tau / \gamma \geq -(\min\{\pi + \hat{\lambda}_S/2,1\} - a(x_L))$.

    \item Given $\theta_L$, S will have a best response of $x_H$ iff $\gamma \cdot \min\{\pi + \hat{\lambda}_S/2,1\} \geq \tau + \gamma a(x_L)$, i.e. when $\tau / \gamma \leq \min\{\pi + \hat{\lambda}_S/2,1\} - a(x_L)$.

    \item For a ME in which R plays a best response to $x_H | \theta_H$ and $x_H | \theta_L$, and S plays $x_H | \theta_H$ and $x_H | \theta_L$, we need for $\tau / \gamma \leq \min\{\pi + \hat{\lambda}_R/2,1\} - a(x_L)$, and both $\tau / \gamma \geq -(\min\{\pi + \hat{\lambda}_S/2,1\} - a(x_L))$ and $\tau / \gamma \leq \min\{\pi + \hat{\lambda}_S/2,1\} - a(x_L)$. Since $\hat{\lambda}_S \geq \hat{\lambda}_R$, both of these conditions are implied, so this condition reduces to 
    \[ \tau / \gamma \leq \min\{\pi + \hat{\lambda}_R/2,1\} - a(x_L). \] 

    \item By the same argument, there are no ME in which R plays a best response to $x_H | \theta_H$ and $x_L | \theta_L$, but S does not play $x_H | \theta_H$ and $x_H | \theta_L$.



\end{itemize}

\item In the ME generated from the $x_L$-pooling BNE: 
\begin{itemize}
    \item Given $\theta_H$, S will have a best response of $x_H$ iff $\tau + \gamma \cdot a(x_H) \geq \gamma \max\{1-\pi - \hat{\lambda}_S/2,0\}$, i.e. when $\tau / \gamma \geq \max\{1-\pi - \hat{\lambda}_S/2,0\} - a(x_H)$.

    \item Given $\theta_L$, S will have a best response of $x_H$ iff $\gamma \cdot a(x_H) \geq \tau + \gamma \cdot \max\{1-\pi - \hat{\lambda}_S/2, 0 \}$, i.e. when $\tau / \gamma \leq -(\max\{1-\pi - \hat{\lambda}_S/2,0\} - a(x_H))$.

    \item For a ME in which R plays a best response to $x_L | \theta_H$ and $x_L | \theta_L$, and S plays $x_L | \theta_H$ and $x_L | \theta_L$, we need for $\tau / \gamma \leq \max\{1-\pi - \hat{\lambda}_R/2,0\} - a(x_H)$, and both $\tau / \gamma \leq \max\{1-\pi - \hat{\lambda}_S/2,0\} - a(x_H)$ and $\tau / \gamma \geq -(\max\{1-\pi - \hat{\lambda}_S/2,0\} - a(x_H))$. Since $\hat{\lambda}_S \geq \hat{\lambda}_R$ and $\tau/\gamma \geq 0$, these conditions can be reduced to: 
    \[ \tau / \gamma \leq \max\{1-\pi - \hat{\lambda}_S/2,0\} - a(x_H).\footnote{Note that if $\max\{1-\pi - \hat{\lambda}_S/2,0\} < -(\max\{1-\pi - \hat{\lambda}_S/2,0\} - a(x_H))$, then it is not possible for $\tau/\gamma$ to be less than the first and greater than the second term. If $\max\{1-\pi - \hat{\lambda}_S/2,0\} \geq -(\max\{1-\pi - \hat{\lambda}_S/2,0\} - a(x_H))$, then $-(\max\{1-\pi - \hat{\lambda}_S/2,0\} - a(x_H)) \leq 0$.} \]

    \item For a ME in which R plays a best response to $x_L | \theta_H$ and $x_L | \theta_L$, and S plays $x_H | \theta_H$ and $x_L | \theta_L$, we need for 
    \begin{align*}
    \tau / \gamma &\leq \max\{1-\pi - \hat{\lambda}_R/2,0\} - a(x_H), \\
    \tau / \gamma &\geq \max\{1-\pi - \hat{\lambda}_S/2,0\} - a(x_H), \\
    \text{ and }\tau / \gamma &\geq -(\max\{1-\pi - \hat{\lambda}_S/2,0\} - a(x_H)).
    \end{align*}  
    
    \item For a ME in which R plays a best response to $x_L | \theta_H$ and $x_L | \theta_L$, and S plays $x_H | \theta_H$ and $x_H | \theta_L$, we need for $\tau / \gamma \leq \max\{1-\pi - \hat{\lambda}_R/2,0\} - a(x_H)$, $\tau / \gamma \geq \max\{1-\pi - \hat{\lambda}_S/2,0\} - a(x_H)$, and $\tau / \gamma \leq -(\max\{1-\pi - \hat{\lambda}_S/2,0\} - a(x_H))$. Since $\tau/\gamma \geq 0$, the second inequality is implied by the third inequality, so these conditions can be reduced to: 
    \begin{align*}
        \tau / \gamma &\leq \max\{1-\pi - \hat{\lambda}_R/2,0\} - a(x_H) \\ 
        \text{ and } \tau / \gamma &\leq -(\max\{1-\pi - \hat{\lambda}_S/2,0\} - a(x_H)).
    \end{align*}
\end{itemize}
\end{enumerate}


    




These enumerate all the ME in the game. \hspace{77mm} $\blacksquare$

\end{proof}

\newpage 

\subsection{Visualizing Proposition 2}


In \Cref{fig:eq-viz} below, we see how the different equilibria depend on higher-order beliefs. I plot the range of values for $\tau/\gamma$ for which each of the strategies are in a motivated equilibrium. For other parameters, I set $\pi = 0.587$ (the average prior in \Cref{1-priors}), receivers' motivated beliefs to be $\lambda = \hat{\lambda}_R = 0.114$ (corresponding to the average gap in ratings of pro- and anti-party news), and senders' beliefs about receivers' motivated beliefs to be $\hat{\lambda}_S = 0.30$ (taken from their average predictions). $\tau$ is set to 1; this choice is arbitrary as behavior only depends on $\tau$ through $\tau/\gamma$. When relevant, I assume full punishment $a(x_H) = a(x_L) = 0$ for off-path play.

We see that if $\gamma$ is small, there is only one equilibrium, and the sender always reveals the state. As $\gamma$ increases, pooling strategies are in equilibrium as well. For $\gamma \in [6.7,17.5]$, truthtelling is no longer an equilibrium, but a motivated equilibrium exists. For $\gamma > 17.5$, only pooling strategies survive as equilibria. 

If instead, $\hat{\lambda}_S = \hat{\lambda}_R$, then only BNE remain; if both players agree that $\lambda$ is low, the truthtelling equilibrium includes the full range of the non-BNE ME; if they agree that $\lambda$ is high, truthtelling is only an equilibrium for the initial levels of $\gamma$, and only pooling strategies survive for higher $\gamma$.

\newpage 

\begin{figure}[ht]
    \centering

\caption{Visualizing Motivated Equilibria}
\label{fig:eq-viz}

\vspace{5mm}
\textbf{Panel A:} $\hat{\lambda}_S = 0.30$ and $\lambda = \hat{\lambda}_R = 0.114$
\vspace{7mm}

\begin{tikzpicture}[scale=.5]
\draw[latex-latex] (0,0) -- (30,0) ; 
\foreach \x in {0,5,10,15,20,25,30} 
\draw[shift={(\x,0)},color=black] (0pt,3pt) -- (0pt,-3pt);
\foreach \x in {0,5,10,15,20,25,30} 
\draw[shift={(\x,0)},color=black] (0pt,0pt) -- (0pt,-3pt) node[below] 
{$\x$};
\draw[shift={(0,3.2)}, *-*, color=black, very thick] (-0.25,0) -- (6.92,0);
\draw[shift={(0,2.6)}, *->, color=gray] (1.40,0) -- (30,0);
\draw[shift={(0,2.0)}, *->, color=lightgray] (3.65,0) -- (30,0);
\draw[shift={(0,1.4)}, *-*, color=cyan, very thick] (6.42,0) -- (17.79,0);
\draw[shift={(0,0.8)}, *-*, color=teal] (2.66,0) -- (3.95,0);
\end{tikzpicture}

\vspace{8mm}

\textbf{Panel B:} $\hat{\lambda}_S = 0.114$ and $\lambda = \hat{\lambda}_R = 0.114$
\vspace{7mm}



\begin{tikzpicture}[scale=.5]
\draw[latex-latex] (0,0) -- (30,0) ; 
\foreach \x in {0,5,10,15,20,25,30} 
\draw[shift={(\x,0)},color=black] (0pt,3pt) -- (0pt,-3pt);
\foreach \x in {0,5,10,15,20,25,30} 
\draw[shift={(\x,0)},color=black] (0pt,0pt) -- (0pt,-3pt) node[below] 
{$\x$};
\draw[shift={(0,2.0)}, *-*, color=black, very thick] (-0.25,0) -- (17.79,0);
\draw[shift={(0,1.4)}, *->, color=gray] (1.40,0) -- (30,0);
\draw[shift={(0,0.8)}, *->, color=lightgray] (2.66,0) -- (30,0);
\end{tikzpicture}



\vspace{8mm}

\textbf{Panel C:} $\hat{\lambda}_S = 0.30$ and $\lambda = \hat{\lambda}_R = 0.30$
\vspace{7mm}

\begin{tikzpicture}[scale=.5]
\draw[latex-latex] (0,0) -- (30,0) ; 
\foreach \x in {0,5,10,15,20,25,30} 
\draw[shift={(\x,0)},color=black] (0pt,3pt) -- (0pt,-3pt);
\foreach \x in {0,5,10,15,20,25,30} 
\draw[shift={(\x,0)},color=black] (0pt,0pt) -- (0pt,-3pt) node[below] 
{$\x$};
\draw[shift={(0,2.0)}, *-*, color=black, very thick] (-0.25,0) -- (6.92,0);
\draw[shift={(0,1.4)}, *->, color=gray] (1.30,0) -- (30,0);
\draw[shift={(0,0.8)}, *->, color=lightgray] (3.65,0) -- (30,0);
\end{tikzpicture}
\begin{threeparttable}
\begin{tablenotes}
\begin{scriptsize}
\vspace{-2mm}
\item \textbf{Notes:} Visualization of which sets of strategies are in each motivated equilibria as $\gamma$ changes. Panels vary in their values of $\hat{\lambda}_S$ and $\hat{\lambda}_R$. Within each panel, the ME from top to bottom are: \textbf{BNE with full truthtelling}; {\color{gray}{BNE where S always sends $x_H$}}; {\color{lightgray}{BNE where S always sends $x_L$}}; \textbf{{\color{cyan}{S sends $x_H|\theta_H$ and $x_H|\theta_L$ but R plays a best response to $x_H|\theta_H$ and $x_L|\theta_L$}}}; {\color{teal}{S sends $x_L|\theta_H$ and $x_L|\theta_L$ but R plays a best response to $x_H|\theta_H$ and $x_L|\theta_L$}}.
\end{scriptsize}
\end{tablenotes}
\end{threeparttable}
\end{figure}

\vspace{5mm}

\noindent
    
\newpage

\section{Additional Tables and Figures for the Primary Experiment}

\begin{center}
\begin{threeparttable}
\begin{footnotesize}
\caption{Balance Table for Senders}
\def\sym#1{\ifmmode^{#1}\else\(^{#1}\)\fi}
\begin{tabular}{p{3.4cm}*{4}{c}}
\hline\hline
&\multicolumn{1}{c}{Incentivized}&\multicolumn{1}{c}{Unincentivized}&\multicolumn{1}{c}{Inc. vs. Uninc.}&\multicolumn{1}{p{2.2cm}}{\centering p-value}\\
\hline
Age  & 35.431 & 33.227 & 2.165 & 0.148 \\   
& (1.074) & (1.038) & (1.490) & \\
White & 0.766 & 0.745 & 0.022 & 0.694 \\   
& (0.038) & (0.039) & (0.055) & \\
Female  & 0.516 & 0.624 & -0.108 & 0.087 \\   
& (0.045) & (0.044) & (0.063) & \\
Education  & 15.064 & 14.849 & 0.215 & 0.402 \\   
& (0.177) & (0.186) & (0.257) & \\
CRT score  & 1.387 & 1.554 & -0.167 & 0.262 \\   
& (0.106) & (0.104) & (0.148) & \\
Party  & 0.508 & 0.512 & -0.004 & 0.944 \\   
& (0.045) & (0.045) & (0.063) & \\
Other's party revealed  & 0.693 & 0.657 & 0.036 & 0.108 \\   
 & (0.016) & (0.016) & (0.022) & \\
Own party revealed  & 0.654 & 0.672 & -0.019 & 0.754 \\   
& (0.043) & (0.042) & (0.060) & \\
Others' have  & 0.500 & 0.516 & -0.156 & 0.584 \\   
party-truth aligned & (0.022) & (0.018) & (0.028) & \\
Self has  & 0.512 & 0.526 & -0.015 & 0.344 \\   
party-truth aligned & (0.011) & (0.011) & (0.016) & \\
\hline
\(N\)  & 10,036 & 10,107 & 20,143 & \\
\hline
\hline
\end{tabular}
\label{balance-table-s}
\end{footnotesize}
\begin{tablenotes}
\item \begin{scriptsize} \textbf{Notes:} Standard errors in parentheses. Education is in years. CRT score is number of correct answers on the cognitive reflection task. Party is 1 if subject is Republican or Republican-leaning. Other party known only pertains to Rounds 1-7 (in later rounds the party is always revealed). Party-truth alignment is defined in the main text. Party-truth alignment restricted to observations where party is known.
\end{scriptsize}
\end{tablenotes}
\end{threeparttable}
\end{center}

\begin{center}
\begin{threeparttable}
\begin{footnotesize}
\caption{Balance Table for Receivers}
\def\sym#1{\ifmmode^{#1}\else\(^{#1}\)\fi}
\begin{tabular}{p{3.4cm}*{4}{c}}
\hline\hline
&\multicolumn{1}{c}{Incentivized}&\multicolumn{1}{c}{Unincentivized}&\multicolumn{1}{c}{Inc. vs. Uninc.}&\multicolumn{1}{p{2.2cm}}{\centering p-value}\\
\hline
Age  & 34.096 & 33.635 & 0.461 & 0.772 \\   
& (1.172) & (1.075) & (1.587) & \\
White & 0.718 & 0.729 & -0.012 & 0.837 \\   
& (0.041) & (0.039) & (0.056) & \\
Female  & 0.572 & 0.558 & 0.014 & 0.823 \\   
& (0.045) & (0.044) & (0.062) & \\
Education  & 15.072 & 14.828 & 0.244 & 0.327 \\   
& (0.173) & (0.179) & (0.248) & \\
CRT score  & 1.354 & 1.296 & 0.058 & 0.679 \\   
& (0.098) & (0.102) & (0.141) & \\
Party  & 0.500 & 0.512 & -0.012 & 0.849 \\   
& (0.045) & (0.044) & (0.063) & \\
Other's party revealed  & 0.681 & 0.650 & 0.019 & 0.206 \\   
 & (0.018) & (0.016) & (0.024) & \\
Own party revealed  & 0.670 & 0.682 & -0.012 & 0.834 \\   
& (0.042) & (0.041) & (0.059) & \\
\hline
\(N\)  & 2,726 & 2,832 & 5,558 & \\
\hline
\hline
\end{tabular}
\label{balance-table-r}
\end{footnotesize}
\begin{tablenotes}
\item \begin{scriptsize} \textbf{Notes:} Standard errors in parentheses. Education is in years. CRT score is number of correct answers on the cognitive reflection task. Party is 1 if subject is Republican or Republican-leaning. Other party known only pertains to Rounds 1-7 (in later rounds the party is always revealed). 
\end{scriptsize}
\end{tablenotes}
\end{threeparttable}
\end{center}

\clearpage

\begin{center}
\begin{threeparttable}[htbp!]
\begin{footnotesize}
\caption{Comparison of effects in Block 1 and Block 2}
{
\def\sym#1{\ifmmode^{#1}\else\(^{#1}\)\fi}
\begin{tabular}{l*{4}{c}}
\hline\hline
                    &\multicolumn{1}{c}{Block 1, Incentives}         &\multicolumn{1}{c}{Block 2, Incentives}         &\multicolumn{1}{c}{Block 1, No Inc.}         &\multicolumn{1}{c}{Block 2, No Inc.}         \\
\hline
Party-False Aligned &    0.071\sym{**} &    0.214\sym{***}&    0.005         &    0.018         \\
                    &  (0.032)         &  (0.044)         &  (0.032)         &  (0.033)         \\
Prior-False Aligned &    0.248\sym{***}&                  &    0.063         &                  \\
                    &  (0.053)         &                  &  (0.039)         &                  \\
Question FE         &$\checkmark$         &$\checkmark$         &$\checkmark$         &$\checkmark$         \\
Subject FE          &$\checkmark$         &$\checkmark$         &$\checkmark$         &$\checkmark$         \\
Round FE            &$\checkmark$         &$\checkmark$         &$\checkmark$         &$\checkmark$         \\
Vs. Party-True Aligned &$\checkmark$         &$\checkmark$         &$\checkmark$         &$\checkmark$         \\
Incentivized        &$\checkmark$         &$\checkmark$         &                  &                  \\
Block 1             &$\checkmark$         &                  &$\checkmark$         &                  \\
Block 2             &                  &$\checkmark$         &                  &$\checkmark$         \\
\hline
Observations        &     4990         &      496         &     4636         &      500         \\
Subjects            &      124         &      124         &      125         &      125         \\
Mean                &    0.296         &    0.298         &    0.211         &    0.218         \\
\hline\hline
\multicolumn{5}{l}{\footnotesize Standard errors in parentheses}\\
\multicolumn{5}{l}{\footnotesize \sym{*} \(p<0.10\), \sym{**} \(p<0.05\), \sym{***} \(p<0.01\)}\\
\end{tabular}
}

\label{by-block-by-inc}
\end{footnotesize}
\begin{tablenotes}
\begin{scriptsize}
\item \textbf{Notes:} OLS, errors clustered at subject level. Dependent variable: indicator for sender choosing the false message. Each column represents a different subset of the data, and Vs. lines indicate the comparison group. Party-False Aligned: indicator for the receiver's party being revealed and aligned with the false message. Prior-False Aligned: the receiver's prior belief that the incorrect answer is true. Block 1: Senders condition on receivers' priors; Block 2; Senders cannot condition on receivers' priors.
\end{scriptsize}
\end{tablenotes}
\end{threeparttable}
\end{center}

\clearpage

\begin{center}
\begin{threeparttable}[htbp!]
\begin{footnotesize}
\caption{Factors that lead incentivized senders to send false messages: Block 1 and Block 2}
{
\def\sym#1{\ifmmode^{#1}\else\(^{#1}\)\fi}
\begin{tabular}{l*{3}{c}}
\hline\hline
                    &\multicolumn{1}{c}{Vs. Party-True Aligned}         &\multicolumn{1}{c}{Vs. No Info}         &\multicolumn{1}{c}{Vs. Neutral Topics}         \\
\hline
Party-False Aligned &    0.092\sym{***}&    0.075\sym{***}&    0.051\sym{*}  \\
                    &  (0.028)         &  (0.023)         &  (0.028)         \\
Question FE         &$\checkmark$         &$\checkmark$         &                  \\
Subject FE          &$\checkmark$         &$\checkmark$         &$\checkmark$         \\
Round FE            &$\checkmark$         &$\checkmark$         &$\checkmark$         \\
Vs. Party-True Aligned\hspace{20mm} &$\checkmark$         &                  &                  \\
Vs. No Info         &                  &$\checkmark$         &                  \\
Vs. Neutral Topics  &                  &                  &$\checkmark$         \\
\hline
Observations        &     5486         &     4963         &     4313         \\
Subjects            &      124         &      124         &      124         \\
Mean                &    0.296         &    0.303         &    0.327         \\
\hline\hline
\multicolumn{4}{l}{\footnotesize Standard errors in parentheses}\\
\multicolumn{4}{l}{\footnotesize \sym{*} \(p<0.10\), \sym{**} \(p<0.05\), \sym{***} \(p<0.01\)}\\
\end{tabular}
}

\label{among-incentivized-noprior}
\end{footnotesize}
\begin{tablenotes}
\begin{scriptsize}
\item \textbf{Notes:} OLS, errors clustered at subject level. Dependent variable: indicator for sender choosing the false message. Each column represents a different subset of the data, and Vs. lines indicate the comparison group. Party-False Aligned: indicator for the receiver's party being revealed and aligned with the false message. Party-True Aligned: indicator for the receiver's party being revealed and aligned with the true message. No Info: indicator for the receiver's party not being revealed. Only incentivized senders included. Observations include questions in both Block 1 (where senders condition on receiver priors) and Block 2 (where they do not). 
\end{scriptsize}
\end{tablenotes}
\end{threeparttable}
\end{center}

\clearpage

\begin{center}
\begin{threeparttable}[htbp!]
\begin{footnotesize}
\caption{Factors that lead unincentivized senders to send false messages: Block 1 and Block 2}
{
\def\sym#1{\ifmmode^{#1}\else\(^{#1}\)\fi}
\begin{tabular}{l*{3}{c}}
\hline\hline
                    &\multicolumn{1}{c}{Vs. Party-True Aligned}         &\multicolumn{1}{c}{Vs. No Info}         &\multicolumn{1}{c}{Vs. Neutral Topics}         \\
\hline
Party-False Aligned &    0.006         &   -0.034         &   -0.009         \\
                    &  (0.028)         &  (0.030)         &  (0.024)         \\
Question FE         &$\checkmark$         &$\checkmark$         &                  \\
Subject FE          &$\checkmark$         &$\checkmark$         &$\checkmark$         \\
Round FE            &$\checkmark$         &$\checkmark$         &$\checkmark$         \\
Vs. Party-True Aligned\hspace{20mm} &$\checkmark$         &                  &                  \\
Vs. No Info         &                  &$\checkmark$         &                  \\
Vs. Neutral Topics  &                  &                  &$\checkmark$         \\
\hline
Observations        &     5136         &     4888         &     4189         \\
Subjects            &      125         &      125         &      125         \\
Mean                &    0.212         &    0.215         &    0.204         \\
\hline\hline
\multicolumn{4}{l}{\footnotesize Standard errors in parentheses}\\
\multicolumn{4}{l}{\footnotesize \sym{*} \(p<0.10\), \sym{**} \(p<0.05\), \sym{***} \(p<0.01\)}\\
\end{tabular}
}

\label{among-unincentivized-noprior}
\end{footnotesize}
\begin{tablenotes}
\begin{scriptsize}
\item \textbf{Notes:} OLS, errors clustered at subject level. Dependent variable: indicator for sender choosing the false message. Each column represents a different subset of the data, and Vs. lines indicate the comparison group. Party-False Aligned: indicator for the receiver's party being revealed and aligned with the false message. Party-True Aligned: indicator for the receiver's party being revealed and aligned with the true message. No Info: indicator for the receiver's party not being revealed. Only unincentivized senders included. Observations include questions in both Block 1 (where senders condition on receiver priors) and Block 2 (where they do not). 
\end{scriptsize}
\end{tablenotes}
\end{threeparttable}
\end{center}

\clearpage

\begin{figure}[htb!]
    \caption{Receivers' Prior Beliefs by Topic}
    \label{1-priors}
    \centering
\includegraphics[width=\textwidth]{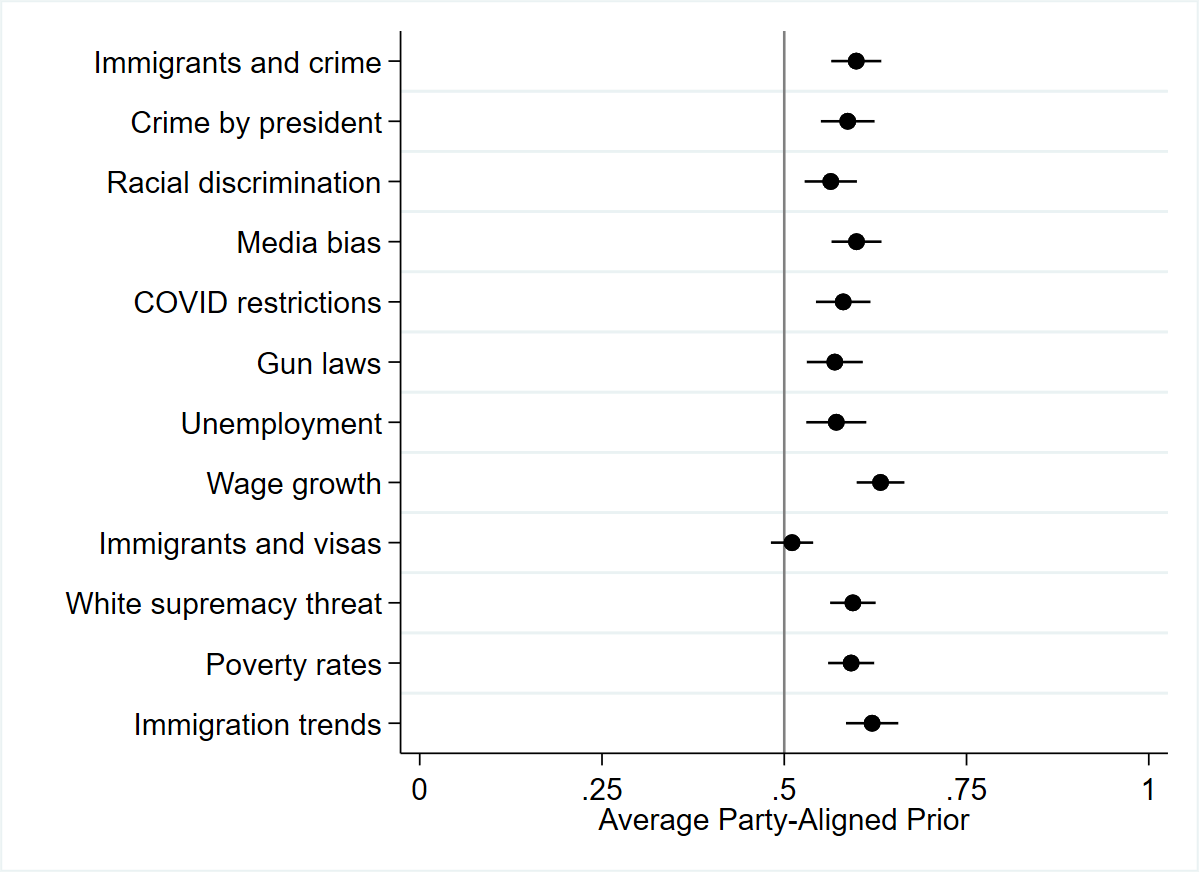}
\begin{threeparttable}
\begin{tablenotes}
\begin{scriptsize}
\vspace{-4mm}
\item \textbf{Notes:} Standard errors clustered at subject level. Party-Aligned Prior denotes the prior belief that the receiver has that the pro-party state is true, as described in Table 1. Error bars correspond to 95 percent confidence intervals.
\end{scriptsize}
\end{tablenotes}
\end{threeparttable}
\end{figure}

\clearpage

\begin{figure}[htb!]
    \caption{The Effect of Incentives on Survey Beliefs about Senders and Receivers}
    \label{effects-survey}
    \centering
\includegraphics[width=.77\textwidth]{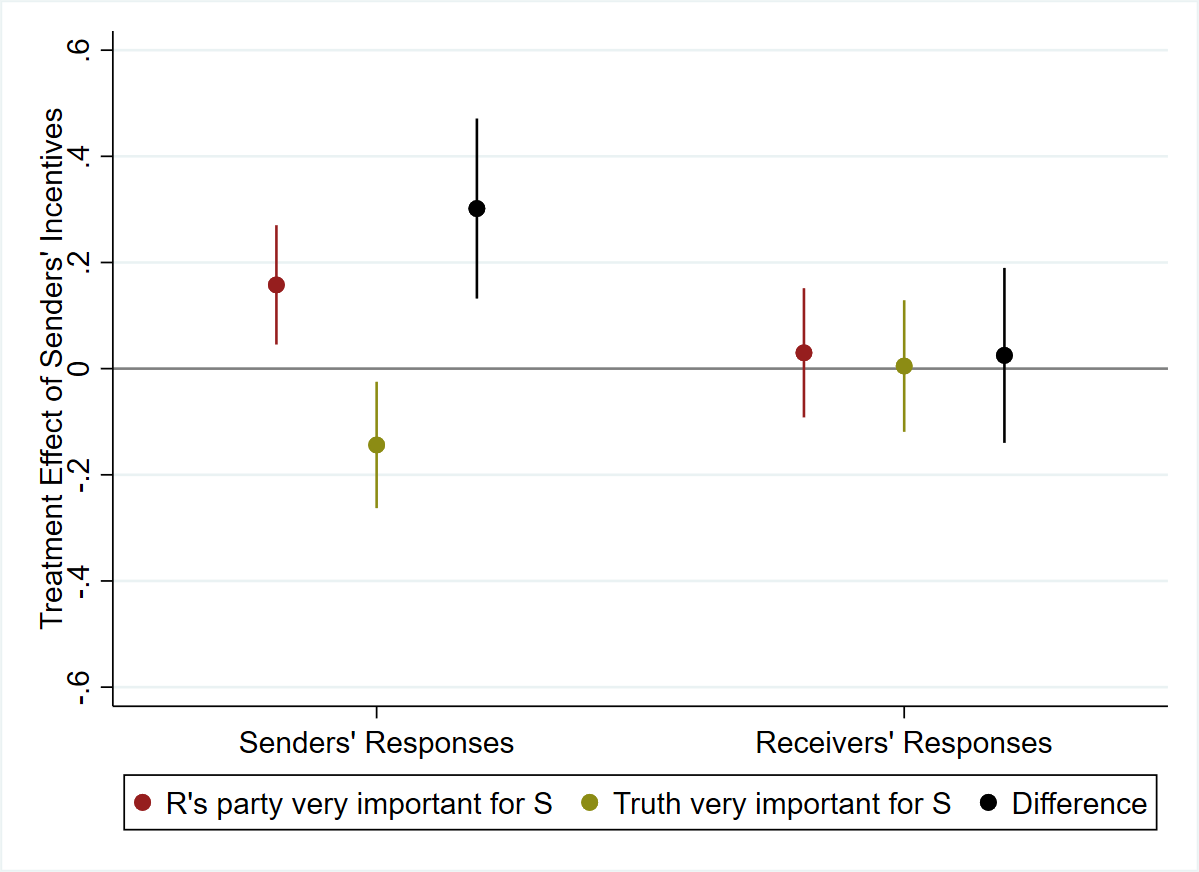}
\includegraphics[width=.77\textwidth]{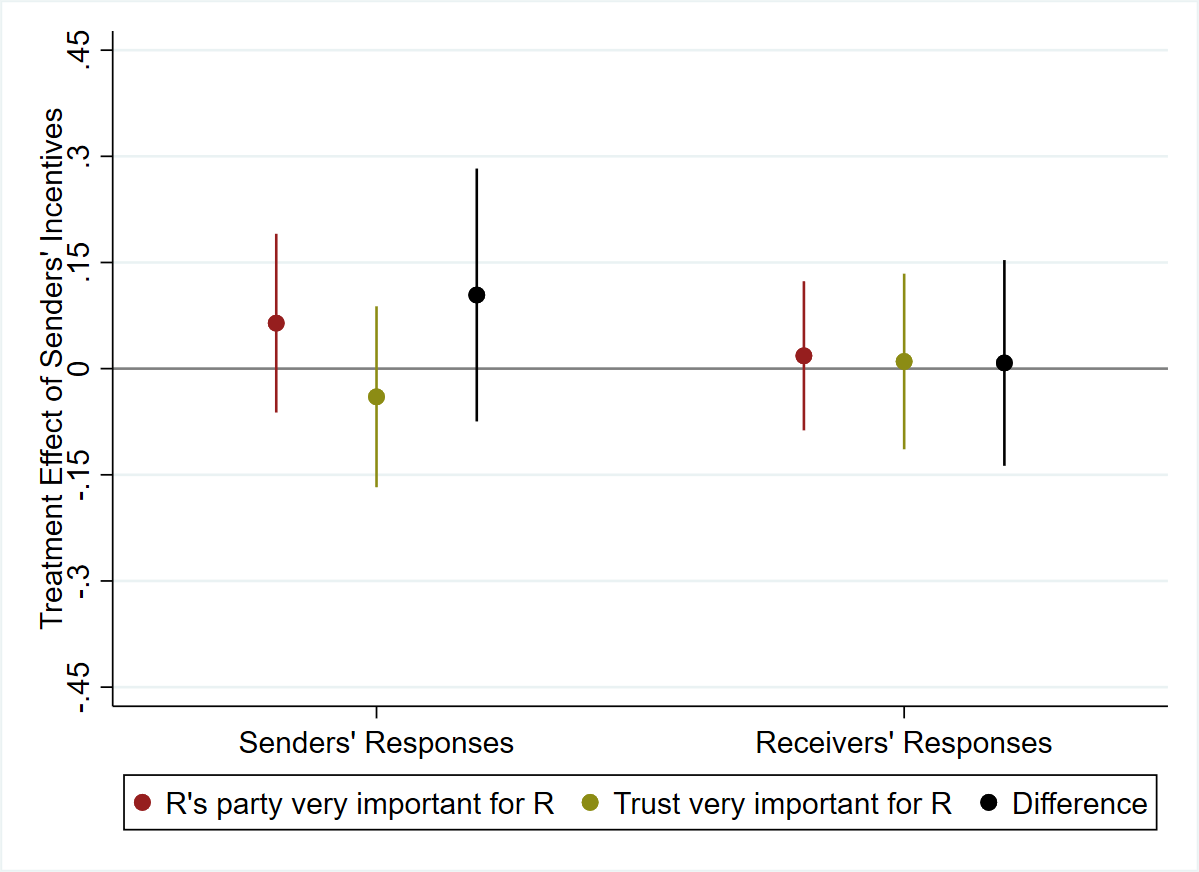}
\begin{threeparttable}
\begin{tablenotes}
\begin{scriptsize}
\vspace{-4mm}
\item \textbf{Notes:} OLS regression coefficients, robust standard errors. DV takes 1 if subject answers ``very important'' or ``extremely important'' and 0 otherwise. Exact questions are provided in the experimental materials. Controls for age, race, gender, education, CRT score, and own party are included. The top panel shows that incentivized Ss believe that R's party matters more, and the truth matters less, in their decisions, while Rs are unaffected by S's incentives. The bottom panel shows that incentives do not significantly affect senders' or receivers' beliefs about the impact of R's party and R's trust on R behavior. Error bars correspond to 95 percent confidence intervals. 
\end{scriptsize}
\end{tablenotes}
\end{threeparttable}
\end{figure}

\clearpage

\begin{center}
\begin{threeparttable}[htbp!]
\begin{footnotesize}
\caption{The interaction between beliefs and incentives}
{
\def\sym#1{\ifmmode^{#1}\else\(^{#1}\)\fi}
\begin{tabular}{l*{4}{c}}
\hline\hline
                    &\multicolumn{1}{c}{All Questions}         &\multicolumn{1}{c}{50-50 Priors}         &\multicolumn{1}{c}{Party-False Aligned}         &\multicolumn{1}{c}{50-50 \& Party-False}         \\
\hline
Incentivized        &    0.021         &   -0.008         &    0.018         &    0.000         \\
                    &  (0.036)         &  (0.042)         &  (0.054)         &  (0.074)         \\
S's Belief of R's Updating  &   -0.121\sym{**} &   -0.132\sym{**} &   -0.198\sym{**} &   -0.206\sym{*}  \\
            &  (0.059)         &  (0.063)         &  (0.088)         &  (0.107)         \\
Incentivized x S's Belief &    0.195\sym{**} &    0.239\sym{**} &    0.323\sym{**} &    0.447\sym{**} \\
              &  (0.082)         &  (0.097)         &  (0.124)         &  (0.172)         \\
Question FE         &$\checkmark$         &$\checkmark$         &$\checkmark$         &$\checkmark$         \\
Round FE            &$\checkmark$         &$\checkmark$         &$\checkmark$         &$\checkmark$         \\
Subject controls    &$\checkmark$         &$\checkmark$         &$\checkmark$         &$\checkmark$         \\
All Questions       &$\checkmark$         &$\checkmark$         &                  &                  \\
Only Party-False Aligned &                  &                  &$\checkmark$         &$\checkmark$         \\
Only 50-50 Priors   &                  &$\checkmark$         &                  &$\checkmark$         \\
\hline
Observations        &    14248         &     1294         &     4702         &      427         \\
Subjects            &      249         &      249         &      220         &      220         \\
Mean                &    0.250         &    0.216         &    0.250         &    0.216         \\
\hline\hline
\multicolumn{5}{l}{\footnotesize Standard errors in parentheses}\\
\multicolumn{5}{l}{\footnotesize \sym{*} \(p<0.10\), \sym{**} \(p<0.05\), \sym{***} \(p<0.01\)}\\
\end{tabular}
}

\label{1-incentives-beliefs}
\end{footnotesize}
\begin{tablenotes}
\begin{scriptsize}
\vspace{-2mm}
\item \textbf{Notes:} OLS, errors clustered at subject level. Dependent variable: indicator for sender choosing the false message. S's belief: sender's belief about party gap in receiver ratings. Subject controls: Gender, race, age, own party, education, and CRT score. Only 50-50 Priors: only observations where the receiver's prior is 1/2. Only includes questions in Block 1.
\end{scriptsize}
\vspace{5mm}
\end{tablenotes}
\end{threeparttable}
\end{center}

\clearpage

\begin{figure}[ht]
    \caption{CDF of Individual-Level False News by Senders' Incentives and Receivers' Party}
    \label{1-indiv-data}
    \centering
\includegraphics[width=.9\textwidth]{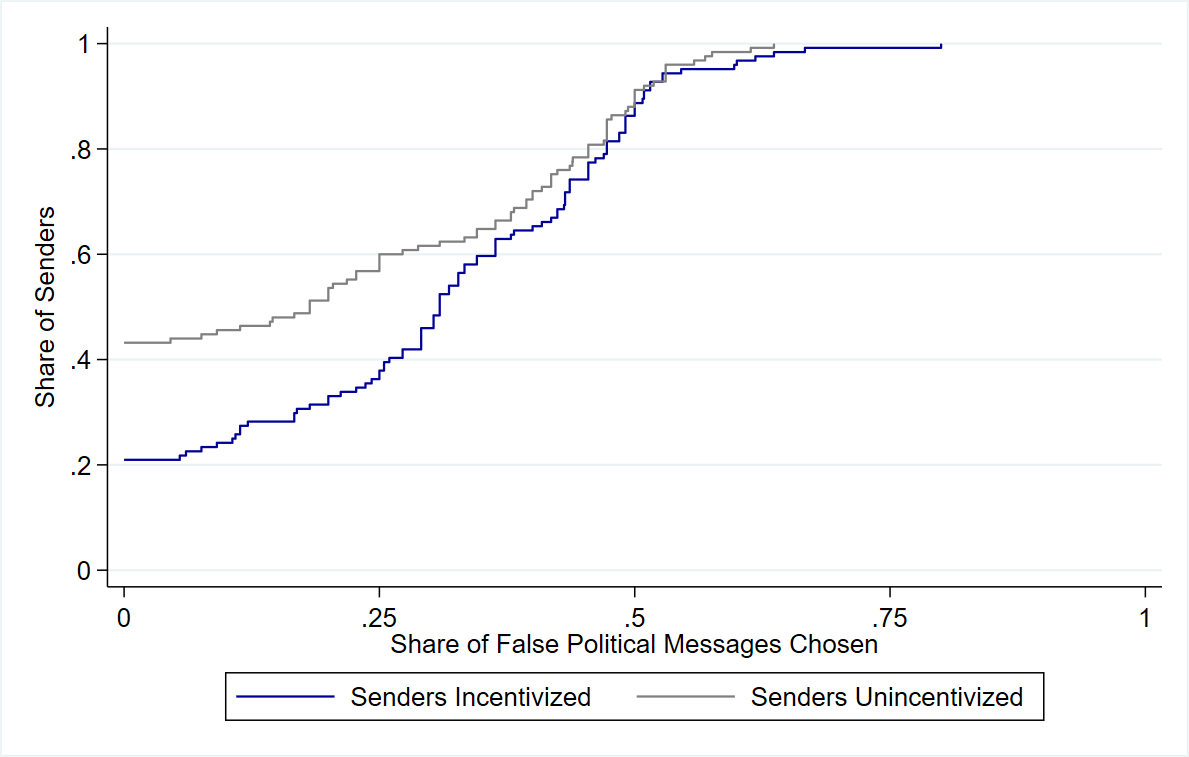}
\includegraphics[width=.9\textwidth]{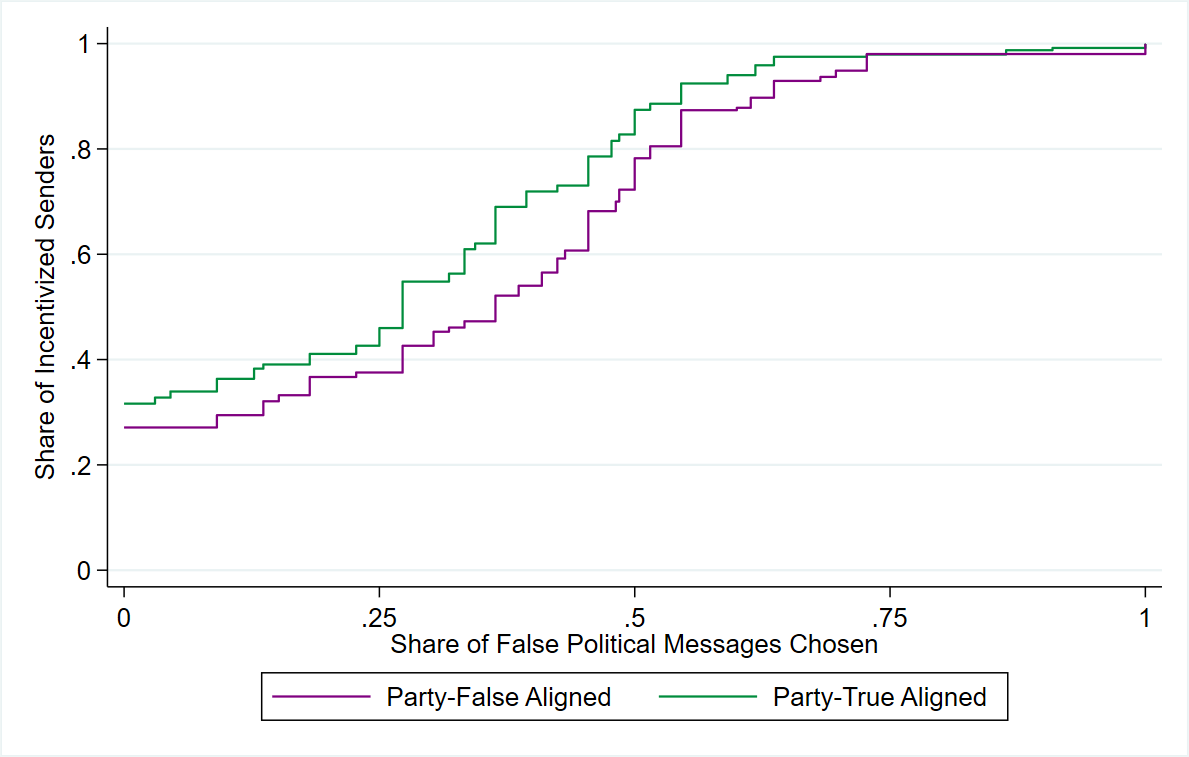}
\begin{threeparttable}
\begin{tablenotes}
\begin{scriptsize}
\vspace{-4mm}
\item \textbf{Notes:} CDF plots of the average share of messages chosen by senders. For instance, the top panel shows that half of incentivized senders send false messages at least 31 percent of the time. Party-True Aligned: indicator for the receiver's party being revealed and aligned with the true message. Party-False Aligned: indicator for the receiver's party being revealed and aligned with the false message. Both panels restrict to political questions in Block 1. The bottom panel restricts to senders who are incentivized and learn the party of the receiver.
\end{scriptsize}
\end{tablenotes}
\end{threeparttable}
\end{figure}

\clearpage

\begin{center}
\begin{threeparttable}[htbp!]
\begin{footnotesize}
\caption{The effect of competition incentives on choosing false messages}
{
\def\sym#1{\ifmmode^{#1}\else\(^{#1}\)\fi}
\begin{tabular}{l*{3}{c}}
\hline\hline
                    &\multicolumn{1}{c}{All Pol. Questions}         &\multicolumn{1}{c}{Party-False Aligned}         &\multicolumn{1}{c}{Interaction}         \\
\hline
Competition Incentives        &    0.027         &   -0.005         &    0.030         \\
          &  (0.037)         &  (0.051)         &  (0.039)         \\
Party-False Aligned x Incentivized &                  &                  &   -0.002         \\
      &                  &                  &  (0.027)         \\
Party-False Aligned x Unincentivized &                  &                  &    0.007         \\
    &                  &                  &  (0.026)         \\
Question FE         &$\checkmark$         &$\checkmark$         &$\checkmark$         \\
Round FE            &$\checkmark$         &$\checkmark$         &$\checkmark$         \\
Subject controls    &$\checkmark$         &$\checkmark$         &$\checkmark$         \\
All Questions       &$\checkmark$         &                  &$\checkmark$         \\
Only Party-False Aligned &                  &$\checkmark$         &                  \\
\hline
Observations        &    14214         &     4682         &    14214         \\
Subjects            &      251         &      225         &      251         \\
Mean                &    0.224         &    0.223         &    0.224         \\
\hline\hline
\multicolumn{4}{l}{\footnotesize Standard errors in parentheses}\\
\multicolumn{4}{l}{\footnotesize \sym{*} \(p<0.10\), \sym{**} \(p<0.05\), \sym{***} \(p<0.01\)}\\
\end{tabular}
}

\label{1-competition-effect}
\end{footnotesize}
\begin{tablenotes}
\begin{scriptsize}
\vspace{-2mm}
\item \textbf{Notes:} OLS, errors clustered at subject level. Dependent variable: indicator for sender choosing the false message. Subject controls: Gender, race, age, own party, education, and CRT score. Only includes questions in Block 1.
\end{scriptsize}
\vspace{5mm}
\end{tablenotes}
\end{threeparttable}
\end{center}

\clearpage

\begin{center}
\begin{threeparttable}[htbp!]
\begin{footnotesize}
\caption{Factors that lead senders with competition incentives to send false messages}
{
\def\sym#1{\ifmmode^{#1}\else\(^{#1}\)\fi}
\begin{tabular}{l*{3}{c}}
\hline\hline
                    &\multicolumn{1}{c}{Vs. Party-True Aligned}         &\multicolumn{1}{c}{Vs. No Info}         &\multicolumn{1}{c}{Vs. Neutral Topics}         \\
\hline
Party-False Aligned &    0.012         &    0.019         &    0.040         \\
                    &  (0.028)         &  (0.028)         &  (0.026)         \\
Prior-False Aligned &    0.167\sym{***}&    0.202\sym{***}&    0.166\sym{***}\\
                    &  (0.043)         &  (0.045)         &  (0.046)         \\
Question FE         &$\checkmark$         &$\checkmark$         &                  \\
Subject FE          &$\checkmark$         &$\checkmark$         &$\checkmark$         \\
Round FE            &$\checkmark$         &$\checkmark$         &$\checkmark$         \\
Vs. Party-True Aligned\hspace{20mm} &$\checkmark$         &                  &                  \\
Vs. No Info         &                  &$\checkmark$         &                  \\
Vs. Neutral Topics  &                  &                  &$\checkmark$         \\
\hline
Observations        &     4750         &     4891         &     4123         \\
Subjects            &      126         &      125         &      124         \\
Mean                &    0.230         &    0.238         &    0.222         \\
\hline\hline
\multicolumn{4}{l}{\footnotesize Standard errors in parentheses}\\
\multicolumn{4}{l}{\footnotesize \sym{*} \(p<0.10\), \sym{**} \(p<0.05\), \sym{***} \(p<0.01\)}\\
\end{tabular}
}

\label{among-competition}
\end{footnotesize}
\begin{tablenotes}
\begin{scriptsize}
\item \textbf{Notes:} OLS, errors clustered at subject level. Dependent variable: indicator for sender choosing the false message. Each column represents a different subset of the data, and Vs. lines indicate the comparison group. Party-False Aligned: indicator for the receiver's party being revealed and aligned with the false message. Prior-False Aligned: the receiver's prior belief that the incorrect answer is true. Party-True Aligned: indicator for the receiver's party being revealed and aligned with the true message. No Info: indicator for the receiver's party not being revealed. 
\end{scriptsize}
\end{tablenotes}
\end{threeparttable}
\end{center}

\clearpage



\begin{figure}[htb!]
    \caption{CDF of Individual-Level False News when Senders are Unincentivized}
    \label{1-indiv-data-unin}
    \centering
\includegraphics[width=\textwidth]{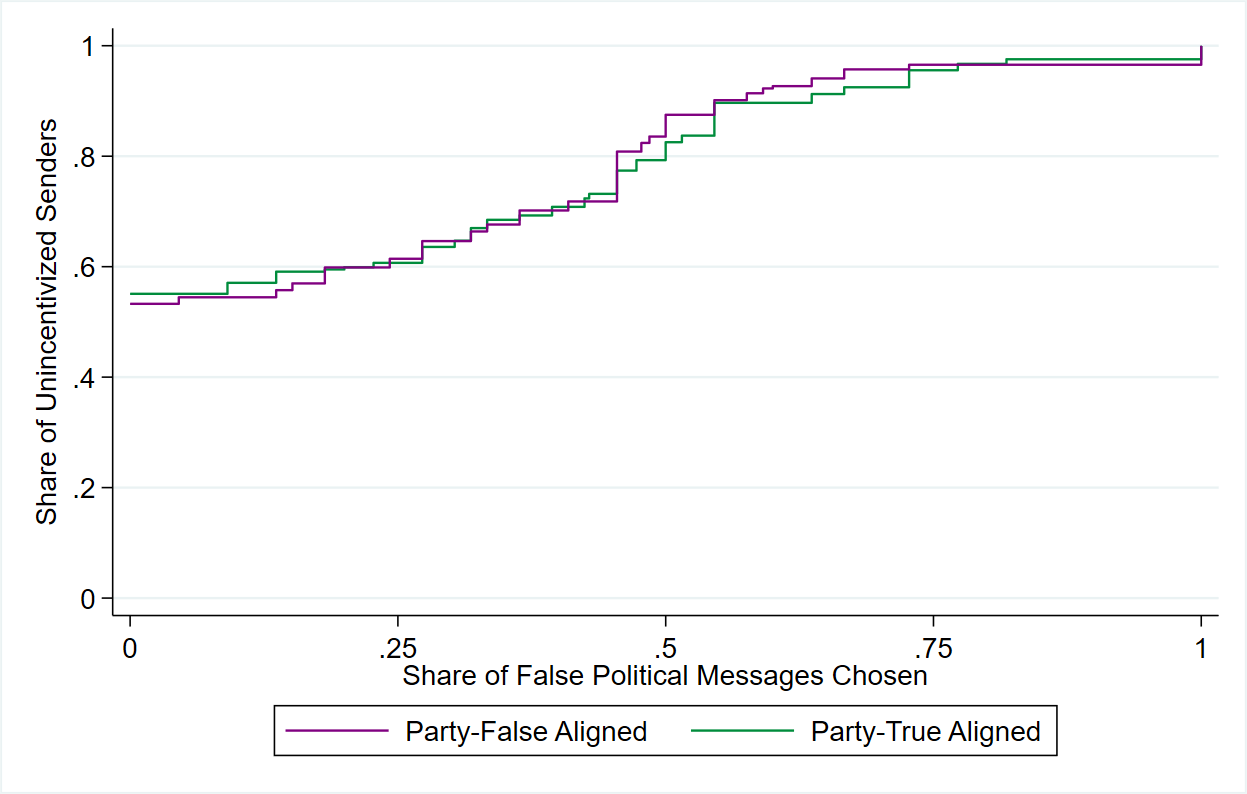}
\begin{threeparttable}
\begin{tablenotes}
\begin{scriptsize}
\vspace{-4mm}
\item \textbf{Notes:} CDF plot of the average share of messages chosen by senders. Party-True Aligned: indicator for the receiver's party being revealed and aligned with the true message. Party-False Aligned: indicator for the receiver's party being revealed and aligned with the false message. Data are restricted to political questions for which senders condition on receivers' priors. Senders are included if they are unincentivized and learn the party of the receiver.
\end{scriptsize}
\end{tablenotes}
\end{threeparttable}
\end{figure}

\clearpage

\begin{figure}[htb!]
    \caption{Alternative Specifications}
    \label{1-alt-specs}
    \centering
\includegraphics[width=\textwidth]{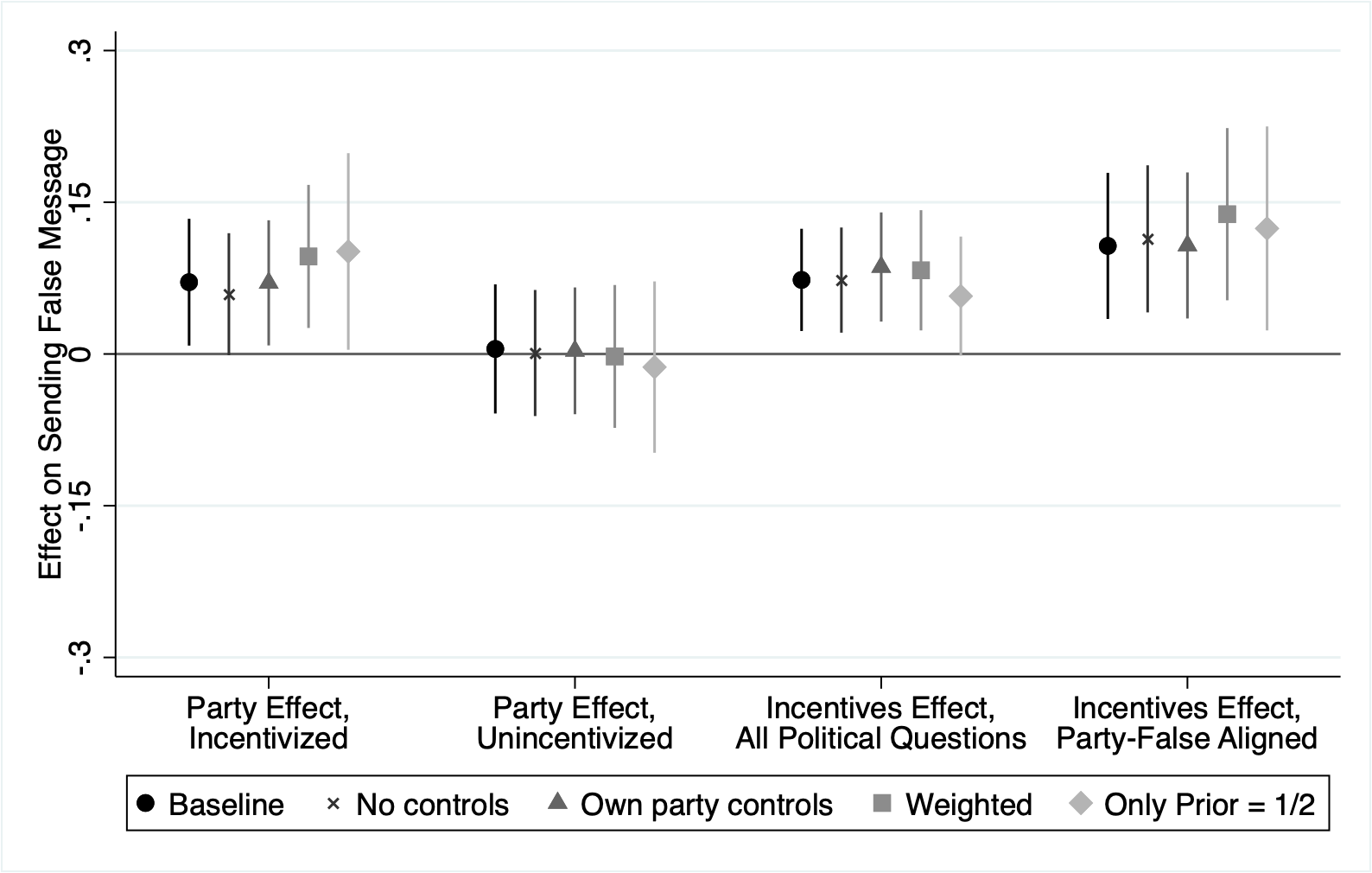}
\begin{threeparttable}
\begin{tablenotes}
\begin{scriptsize}
\vspace{-4mm}
\item \textbf{Notes:} Coefficients from alternative specifications. The first two categories correspond to the main treatment effects of the receiver's party on sending false messages in \Cref{among-incentivized} and \Cref{among-unincentivized}. The second two categories correspond to the main treatment effects of the sender's incentives on sending false messages in \Cref{incentives-effect}. All Political Questions: \Cref{incentives-effect}, column (2); Party-False Aligned: \Cref{incentives-effect}, column (4). No controls: No FE for round and question, no individual-level FE for within-subject tests and no demographic controls for between-subject tests. Own party controls: Inclusion of controls for sender's party alignment with the truth and sender's party alignment with the receiver. Weighted: Observations are weighted based on the frequency of senders being matched with receivers for each possible receiver's prior. For instance, if a sender in a given treatment is twice as likely to be matched with a receiver in their treatment group with prior 3/10 than they are to be matched with a receiver with prior 7/10, then their message choice given a receiver prior of 3/10 is weighted twice as much as their choice given a receiver prior of 7/10. Only Prior = 1/2: Observations restricted to senders' choices when faced with receivers with a prior of 1/2. 
\end{scriptsize}
\end{tablenotes}
\end{threeparttable}
\end{figure}

\clearpage

\section{Additional Tables and Figures for the Additional Experiment}

\begin{center}
\begin{threeparttable}
\begin{footnotesize}
\caption{Balance Table for Senders}
\def\sym#1{\ifmmode^{#1}\else\(^{#1}\)\fi}
\begin{tabular}{p{3.4cm}*{4}{c}}
\hline\hline
&\multicolumn{1}{c}{Incentivized}&\multicolumn{1}{c}{Unincentivized}&\multicolumn{1}{c}{Inc. vs. Uninc.}&\multicolumn{1}{p{2.2cm}}{\centering p-value}\\
\hline
Age  & 38.706 & 38.344 & 0.362 & 0.766 \\   
& (0.878) & (0.839) & (1.213) & \\
White & 0.767 & 0.788 & -0.021 & 0.578 \\   
& (0.027) & (0.026) & (0.037) & \\
Female  & 0.469 & 0.545 & -0.076 & 0.089 \\   
& (0.032) & (0.031) & (0.045) & \\
Education  & 15.261 & 15.243 & 0.019 & 0.920 \\   
& (0.132) & (0.130) & (0.185) & \\
CRT score  & 1.514 & 1.467 & 0.047 & 0.651 \\   
& (0.075) & (0.073) & (0.105) & \\
Party  & 0.551 & 0.443 & 0.108 & 0.015 \\   
& (0.032) & (0.031) & (0.044) & \\
Other's party revealed  & 0.659 & 0.664 & -0.005 & 0.791 \\   
 & (0.013) & (0.012) & (0.018) & \\
Others' have  & 0.492 & 0.488 & 0.005 & 0.856 \\   
party-truth aligned & (0.019) & (0.017) & (0.025) & \\
Self has  & 0.499 & 0.490 & 0.008 & 0.547 \\   
party-truth aligned & (0.009) & (0.010) & (0.014) & \\
\hline
\(N\)  & 1,470 & 1,529 & 2,999 & \\
\hline
\hline
\end{tabular}
\label{balance-table-s2}
\end{footnotesize}
\begin{tablenotes}
\item \begin{scriptsize} \textbf{Notes:} Standard errors in parentheses. Education is in years. CRT score is number of correct answers on the cognitive reflection task. Party is 1 if subject is Republican or Republican-leaning and 1/2 if subject is Independent (no lean). Party-truth alignment is defined in the main text. Party-truth alignment is restricted to observations where party is revealed.
\end{scriptsize}
\end{tablenotes}
\end{threeparttable}
\end{center}

\clearpage

\begin{center}
\begin{threeparttable}[htbp!]
\begin{footnotesize}
\caption{Factors that lead unincentivized senders to choose false computer messages}
{
\def\sym#1{\ifmmode^{#1}\else\(^{#1}\)\fi}
\begin{tabular}{l*{3}{c}}
\hline\hline
                    &\multicolumn{1}{c}{Vs. Party-True Aligned}         &\multicolumn{1}{c}{Vs. No Info}         &\multicolumn{1}{c}{Vs. Neutral Topics}         \\
\hline
Party-False Aligned &   -0.022         &    0.025         &   -0.065\sym{*}  \\
                    &  (0.028)         &  (0.034)         &  (0.033)         \\
Question FE         &$\checkmark$         &$\checkmark$         &                  \\
Subject FE          &$\checkmark$         &$\checkmark$         &$\checkmark$         \\
Round FE            &$\checkmark$         &$\checkmark$         &$\checkmark$         \\
Vs. Party-True Aligned\hspace{20mm} &$\checkmark$         &                  &                  \\
Vs. No Info         &                  &$\checkmark$         &                  \\
Vs. Neutral Topics  &                  &                  &$\checkmark$         \\
\hline
Observations        &      803         &      812         &      695         \\
Subjects            &      235         &      243         &      224         \\
Mean                &    0.241         &    0.230         &    0.253         \\
\hline\hline
\multicolumn{4}{l}{\footnotesize Standard errors in parentheses}\\
\multicolumn{4}{l}{\footnotesize \sym{*} \(p<0.10\), \sym{**} \(p<0.05\), \sym{***} \(p<0.01\)}\\
\end{tabular}
}

\label{2-among-unincentivized}
\end{footnotesize}
\begin{tablenotes}
\begin{scriptsize}
\item \textbf{Notes:} OLS, errors clustered at subject level. Dependent variable: indicator for sender choosing the false message. Each column represents a different subset of the data, and Vs. lines indicate the comparison group. Party-False Aligned: indicator for the receiver's party being revealed and aligned with the false message. Prior-False Aligned: the receiver's prior belief that the incorrect answer is true. Party-True Aligned: indicator for the receiver's party being revealed and aligned with the true message. No Info: indicator for the receiver's party not being revealed. 
\end{scriptsize}
\end{tablenotes}
\end{threeparttable}
\end{center}

\clearpage

\begin{center}
\begin{threeparttable}[htbp!]
\begin{footnotesize}
\caption{The effect of own party on messages}
{
\def\sym#1{\ifmmode^{#1}\else\(^{#1}\)\fi}
\begin{tabular}{l*{3}{c}}
\hline\hline
                    &\multicolumn{1}{c}{Incentivized}         &\multicolumn{1}{c}{Unincentivized}         &\multicolumn{1}{c}{Interaction}         \\
\hline
Own Party-False Aligned    &    0.070\sym{*}  &    0.190\sym{***}&    0.189\sym{***}\\
             &  (0.039)         &  (0.029)         &  (0.029)         \\
Other's Party-False Aligned &    0.155\sym{***}&   -0.023         &   -0.017         \\
             &  (0.045)         &  (0.027)         &  (0.028)         \\
Own Party-False Aligned x Incentivized    &                  &                  &   -0.108\sym{**} \\
&                  &                  &  (0.049)         \\
Other's Party-False Aligned x Incentivized &                  &                  &    0.165\sym{***}\\
&                  &                  &  (0.053)         \\
Question FE         &$\checkmark$         &$\checkmark$         &$\checkmark$         \\
Subject FE          &$\checkmark$         &$\checkmark$         &$\checkmark$         \\
Round FE            &$\checkmark$         &$\checkmark$         &$\checkmark$         \\
Incentivized subjects &$\checkmark$         &                  &$\checkmark$         \\
Unincentivized subjects\hspace{30mm} &                  &$\checkmark$         &$\checkmark$         \\
\hline
Observations        &      763         &      790         &     1553         \\
Subjects            &      225         &      231         &      456         \\
Mean                &    0.339         &    0.234         &    0.286         \\
\hline\hline
\multicolumn{4}{l}{\footnotesize Standard errors in parentheses}\\
\multicolumn{4}{l}{\footnotesize \sym{*} \(p<0.10\), \sym{**} \(p<0.05\), \sym{***} \(p<0.01\)}\\
\end{tabular}
}

\label{2-ownparty}
\end{footnotesize}
\begin{tablenotes}
\begin{scriptsize}
\vspace{-2mm}
\item \textbf{Notes:} OLS, errors clustered at subject level. Dependent variable: indicator for sender choosing the false message. Prior-False Aligned: the receiver's prior belief that the incorrect answer is true. Party-True Aligned: indicator for the receiver's party being revealed and aligned with the true message. 
\end{scriptsize}
\vspace{5mm}
\end{tablenotes}
\end{threeparttable}
\end{center}

\clearpage

\setcounter{page}{1}

\section{Online Appendix: Study Materials for the Primary Experiment}

\subsection{Question Wordings}
\label{1-question-wordings}

\begin{small}

\subsubsection*{Crime Under Trump}
\vspace{-1mm}

The Trump administration campaigned on tough-on-crime policies. Some people believe that the Trump administration's policies were effective at reducing violent crime, while others believe that his rhetoric provoked more violence.

This question asks how violent crime rates changed during the Trump administration. In 2016 (before Trump became president), the violent crime rate was 386.6 per 100,000 Americans.

In 2020 (at the end of Trump's presidency), do you think it is more likely that the violent crime rate was greater or less than [300 or 500] per 100,000 Americans?

\vspace{1mm}

\textit{Correct answer: 366.7 per 100,000}

\textit{Source linked on results page: \url{http://bit.ly/us-crime-rate}}

\subsubsection*{Undocumented Immigrants}
\vspace{-1mm}

The U.S. has seen a sharp rise in the share of undocumented immigrants over the past several years. Some people believe that undocumented immigrants are more likely to commit violent crime, while others believe that undocumented immigrants are less likely to commit violent crimes.

Texas is the only state that directly compares crime rates for US-born citizens to undocumented immigrants, and provided felony data from 2012-2018. During this time period, the felony violent crime rate was 213 per 100,000 U.S. citizens.

This question asks about the felony violent crime rate for undocumented immigrants. Do you think it is more likely that this rate was greater or less than [90 or 213] per 100,000?

\vspace{1mm}

\textit{Correct answer: 96.2 per 100,000}

\textit{Source linked on results page: \url{http://bit.ly/crime-by-immigrant-status}}

\subsubsection*{Racial Discrimination}
\vspace{-1mm}

In the United States, white Americans have higher salaries than black Americans on average. Some people attribute these differences in income to differences in education, training, and culture, while others attribute them more to racial discrimination.

In a study, researchers sent fictitious resumes to respond to thousands of help-wanted ads in newspapers. The resumes sent had identical skills and education, but the researchers gave half of the (fake) applicants stereotypically White names such as Emily Walsh and Greg Baker, and gave the other half of the applicants stereotypically Black names such as Lakisha Washington and Jamal Jones.

This question asks how the callback rates differed between White- and Black-sounding names. 9.65 percent of the applicants with White-sounding names received a call back. Do you think it is more likely that the percent of the applicants with Black-sounding names who received a call back was greater or less than [5.0 or 8.5] percent?

\vspace{1mm}

\textit{Correct answer: 6.45 percent}

\textit{Source linked on results page: \url{http://bit.ly/labor-market-discrimination}}

\subsubsection*{Media Bias}
\vspace{-1mm}

Some people believe that the media is filled with Democrats and unfairly biased towards the Democratic Party, while some believe the media is more balanced, and others believe it is biased towards Republicans.

This question asks whether journalists are significantly more likely to be Democrats than Republicans.

A representative sample of journalists were asked about their party affiliation. Compared to the number of Republicans, do you think it is more likely that the number of journalists who said they were Democrats was greater or less than [2 or 5] times as much?

\vspace{1mm}

\textit{Correct answer: 4 times as much}

\textit{Source linked on results page: \url{http://bit.ly/journalist-political-affiliation}}

\subsubsection*{COVID-19 Restrictions}
\vspace{-1mm}

In the face of the coronavirus pandemic, some places mandated strict lockdowns, while other places allowed for more activity and opened up sooner. This question asks how effective lockdowns were at preventing the spread of the coronavirus.

A recent study estimated how cases would have changed during the early stages of the pandemic if all areas implemented stay-at-home orders on March 17, 2020.

This question asks about the percent reduction in cases by April 30, 2020 if all areas implemented stay-at-home orders on March 17, 2020. Do you think it is more likely that this reduction was greater or less than [10 or 50] percent?

\vspace{1mm}

\textit{Correct answer: 19.5 percent}

\textit{Source linked on results page: \url{http://bit.ly/covid-restrictions-effect}}

\subsubsection*{Gun Laws}
\vspace{-1mm}

The United States has a homicide rate that is much higher than other wealthy countries. Some people attribute this to the prevalence of guns and favor stricter gun laws, while others believe that stricter gun laws will limit Americans' Second Amendment rights without reducing homicides very much.

After a mass shooting in 1996, Australia passed a massive gun control law called the National Firearms Agreement (NFA). The law illegalized, bought back, and destroyed almost one million firearms by 1997, mandated that all non-destroyed firearms be registered, and required a lengthy waiting period for firearm sales.

Democrats and Republicans have each pointed to the NFA as evidence for/against stricter gun laws. In the five years before the NFA (1991-1996), there were 320 homicides per year in Australia. In the five years after the NFA (1998-2003), do you think it is more likely that the average number of homicides in Australia was greater or less than [220 or 320] per year?

\vspace{1mm}

\textit{Correct answer: 318.6 per year}

\textit{Source linked on results page: \url{http://bit.ly/australia-homicide-rate} and \url{http://bit.ly/impact-australia-gun-laws.}}

\subsubsection*{Unemployment Rate}
\vspace{-1mm}

Some people believe that Donald Trump's policies improved the jobs situation in the United States, while others believe that his policies hindered employment.

This question asks whether the unemployment rate increased or decreased during the Trump administration as compared to the end of the Obama administration.

In the last two years of the Obama administration (Jan 2015-Jan 2017), the average unemployment rate was 5.1 percent. Do you think it is more likely that the average unemployment rate during the Trump administration was greater or less than [3.2 or 5.1] percent?

\vspace{1mm}

\textit{Correct answer: 5.04 percent}

\textit{Source linked on results page: \url{	http://bit.ly/unemployment-rate-data}}

\subsubsection*{Wage Growth}
\vspace{-1mm}

Some people believe that the Trump administration did a better job at increasing wages for most Americans, and some people believe that the Obama administration did a better job of wage growth.

In the last two years of the Obama administration (Jan 2015-Jan 2017), the median growth in Americans' wages was 3.28 percent on average.

Do you think it is more likely that the average median growth in Americans' wages during the Trump administration was greater or less than [3.28 or 4] percent?

\vspace{1mm}

\textit{Correct answer: 3.49 percent}

\textit{Source linked on results page: \url{http://bit.ly/median-wage-growth}}

\subsubsection*{Center of the US}
\vspace{-1mm}

The U.S. National Geodetic Survey approximated the geographic center of the continental United States. (This excludes Alaska and Hawaii, and U.S. territories.)

This question asks how far North the U.S. is located. For reference, the continental U.S. lies in the Northern Hemisphere, the Equator is 0 degrees North, and the North Pole is 90 degrees North.

Do you think it is more likely that this geographic center is greater or less than [30 or 45] degrees North?

\vspace{1mm}

\textit{Correct answer: 39.833 degrees North}

\textit{Source linked on results page: \url{http://bit.ly/center-of-the-us}}

\subsubsection*{Random Number}
\vspace{-1mm}

A computer randomly generated a number between 0 and 100, decimals included. What number do you think the computer chose?

As a reminder, it is in your interest to guess an answer that is close to the computer's choice, even if you don't perfectly guess it.

Do you think it is more likely that this number is greater or less than [40 or 60]?

\vspace{1mm}

\textit{Correct answer: 33.54026}

\subsubsection*{Performance on a CRT Task}
\vspace{-1mm}

Previously in this study, you were asked three quiz questions that some people use as a measure of cognitive ability. 

At the end of the study, your score on this test will be compared to the scores among all participants. This question asks you to predict how your score compared to others.

Do you think it is more likely that your score was greater or less than the average score?

\vspace{1mm}

\textit{The average score was between 1 and 2, so subjects who scored 2 or 3 scored greater than the average, and subjects who scored 0 or 1 scored less than the average.}

\subsubsection*{Quote from Biden: Visas and Immigrants}
\vspace{-1mm}

In 2021, Joe Biden said that there are ``over 11 million undocumented folks -- the vast majority are here overstaying visas.''

Do you think this statement is accurate or inaccurate?

\vspace{1mm}

\textit{Correct answer: Inaccurate}

\textit{Source linked on results page: \url{https://bit.ly/undocumented-mostly-visas}}

\subsubsection*{Quote from Biden: White Supremacists}
\vspace{-1mm}

In 2020, Joe Biden said that ``[Donald Trump's] FBI chief has said the greatest domestic threat to terrorism are white supremacists.''

Do you think this statement is accurate or inaccurate?

\vspace{1mm}

\textit{Correct answer: Accurate}

\textit{Source linked on results page: \url{https://bit.ly/white-supremacists-threat}}

\subsubsection*{Quote from Trump: Poverty Rates}
\vspace{-1mm}

In 2018, Donald Trump said that ``The poverty rates for African Americans and Hispanic Americans ... it's been incredible, they've all reached their lowest levels.''

Do you think this statement is accurate or inaccurate?

\vspace{1mm}

\textit{Correct answer: Accurate}

\textit{Source linked on results page: \url{https://bit.ly/poverty-rates-black-hispanic}}

\subsubsection*{Quote from Trump: Illegal Immigration}
\vspace{-1mm}

In 2021, Donald Trump said that there has been ``a massive flood of illegal immigration into our country, the likes of which we have never seen before.''

Do you think this statement is accurate or inaccurate?

\vspace{1mm}

\textit{Correct answer: Inaccurate}

\textit{Source linked on results page: \url{https://bit.ly/record-illegal-immigration}}

\subsubsection*{Attention Check: Current Year}
\label{comprehension-question}
\vspace{-1mm}
In 1776 our fathers brought forth, upon this continent, a new nation, conceived in Liberty, and dedicated to the proposition that all men are created equal. What is the year right now?

This is not a trick question and the first sentence is irrelevant; this is a check to make sure you are paying attention. If you get this question incorrect, you will not be eligible to receive a bonus payment.

\vspace{1mm}

\textit{Correct answer: 2021.}

\textit{Source linked on results page: \url{http://bit.ly/what-year-is-it}}

\end{small}

\newpage

\subsection{Screenshots}
\label{screenshots1}






\vspace{5mm}
\begin{figure}[h]
\caption{Overview}
\begin{center}
\includegraphics[width = .9\textwidth]{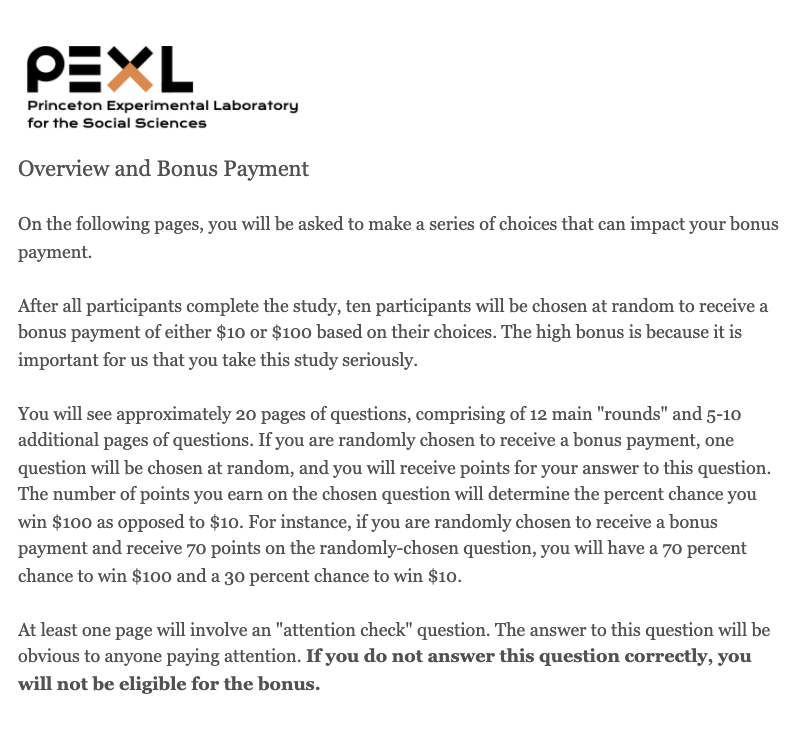}
\end{center}
\end{figure}

\clearpage

\begin{figure}
\caption{Demographics}
\begin{center}
\includegraphics[width = .84\textwidth]{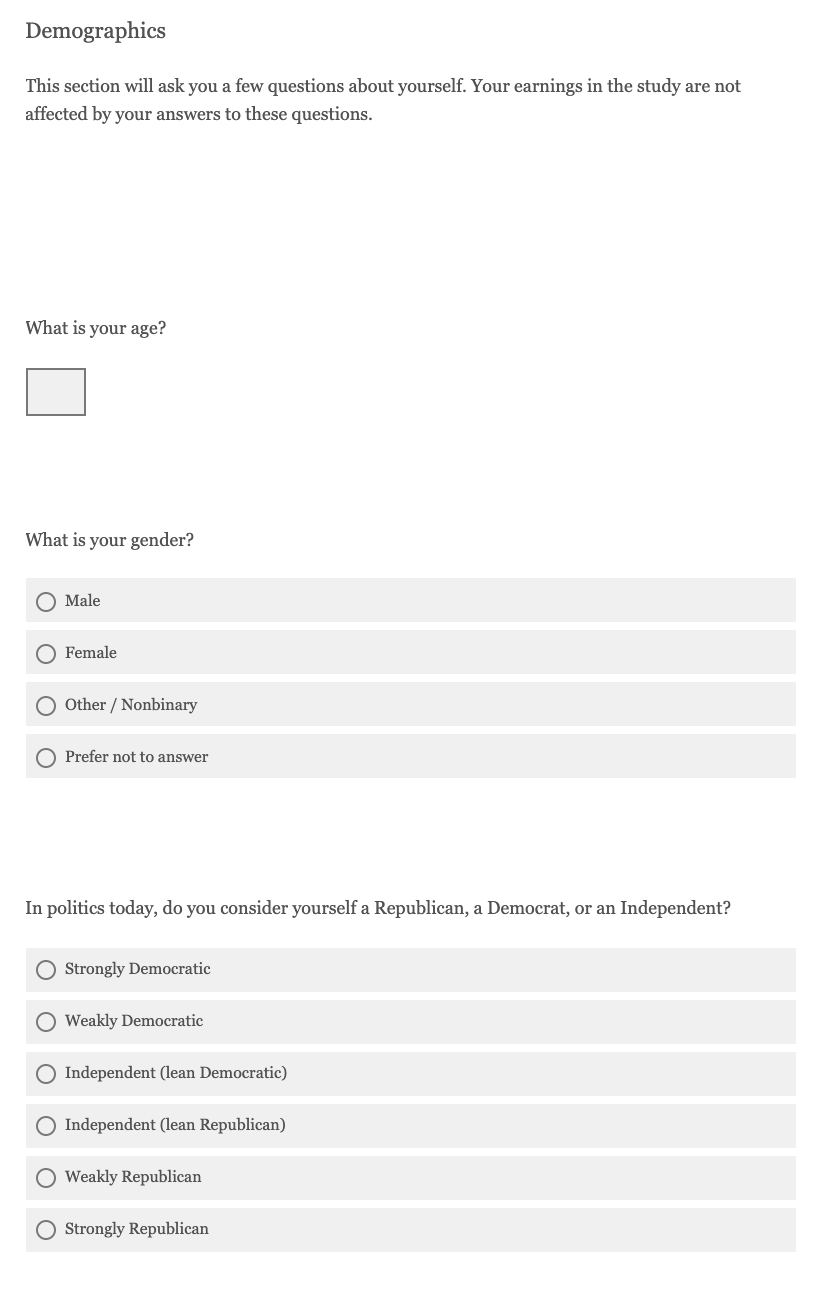}
\end{center}
\end{figure}

\clearpage

\begin{center}
\includegraphics[width = .84\textwidth]{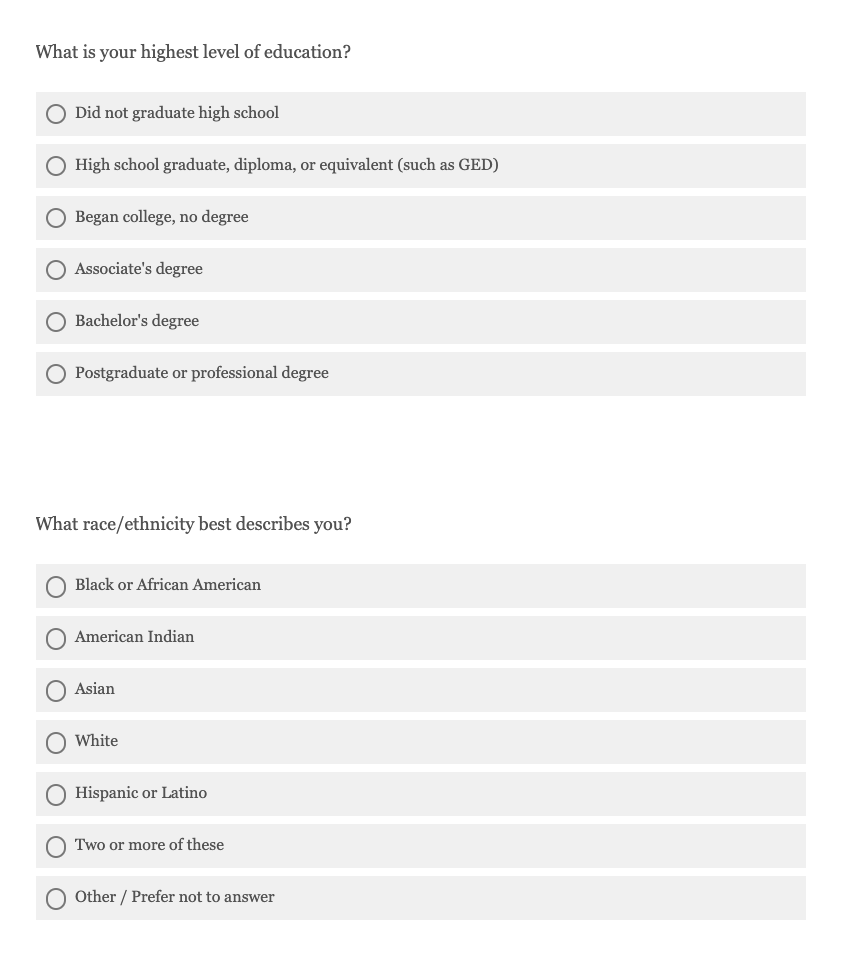}
\end{center}

\clearpage

\begin{figure}
\caption{Cognitive reflection task}
\begin{center}
\includegraphics[height = .92\textheight]{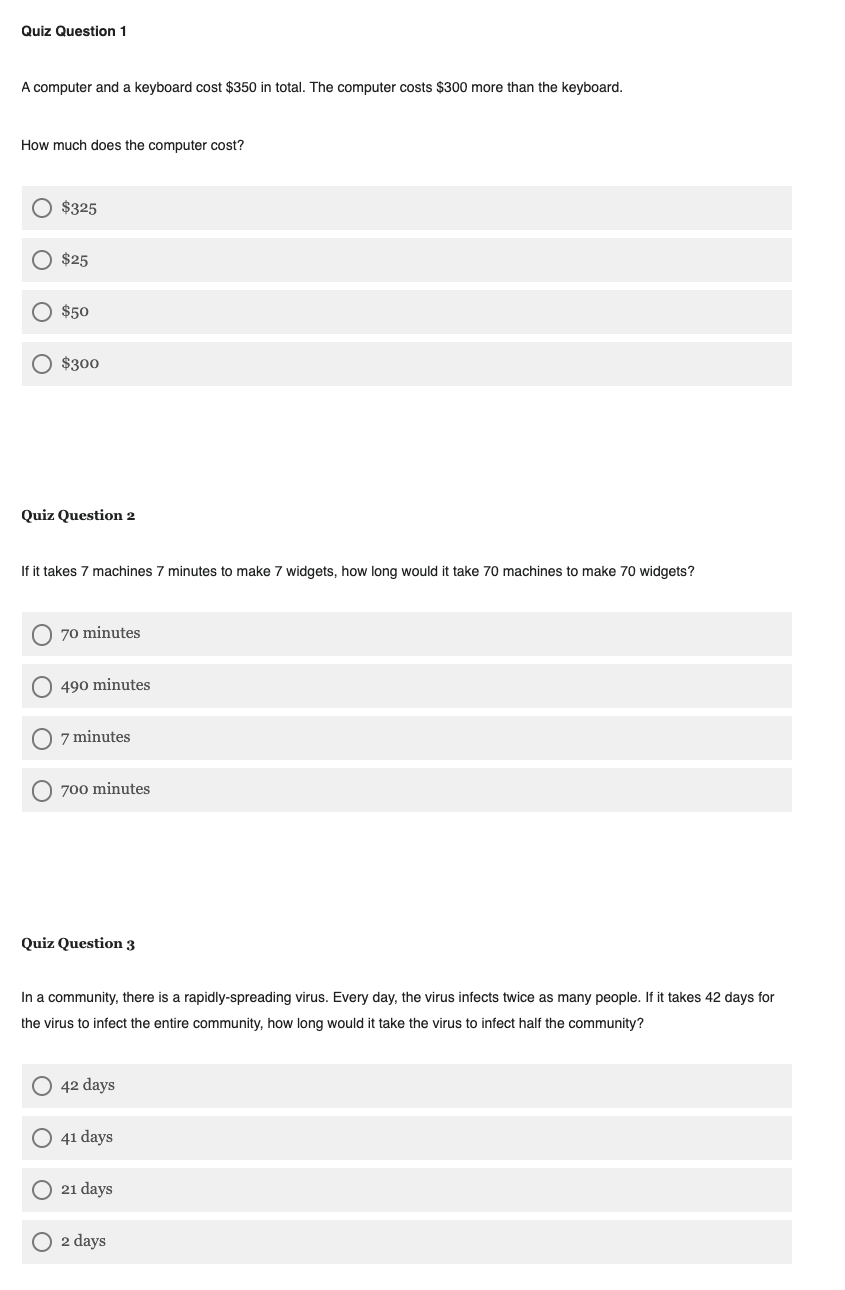}
\end{center}
\end{figure}

\clearpage

\begin{figure}
\caption{Overview for practice questions}
\begin{center}
\includegraphics[width = \textwidth]{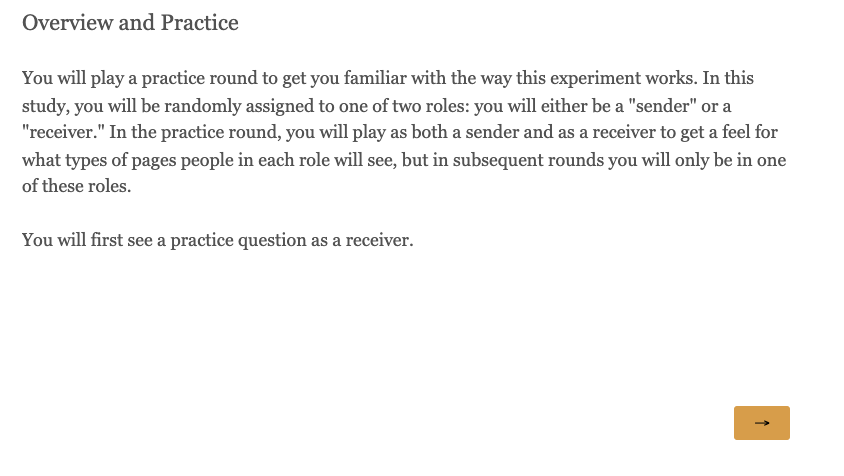}
\end{center}
\end{figure}

\clearpage

\begin{figure}
\caption{Instructions for receiver questions}
\begin{center}
\includegraphics[width = .89\textwidth]{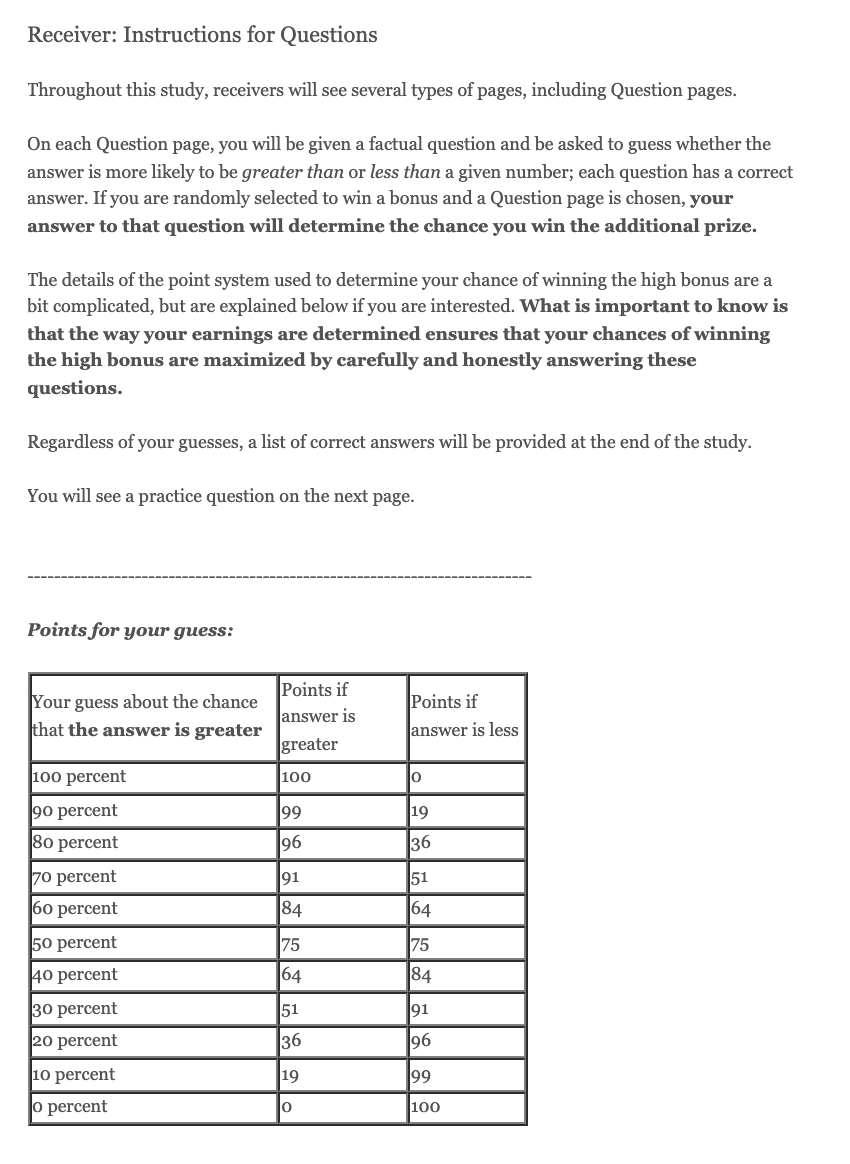}
\end{center}
\hrule
\vspace{3mm}

\noindent \footnotesize{The practice question is omitted since it looks the same as questions in the main rounds.}
\end{figure}

\clearpage

\begin{figure}
\caption{Practice receiver question}
\begin{center}
\includegraphics[width = \textwidth]{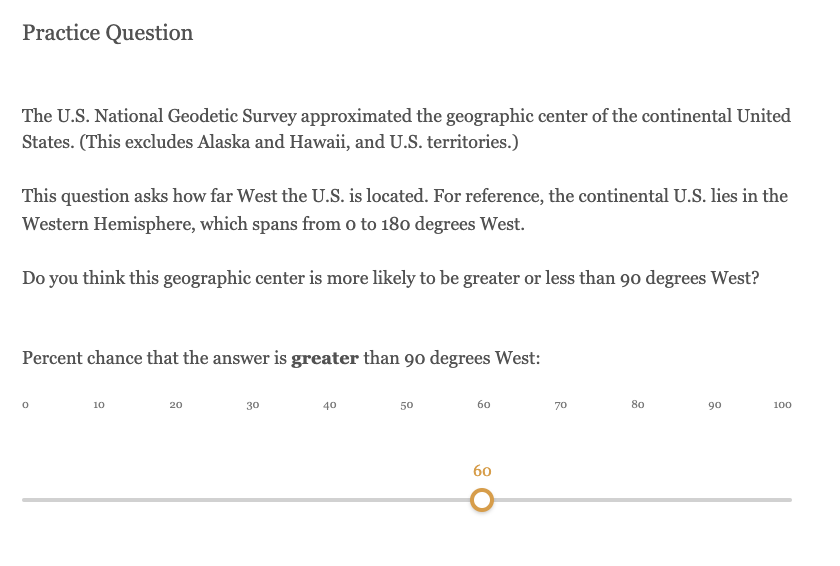}
\end{center}

\end{figure}

\clearpage

\begin{figure}
\caption{Instructions for sending messages: Incentivized treatment}
\begin{center}
\includegraphics[width = \textwidth]{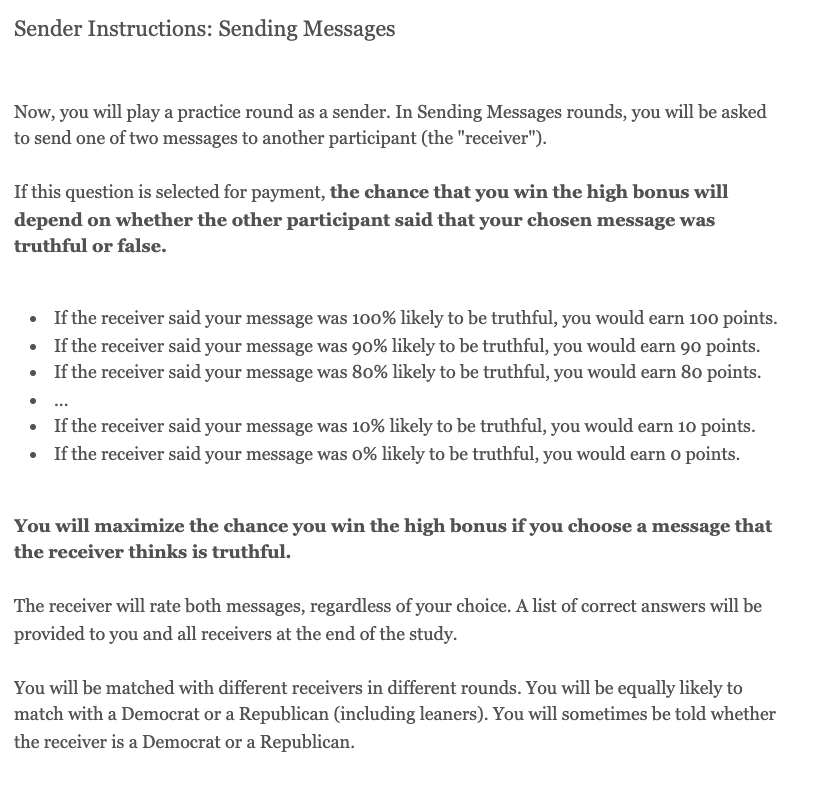}
\end{center}
\hrule
\vspace{5mm}

\small{Subjects in the unincentivized treatment are instead told: \textit{``\textbf{These questions will not affect your bonus payment.}''} and are not given the point system.}

\vspace{5mm} 
\small{Subjects in the competition treatment are instead told that \textit{``If this question is selected for payment, \textbf{the chance that you win the high bonus will depend on whether the receiver said that you were more likely to have sent more truthful messages over the course of the study.}''}}
\end{figure}

\clearpage

\begin{figure}
\caption{Practice sending messages: Incentivized treatment}
\begin{center}
\includegraphics[width = \textwidth]{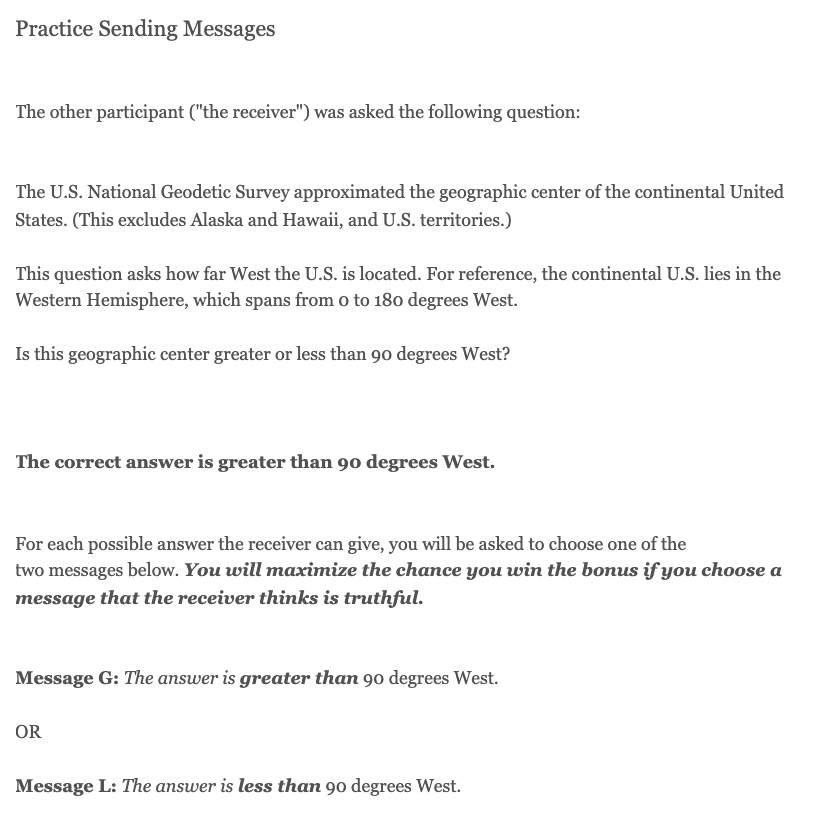}
\end{center}
\end{figure}

\clearpage

\begin{center}
\includegraphics[height = \textheight]{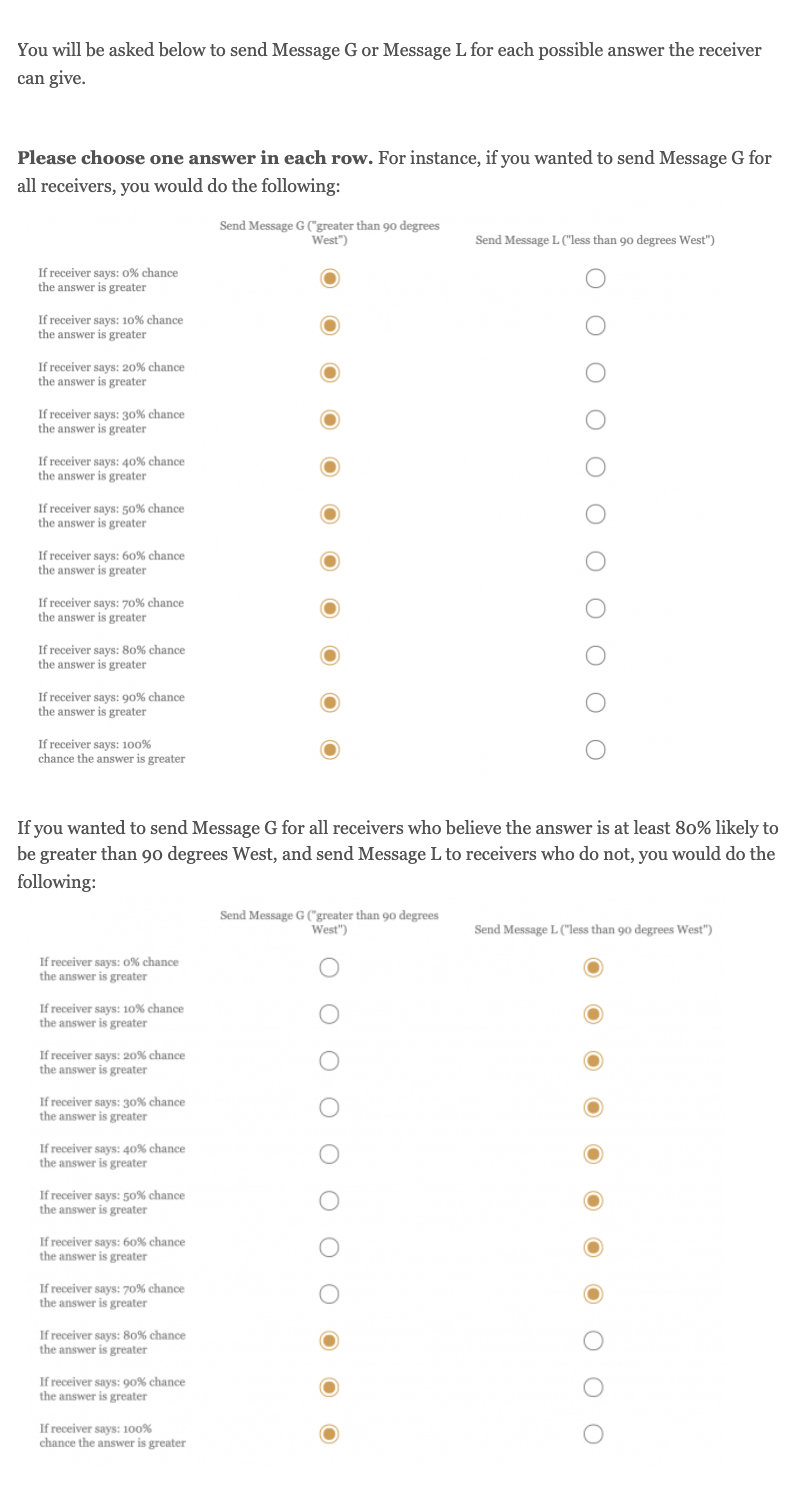}
\end{center}

\clearpage

\begin{center}
\includegraphics[width = \textwidth]{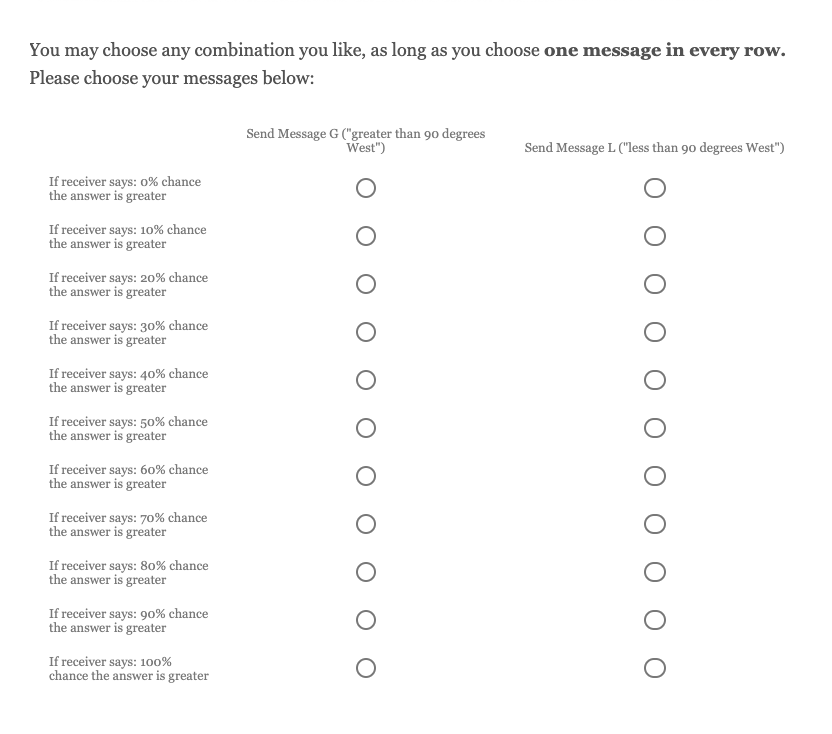}
\end{center}

\hrule
\vspace{5mm}

\noindent \footnotesize{Subjects in the unincentivized treatment are instead told: \textit{``\textbf{Your chance of winning the bonus is not affected by how you answer this question.}''}}

\vspace{5mm} 
\noindent \footnotesize{Subjects in the competition treatment are instead told that \textit{``\textbf{The receiver will predict, based on your message and another sender's message on this question, which sender sent more truthful messages over the course of the experiment. You will maximize the chance you win the bonus if the receiver believes that you sent more truthful messages.}''}}

\clearpage

\begin{figure}
\caption{Instructions for rating messages}
\begin{center}
\includegraphics[width = .97\textwidth]{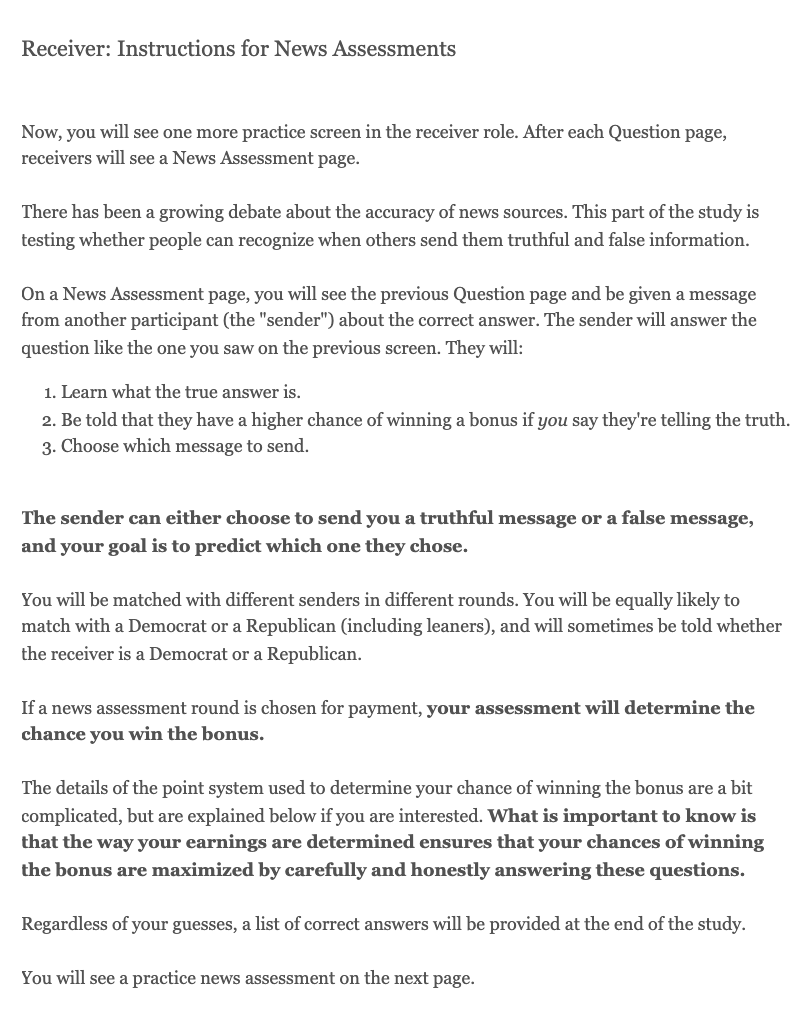}
\end{center}
\hrule
\vspace{3mm}

\noindent \footnotesize{The practice question is omitted since it looks the same as questions in the main rounds.}
\end{figure}

\clearpage

\begin{figure}
\caption{Treatment revelation page}
\begin{center}
\includegraphics[width = \textwidth]{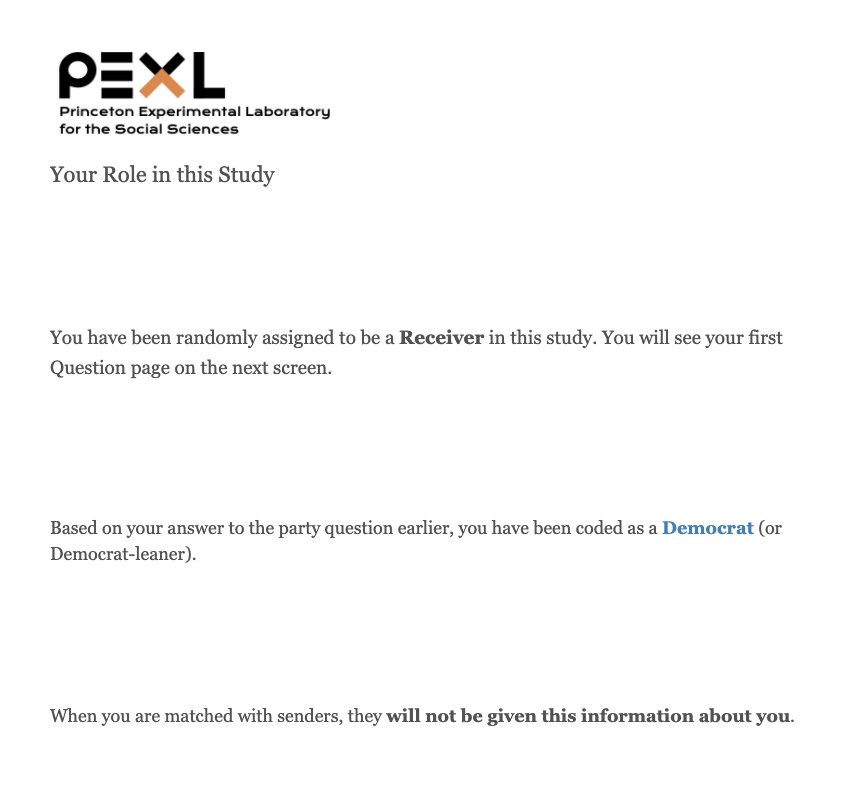}
\end{center}
\end{figure}

\clearpage

\begin{figure}
\caption{Receiver: Prior beliefs}
\begin{center}
\includegraphics[width=.9\textwidth]{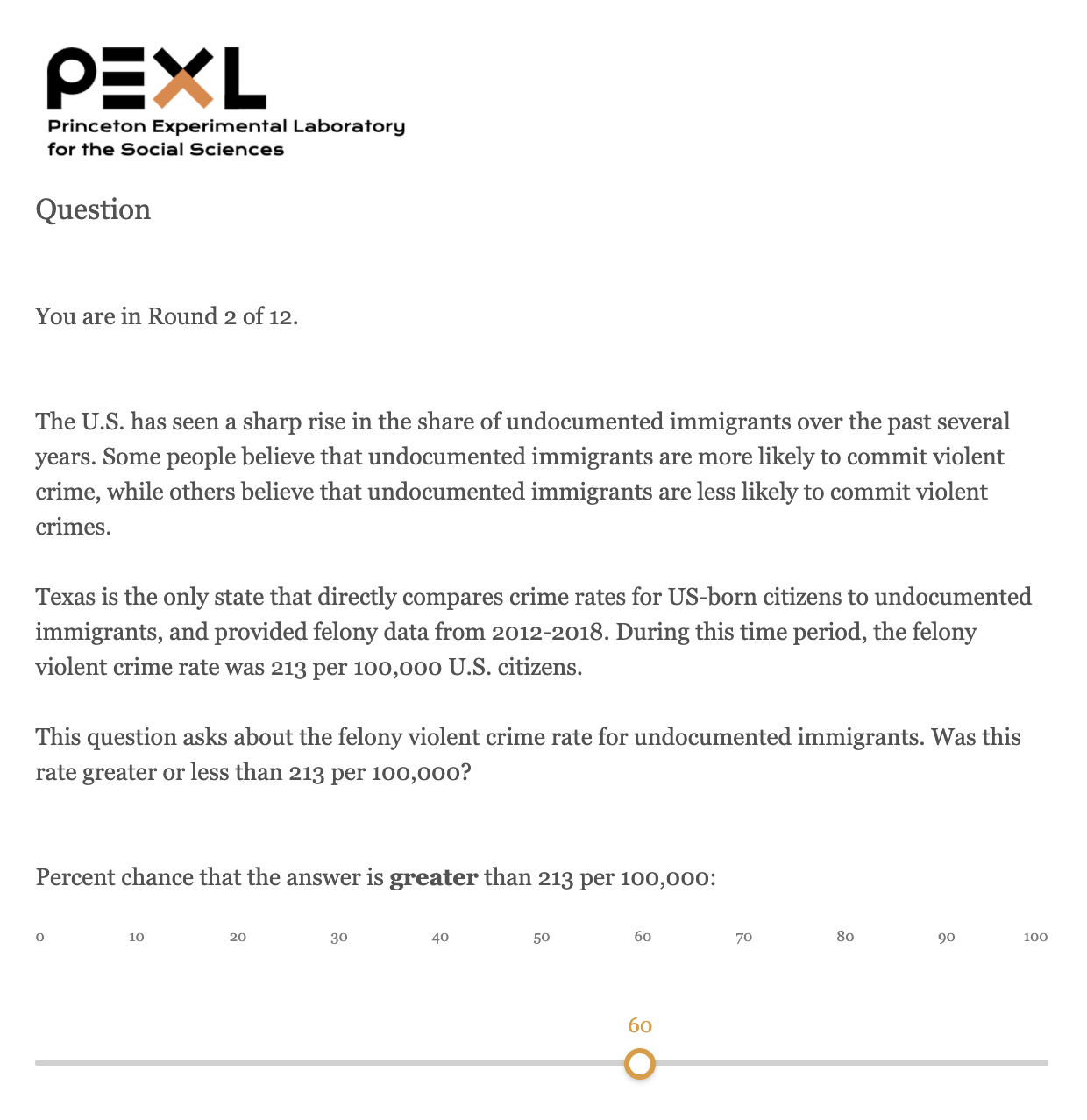}
\end{center}
\end{figure}

\clearpage

\begin{figure}
\caption{Receiver: News ratings}
\begin{center}
\includegraphics[width=.85\textwidth]{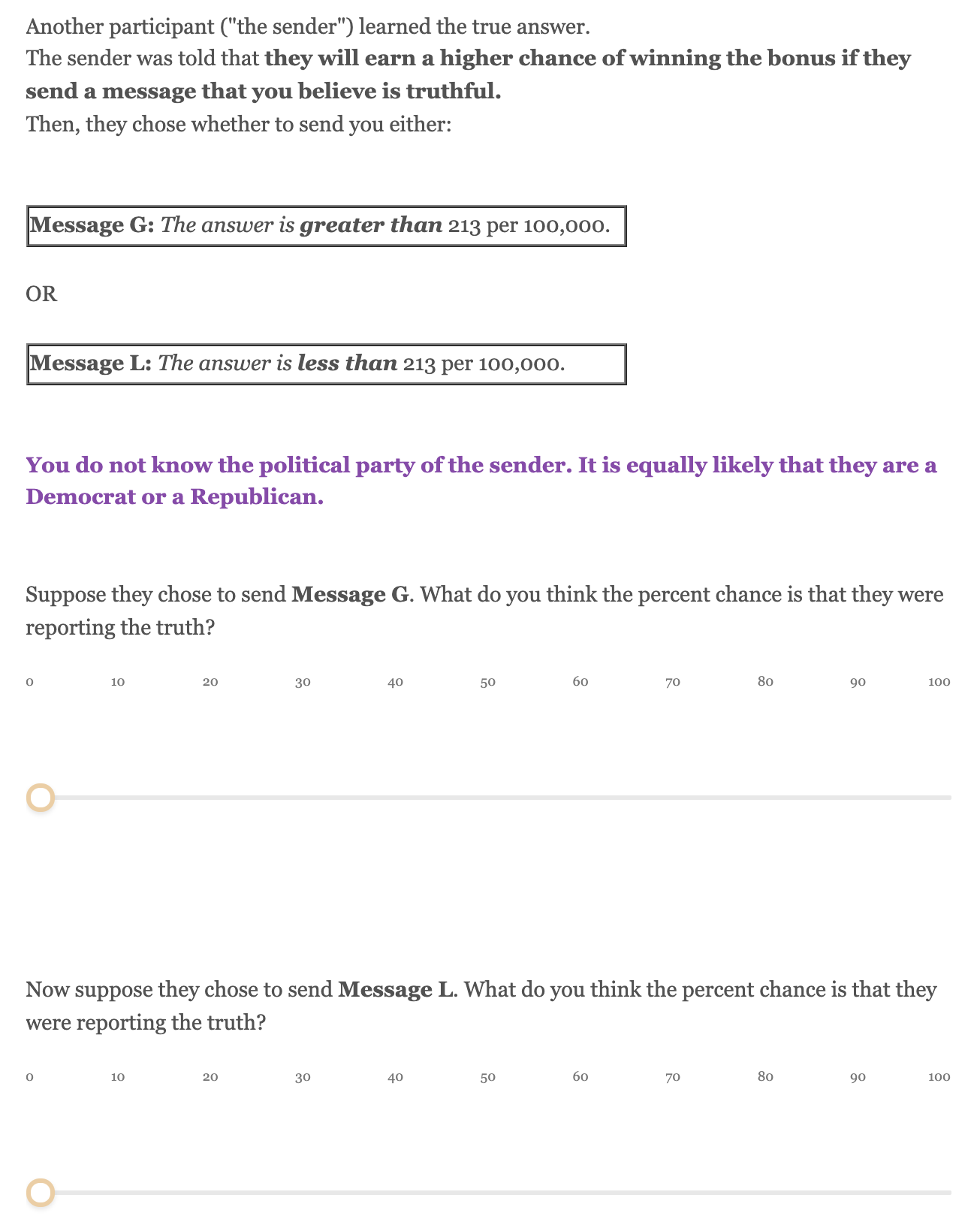}
\end{center}

\vspace{2mm}
\hrule
\vspace{5mm}

\noindent \footnotesize{Receivers in the unincentivized treatment do not see the line about the sender's bonus.}

\vspace{5mm} 
\noindent \footnotesize{Receivers in the competition treatment are instead asked: \textit{Now suppose that \textbf{Sender X} chose to send \textbf{Message G} and that \textbf{Sender Y} chose to send \textbf{Message L}. What do you think is the percent chance is that \textbf{Sender X} sent more truthful signals over the course of the experiment?}}
\end{figure}

\clearpage

\begin{figure}
\caption{Sender: Sending messages}
\begin{center}
\includegraphics[width=.95\textwidth]{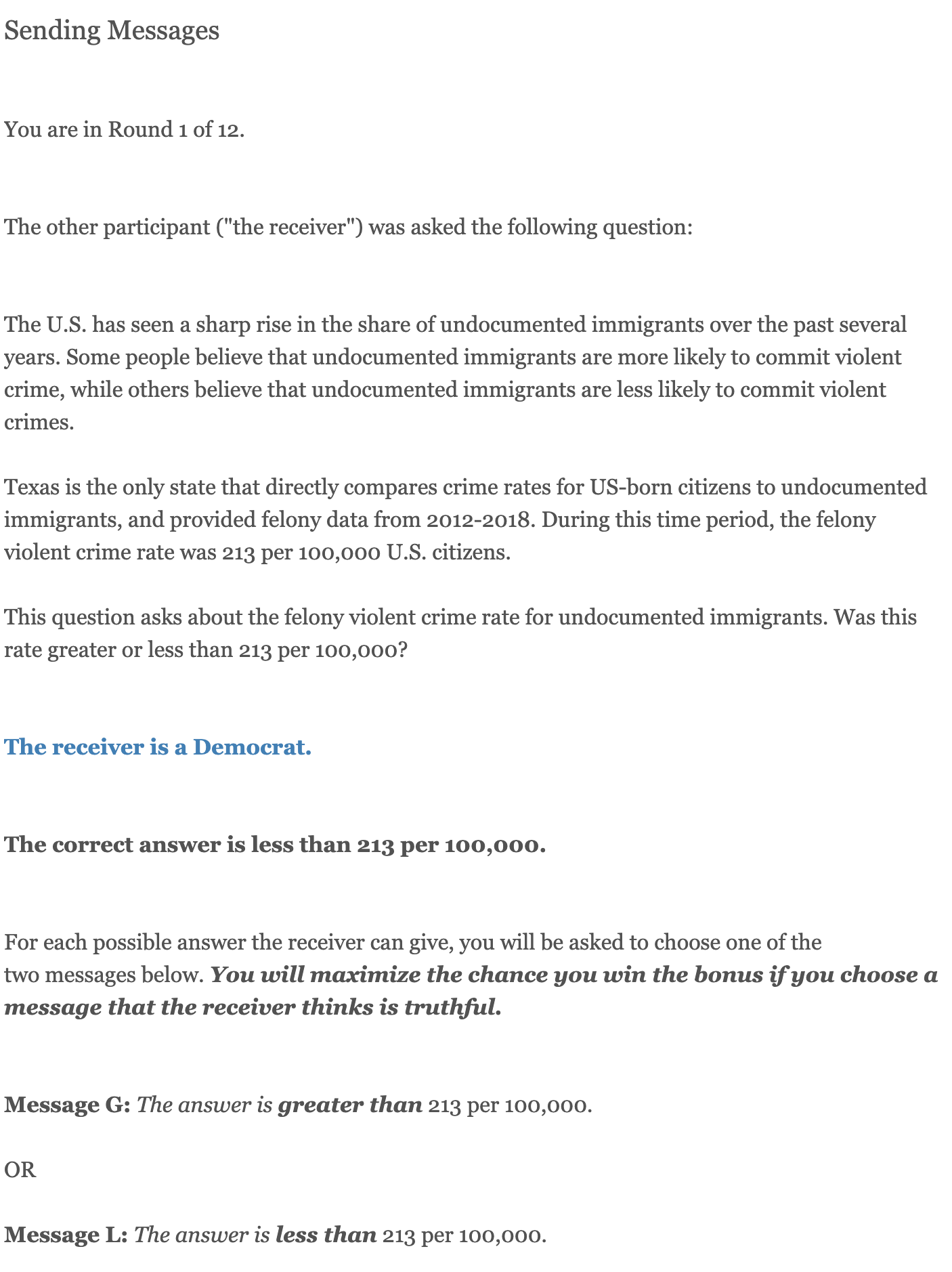}
\end{center}

\end{figure}

\clearpage
\begin{center}
\includegraphics[width=.82\textwidth]{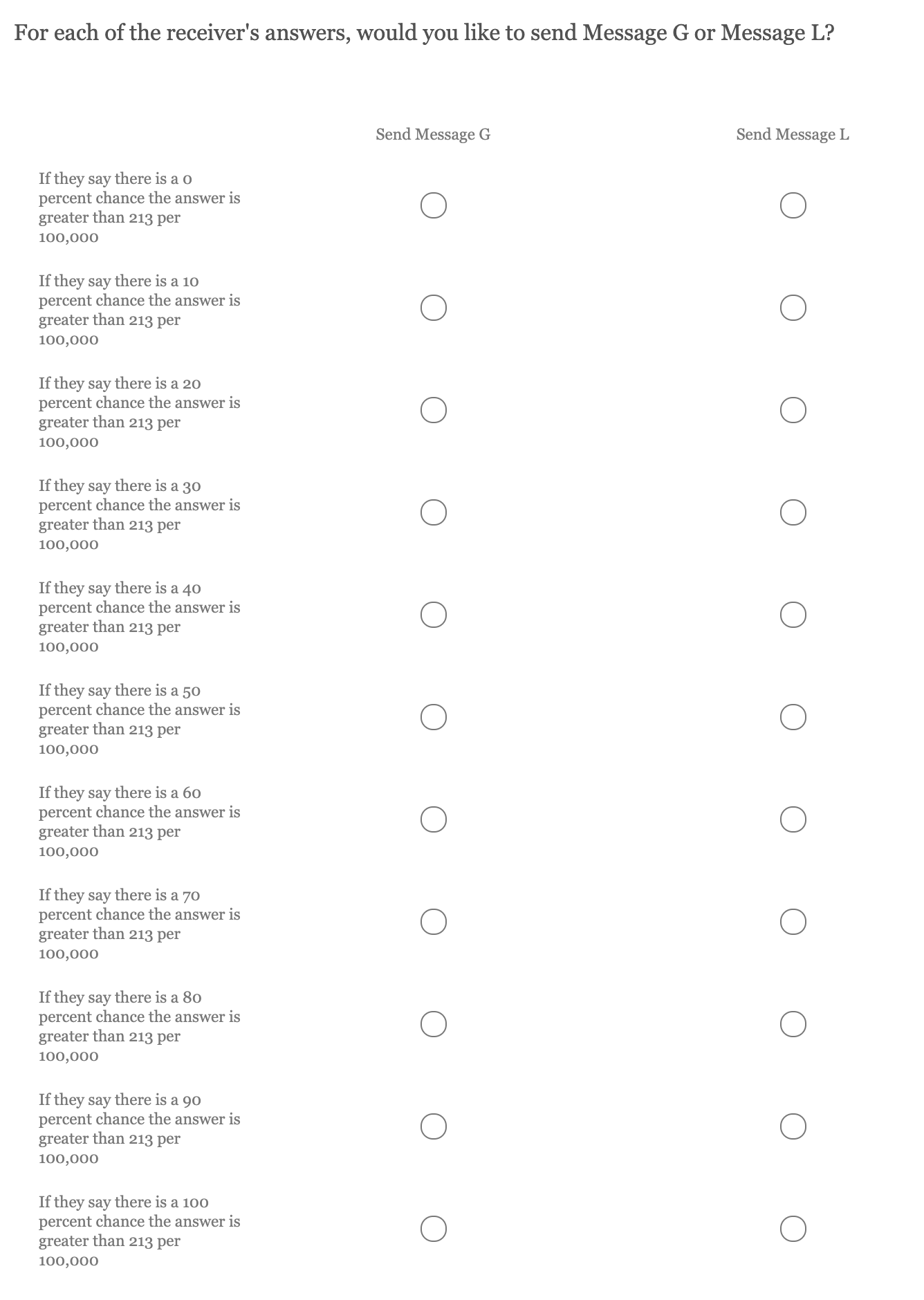}
\end{center}
\hrule
\vspace{2mm}

\noindent \footnotesize{Senders in the unincentivized treatment are instead told: \textit{``\textbf{Your chance of winning the bonus is not affected by how you answer this question.}''}}

\vspace{5mm} 
\noindent \footnotesize{Senders in the competition treatment are instead told that \textit{``\textbf{The receiver will predict, based on your message and another sender's message on this question, which sender sent more truthful messages over the course of the experiment. You will maximize the chance you win the bonus if the receiver believes that you sent more truthful messages.}''}}

\clearpage

\begin{figure}
\caption{Attention check}
\begin{center}
\includegraphics[width = \textwidth]{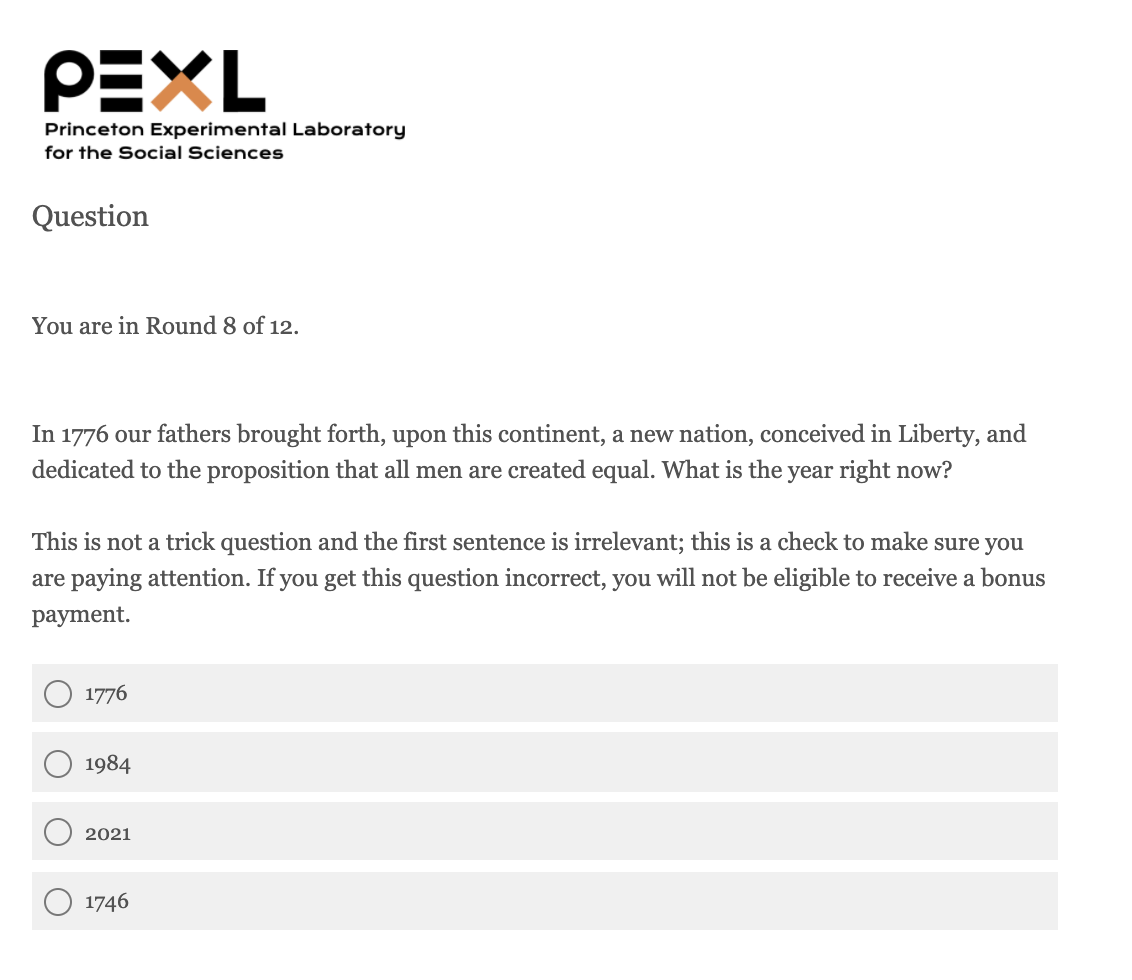}
\end{center}
\end{figure}

\clearpage

\begin{figure}
\caption{Receiver: Quote page instructions}
\begin{center}
\includegraphics[width = \textwidth]{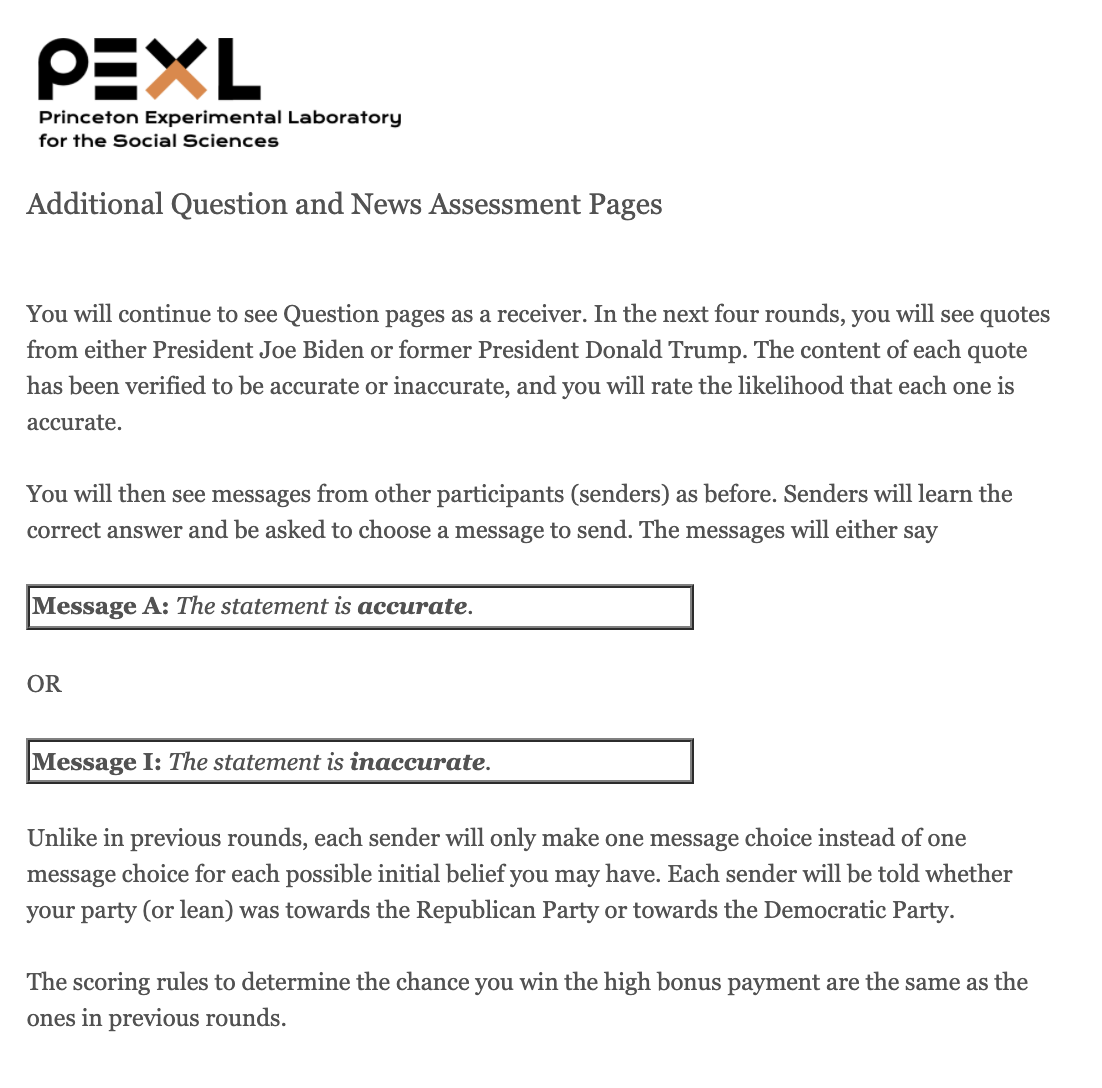}
\end{center}
\end{figure}

\clearpage

\begin{figure}
\caption{Sender: Quote page instructions}
\begin{center}
\includegraphics[width = \textwidth]{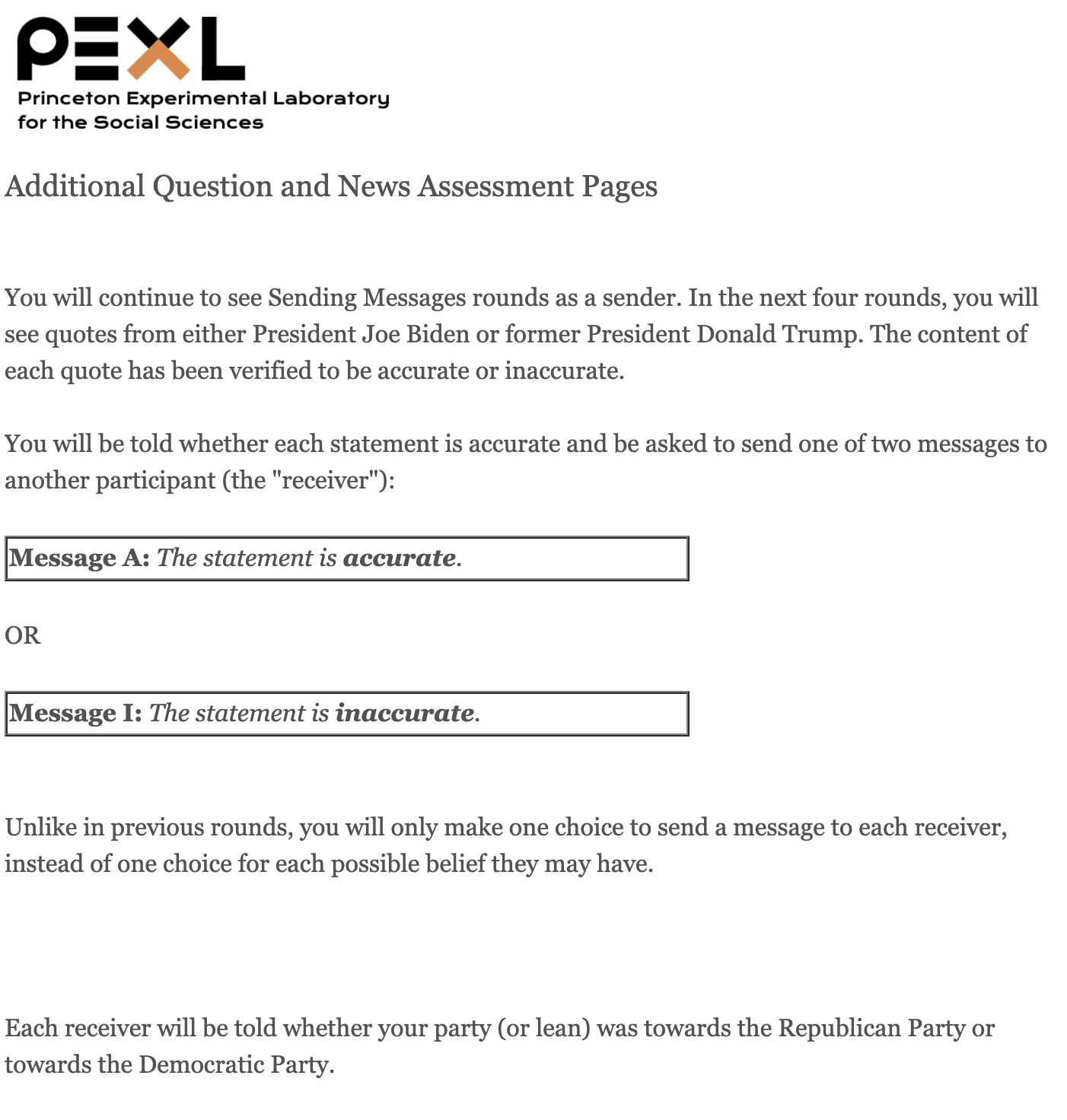}
\end{center}
\end{figure}

\clearpage

\begin{figure}
\caption{Receiver: Quote question}
\begin{center}
\includegraphics[width = \textwidth]{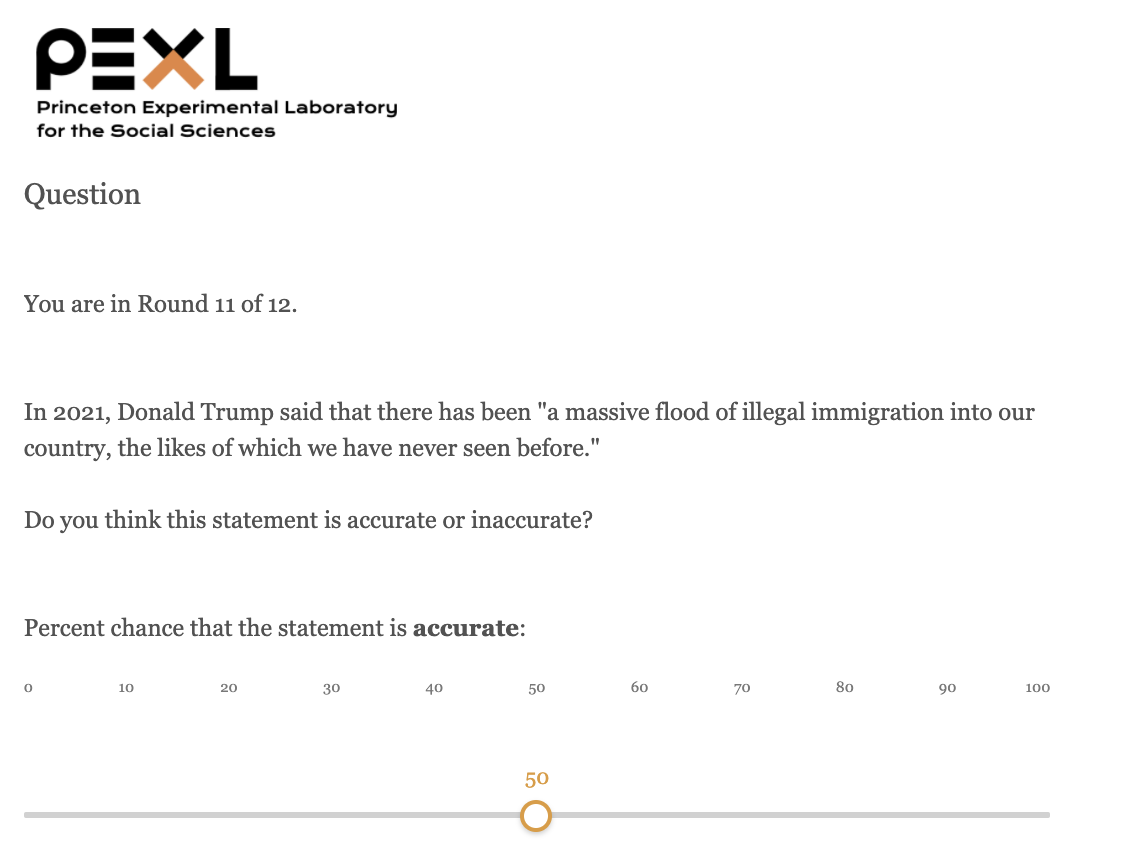}
\end{center}
\end{figure}

\clearpage

\begin{figure}
\caption{Receiver: Quote news ratings}
\begin{center}
\includegraphics[width = .92\textwidth]{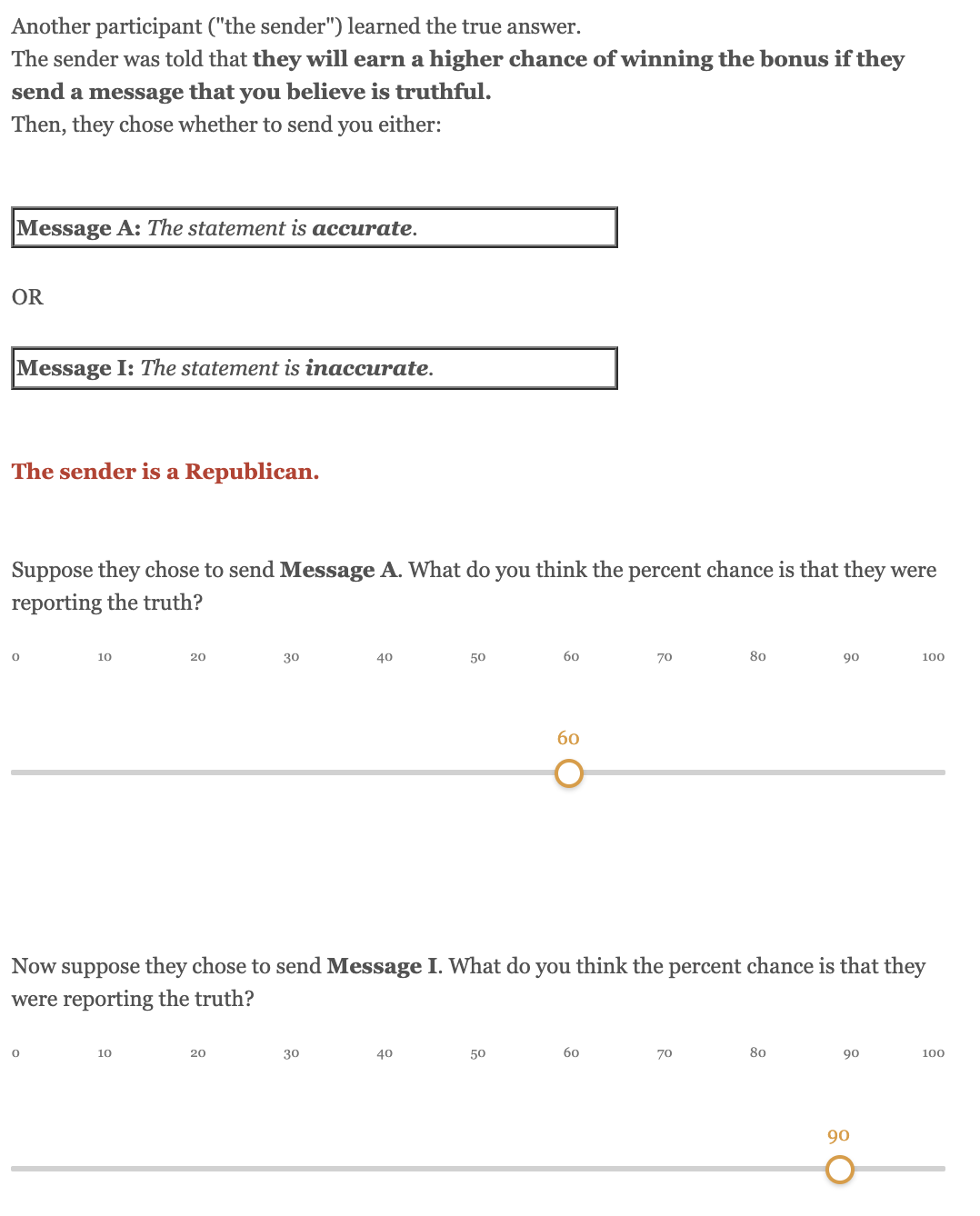}
\end{center}
\vspace{2mm}
\hrule
\vspace{3mm}

\noindent \footnotesize{Receivers in the unincentivized treatment do not see the line about the sender's bonus.}

\vspace{4mm} 
\noindent \footnotesize{Receivers in the competition treatment are instead asked: \textit{Now suppose that \textbf{Sender X} chose to send \textbf{Message G} and that \textbf{Sender Y} chose to send \textbf{Message L}. What do you think is the percent chance is that \textbf{Sender X} sent more truthful signals over the course of the experiment?}}
\end{figure}

\clearpage

\begin{figure}
\caption{Sender: Quote message choice}
\vspace{-7mm}
\begin{center}
\includegraphics[width = .83\textwidth]{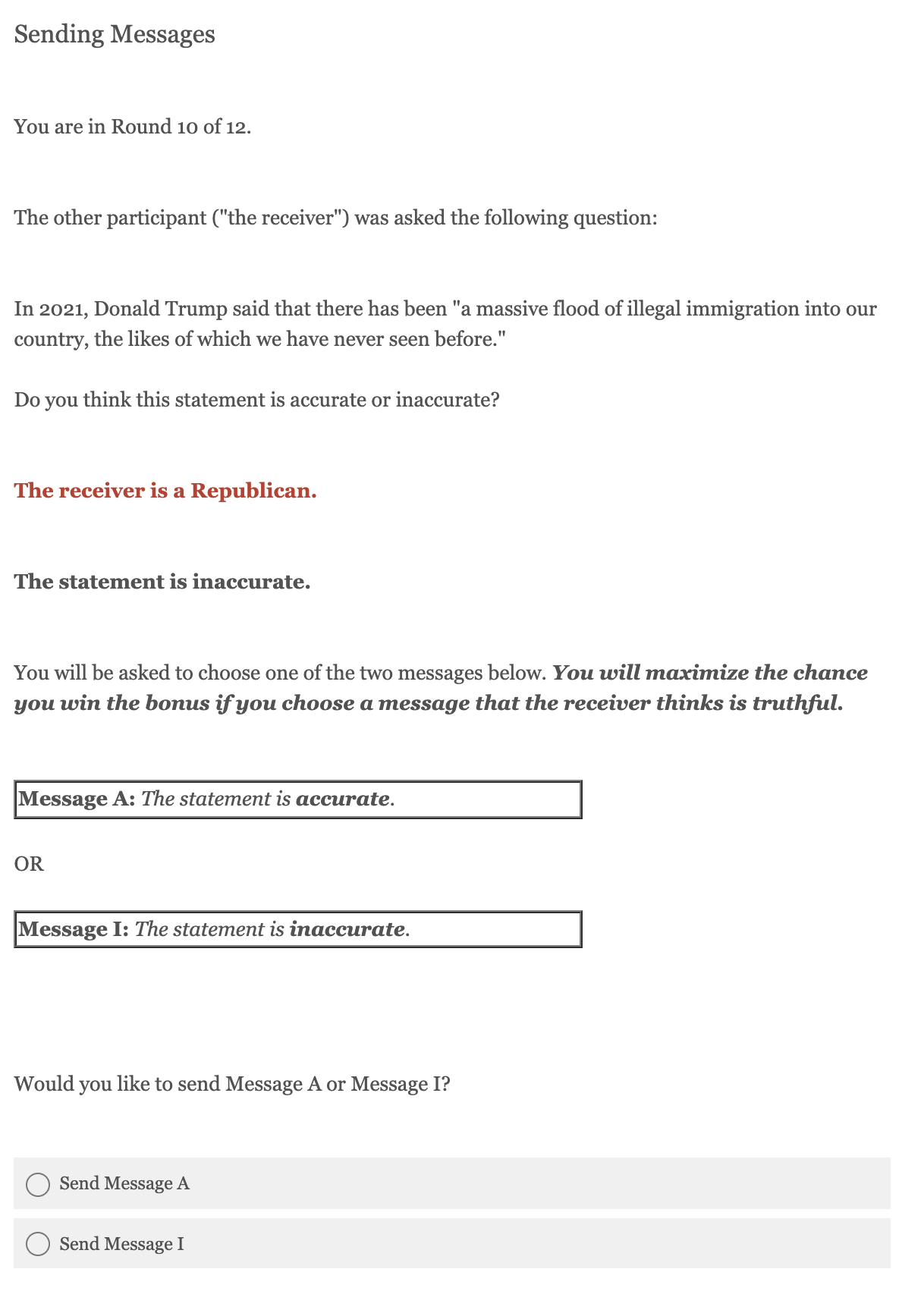}
\end{center}
\hrule
\vspace{2mm}

\noindent \footnotesize{Senders in the unincentivized treatment are instead told: \textit{``\textbf{Your chance of winning the bonus is not affected by how you answer this question.}''}}

\vspace{4mm} 
\noindent \footnotesize{Senders in the competition treatment are instead told that \textit{``\textbf{The receiver will predict, based on your message and another sender's message on this question, which sender sent more truthful messages over the course of the experiment. You will maximize the chance you win the bonus if the receiver believes that you sent more truthful messages.}''}}
\end{figure}

\clearpage

\begin{figure}
\caption{Receiver: Predictions}
\begin{center}
\includegraphics[width = \textwidth]{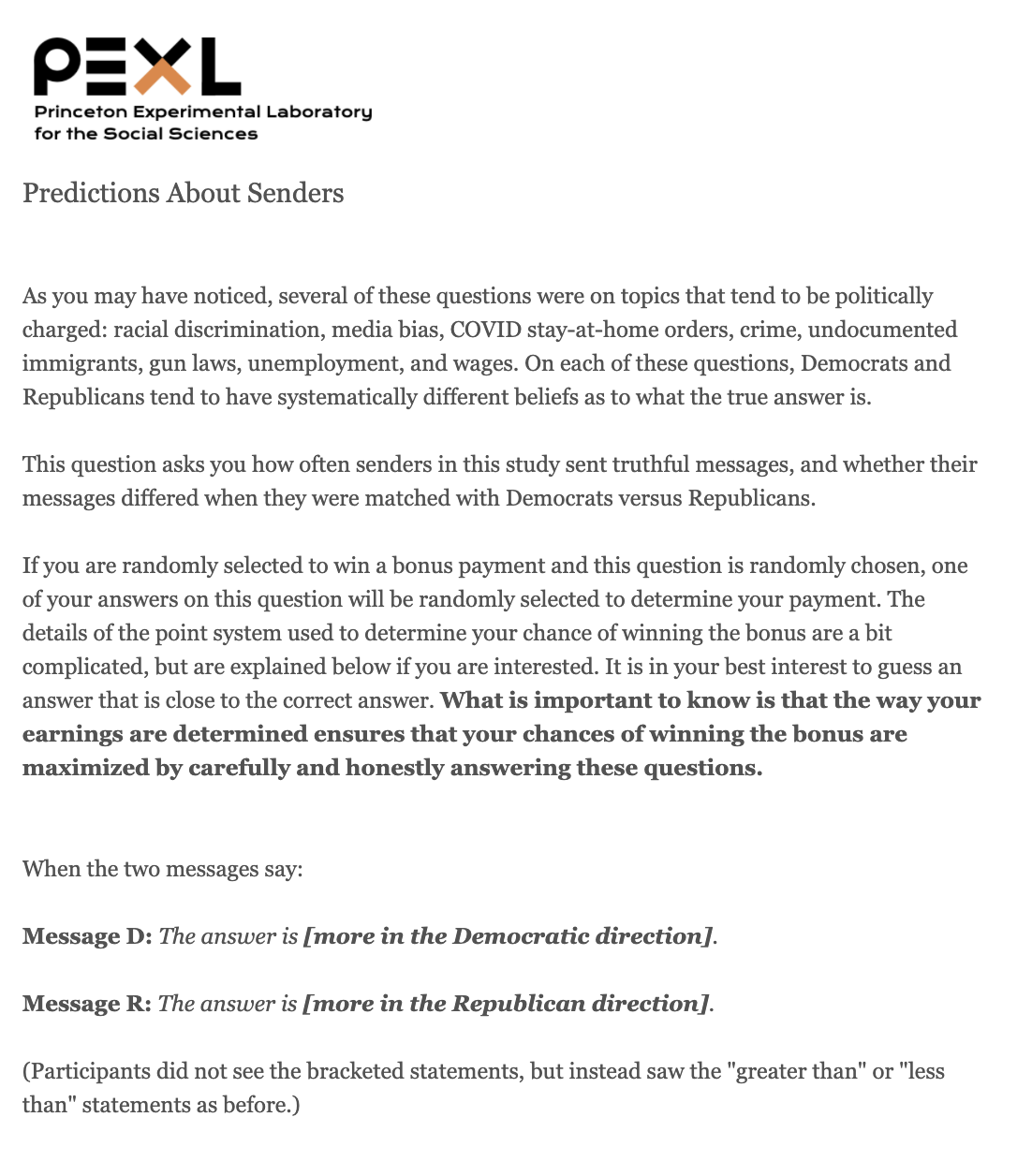}
\end{center}
\end{figure}

\clearpage

\begin{center}
\includegraphics[width = .83\textwidth]{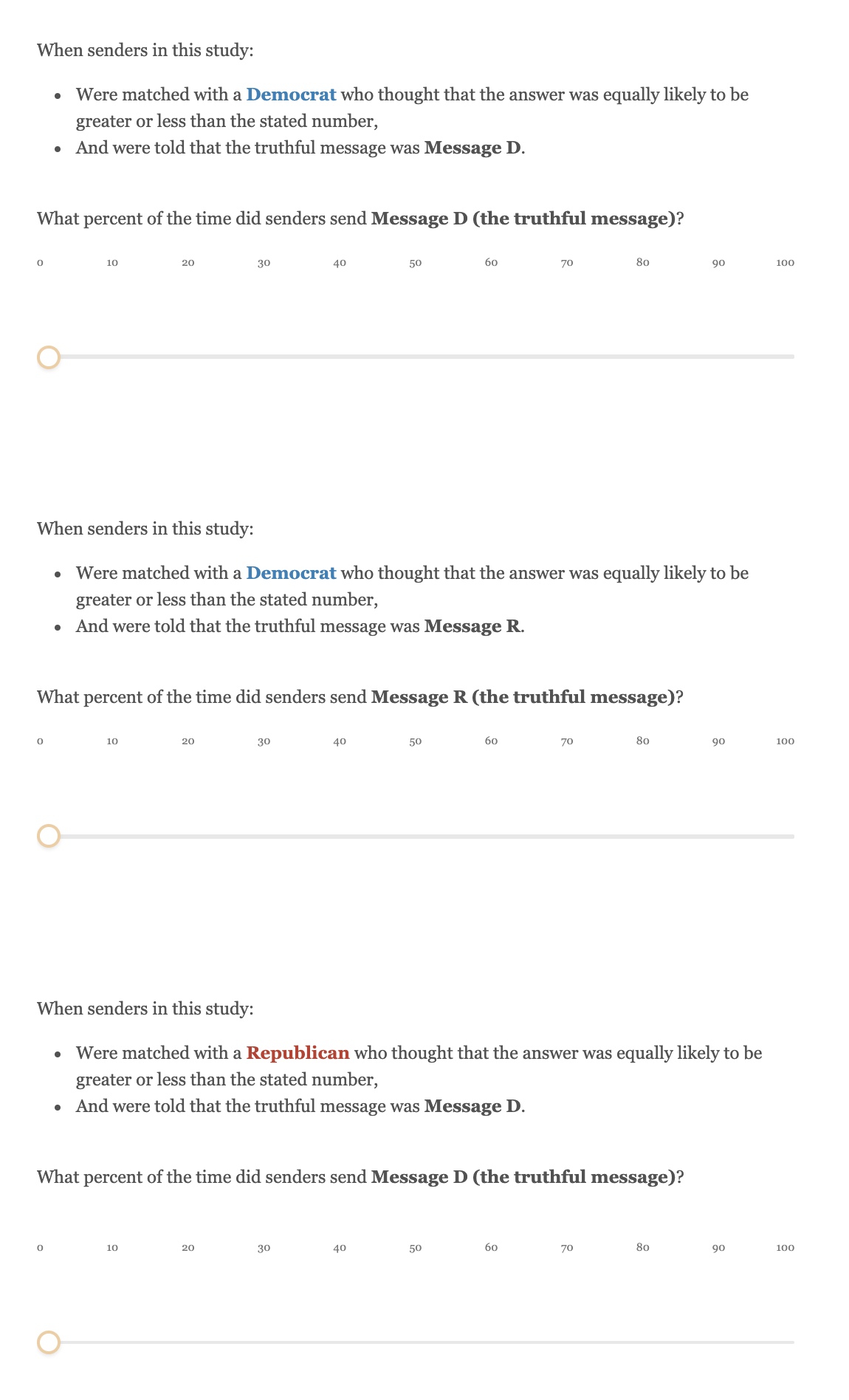}
\end{center}

\clearpage

\begin{center}
\includegraphics[width = .83\textwidth]{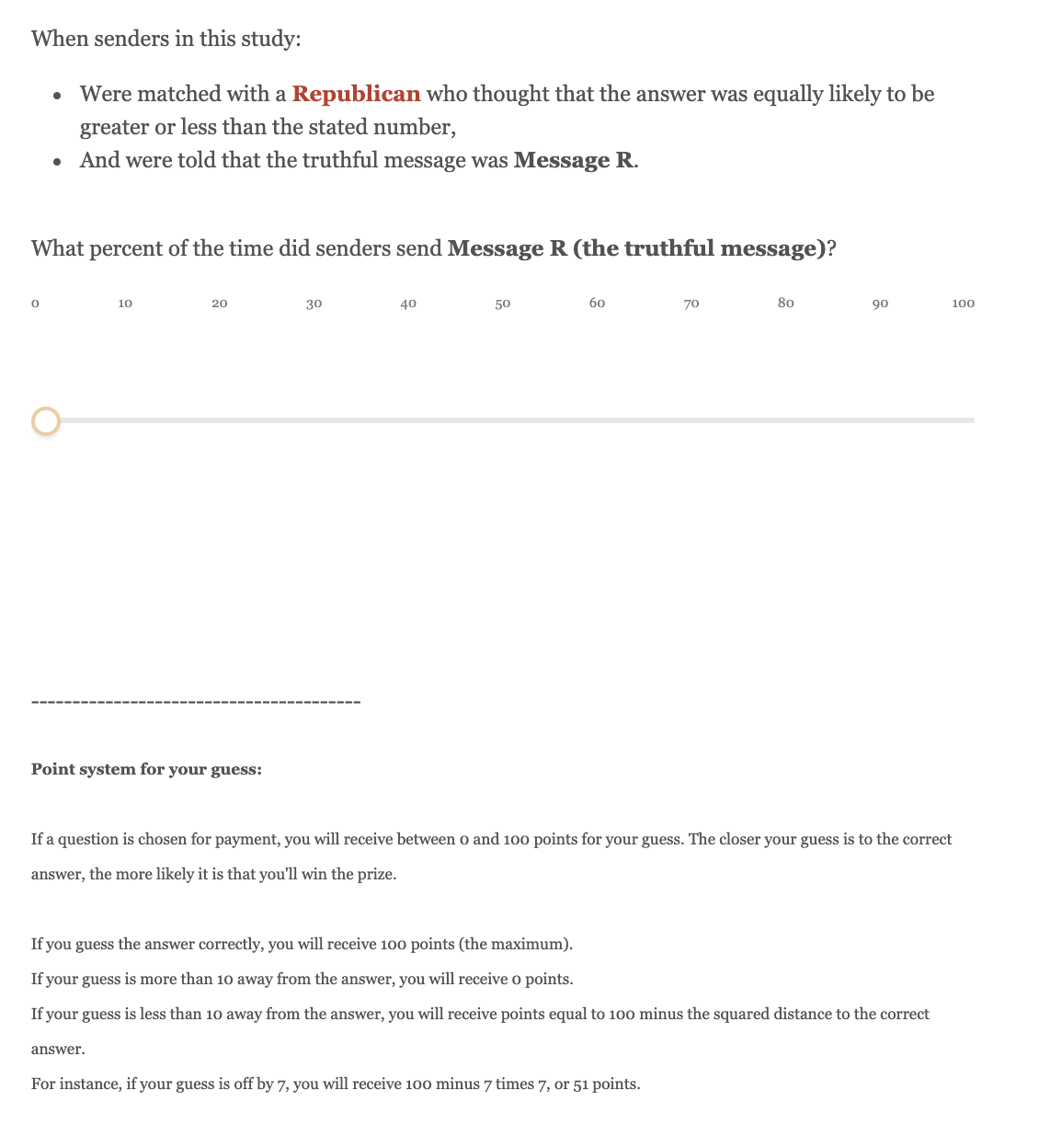}
\end{center}
\clearpage

\begin{figure}
\caption{Sender: Predictions}
\begin{center}
\includegraphics[width = \textwidth]{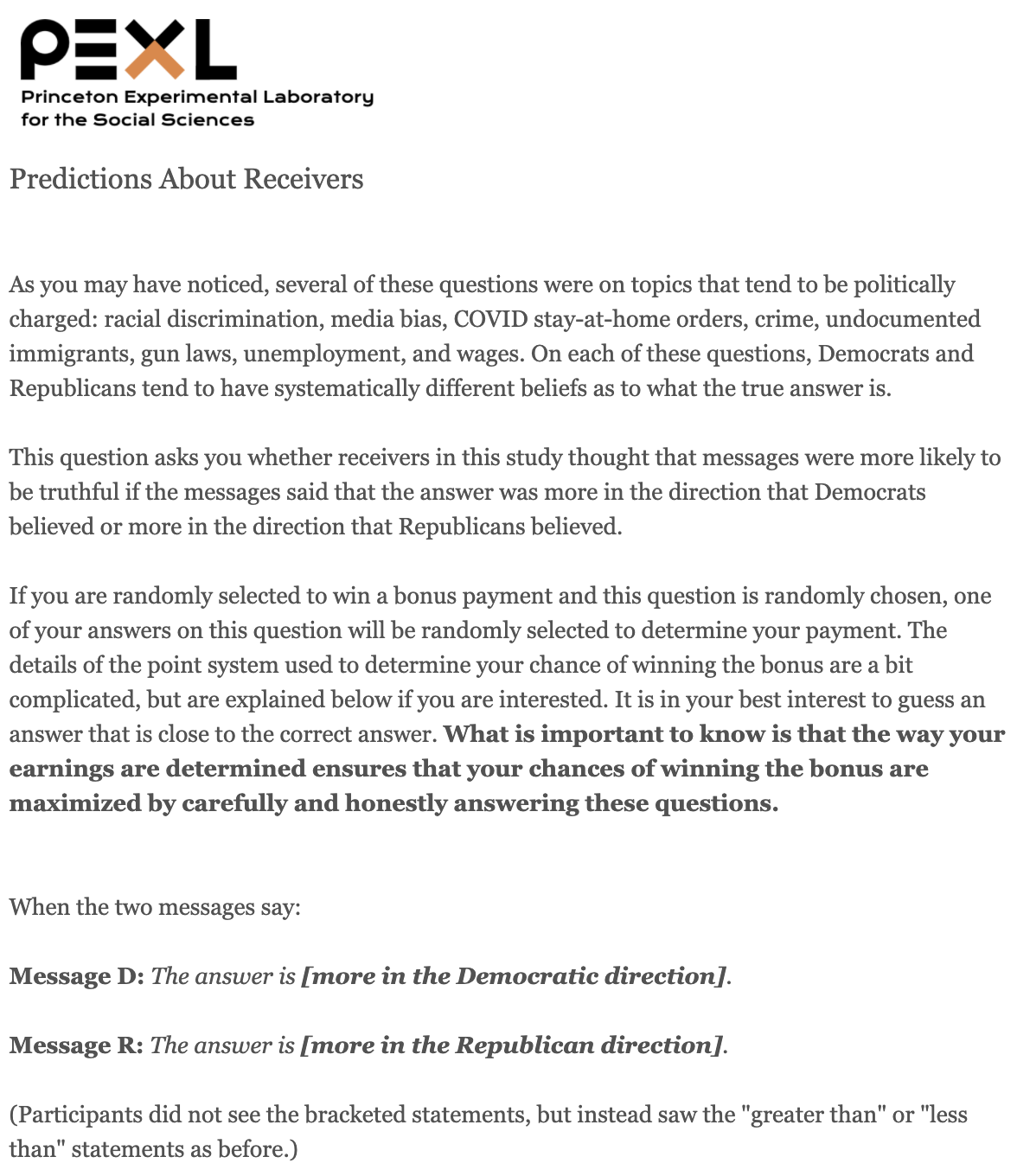}
\end{center}
\end{figure}

\clearpage

\begin{center}
\includegraphics[width = \textwidth]{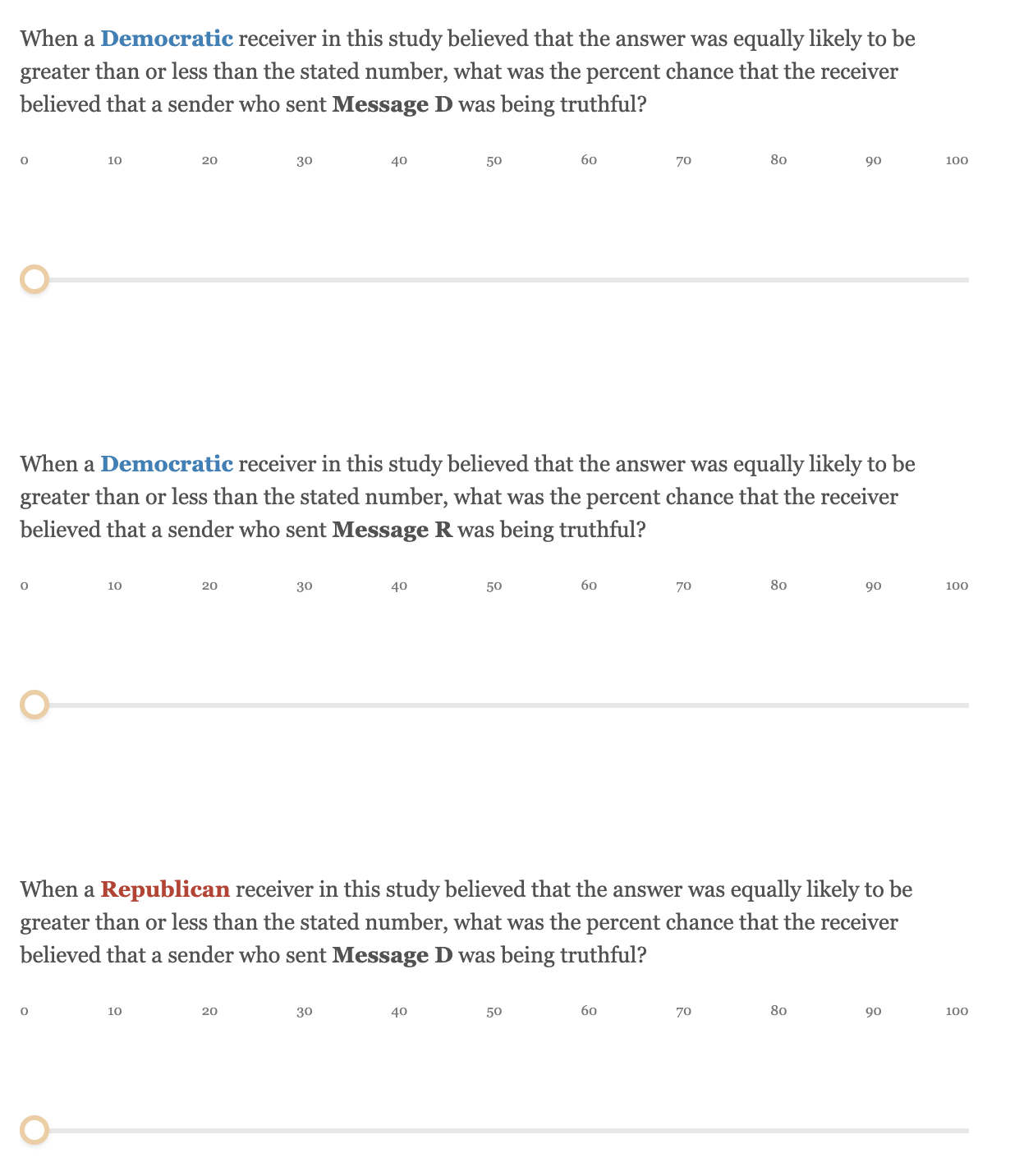}
\end{center}

\clearpage

\begin{center}
\includegraphics[width = \textwidth]{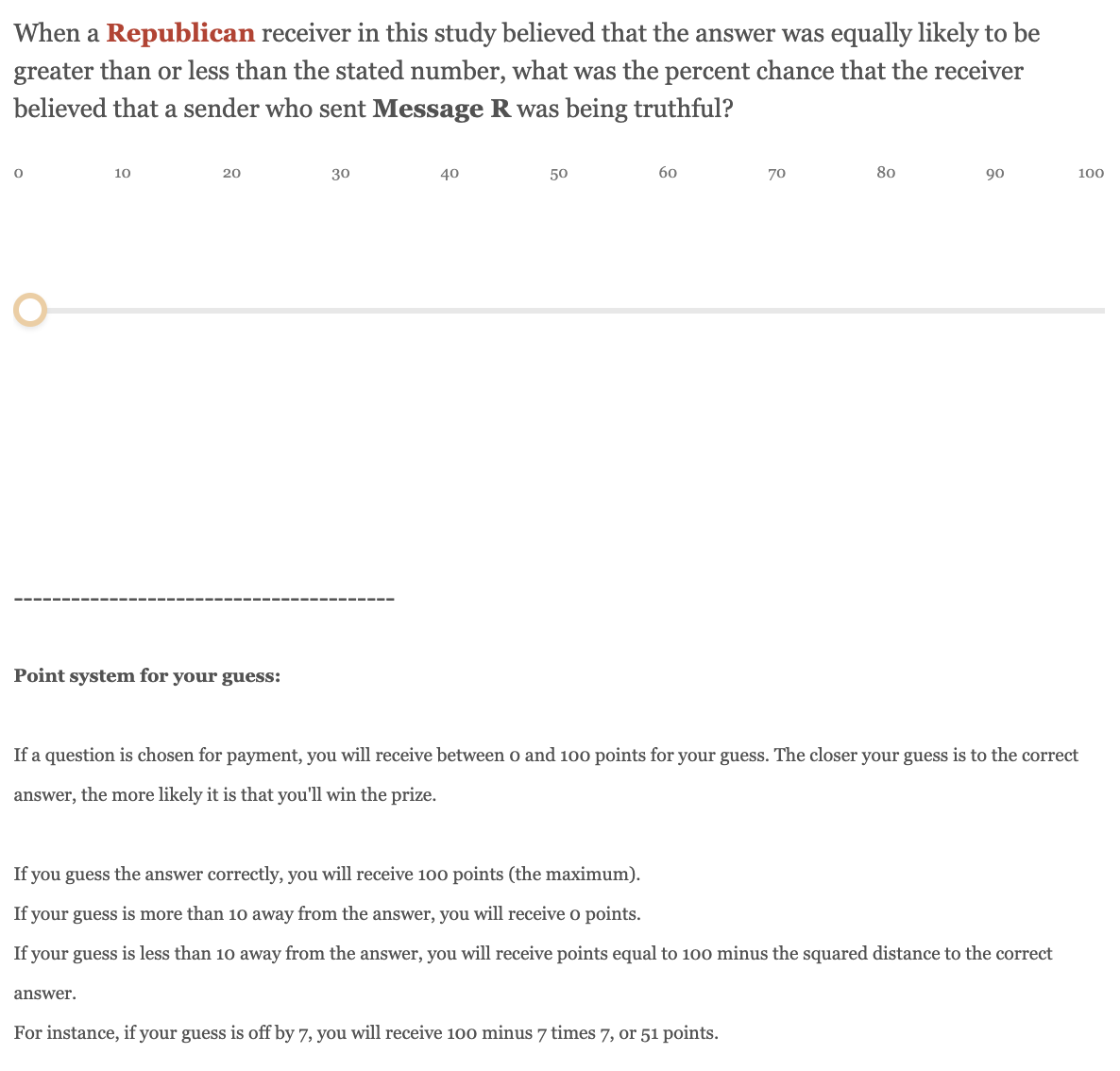}
\end{center}

\clearpage





\begin{figure}
\caption{Receiver: Survey beliefs}
\begin{center}
\includegraphics[width = \textwidth]{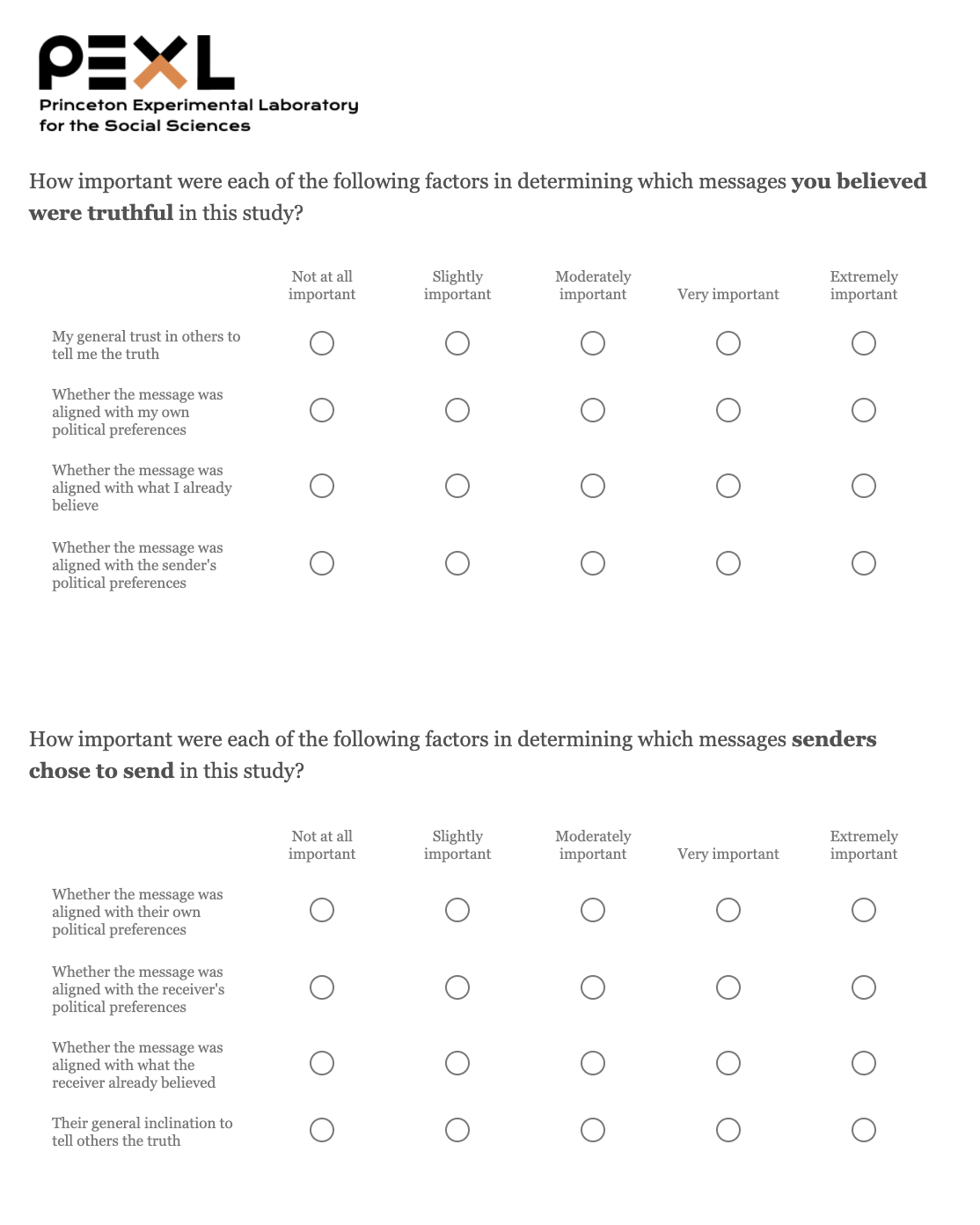}
\end{center}
\end{figure}

\clearpage

\begin{figure}
\caption{Sender: Survey beliefs}
\begin{center}
\includegraphics[width = \textwidth]{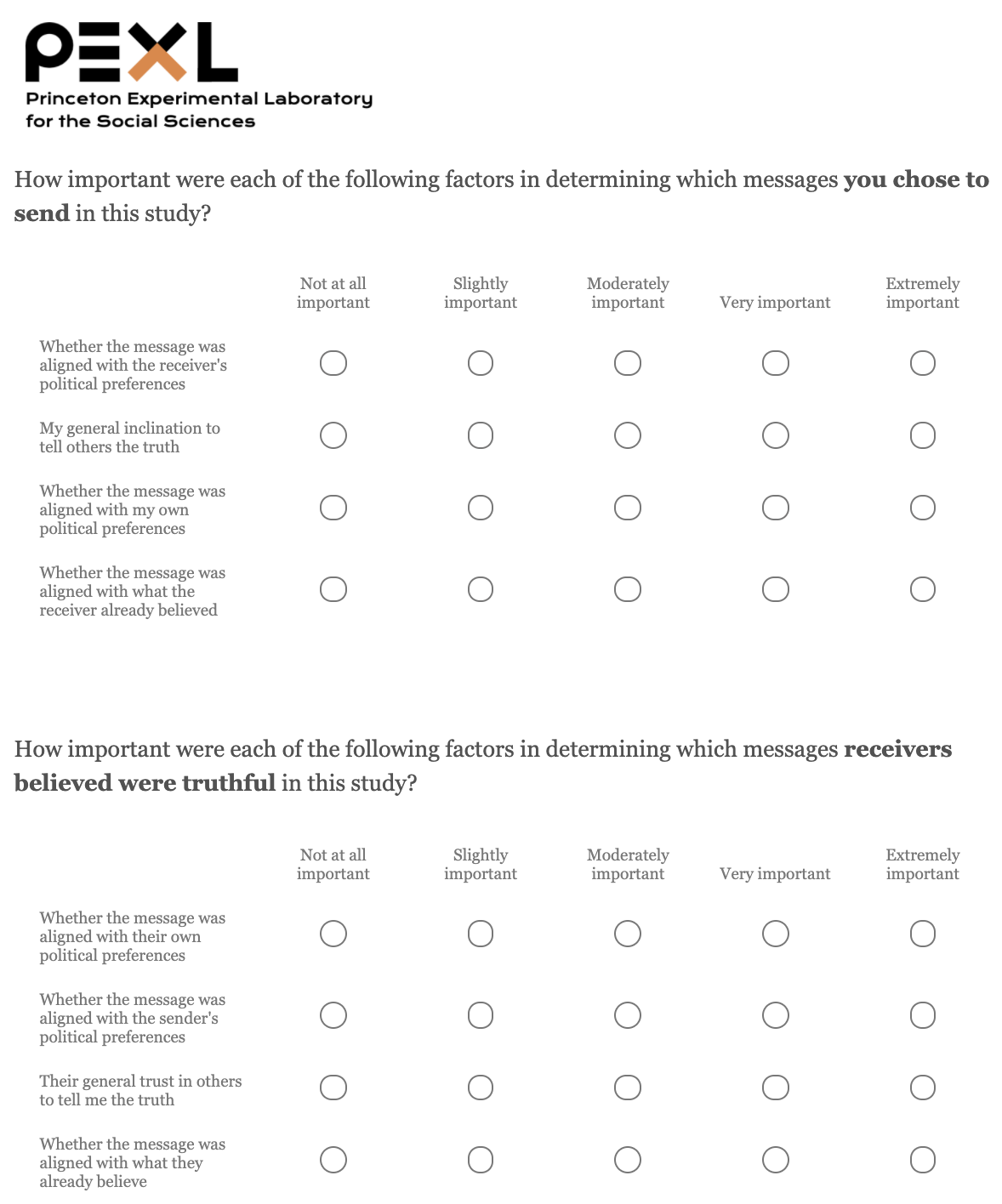}
\end{center}
\end{figure}

\clearpage





\section{Online Appendix: Study Materials for the Additional Experiment}\label{screenshots2}

\subsection{Questions}

\normalsize{
Receivers saw a subset of the questions from the primary experiment. Instead of being given a target number, they were asked to input their guess. For instance, the end of the question about violent crime among undocumented immigrants said ``What was the felony violent crime rate per 100,000 undocumented immigrants?''

Senders saw each question with the following median beliefs of receivers:}

\vspace{5mm}
\renewcommand\arraystretch{1.3}
\begin{center}
\begin{tabular}{l l l}  
\toprule
\textbf{Topic} & \textbf{Median belief} & \textbf{Truthful message} \\
\midrule
US crime & 500.0 & Less than \\
Immigrants' crime & 213.0 & Less than \\
Racial discrimination & 8.50 & Less than \\
Media bias & 65 & Greater than \\
COVID-19 restrictions & 50.0 & Less than \\
Gun reform & 220.0 & Greater than \\
Unemployment & 3.20 & Greater than \\
Wages & 4.00 & Less than \\
Latitude of US & 45.0 & Less than \\
Random number & 50.0 & Greater than  \\
\bottomrule
\end{tabular}
\end{center}
\renewcommand\arraystretch{1}

\vspace{3mm}

\noindent The truthful message corresponds to which one of the ``Greater than'' / ``Less than'' messages is true, given the receiver's median belief.

\clearpage

\subsection{Screenshots}

\vspace{5mm}
\begin{figure}[h]
\caption{Overview page}
\begin{center}
\includegraphics[width = .95\textwidth]{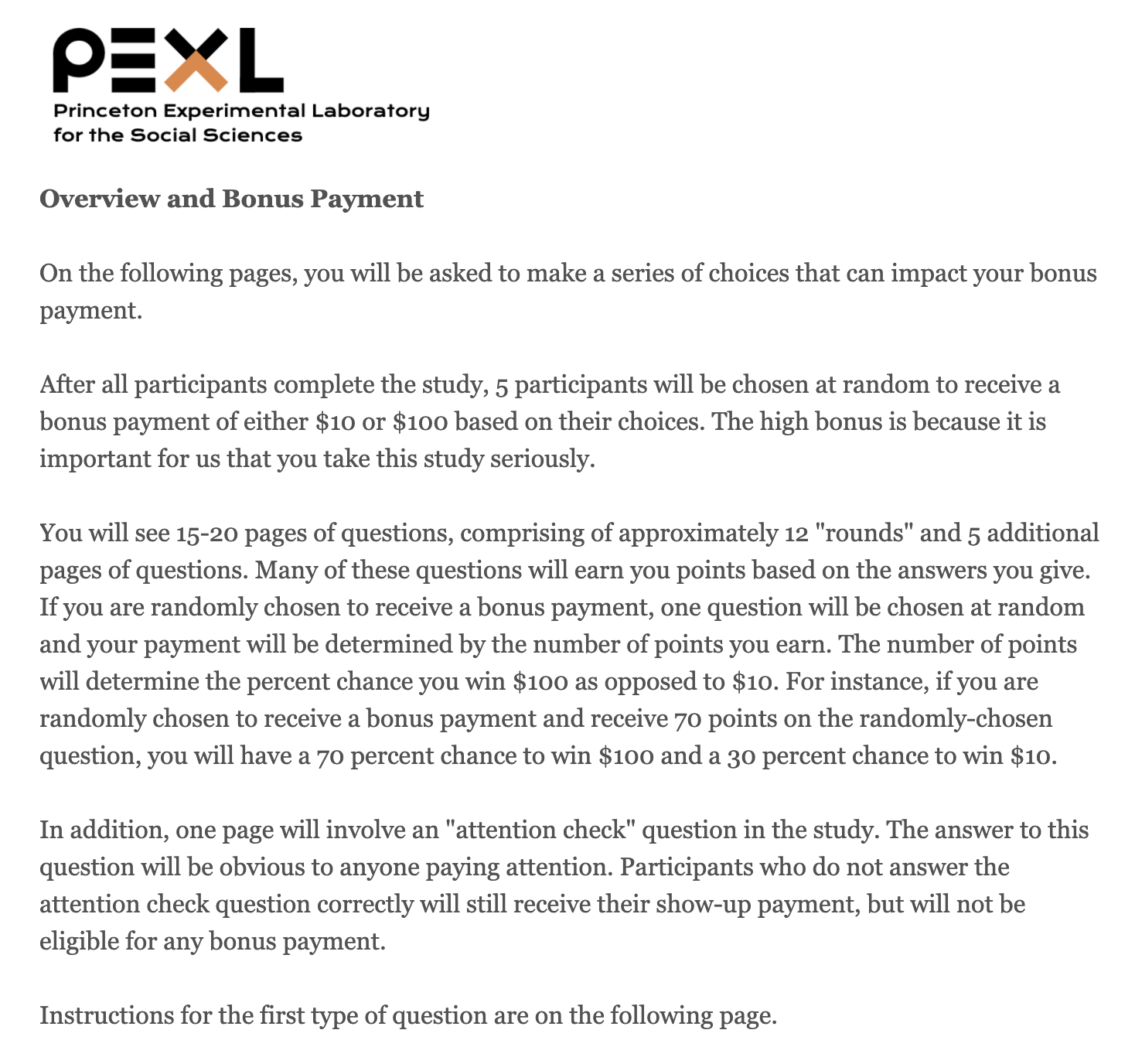}
\end{center}
\end{figure}

\clearpage

\begin{figure}
\caption{Instructions for median beliefs}
\begin{center}
\includegraphics[width = \textwidth]{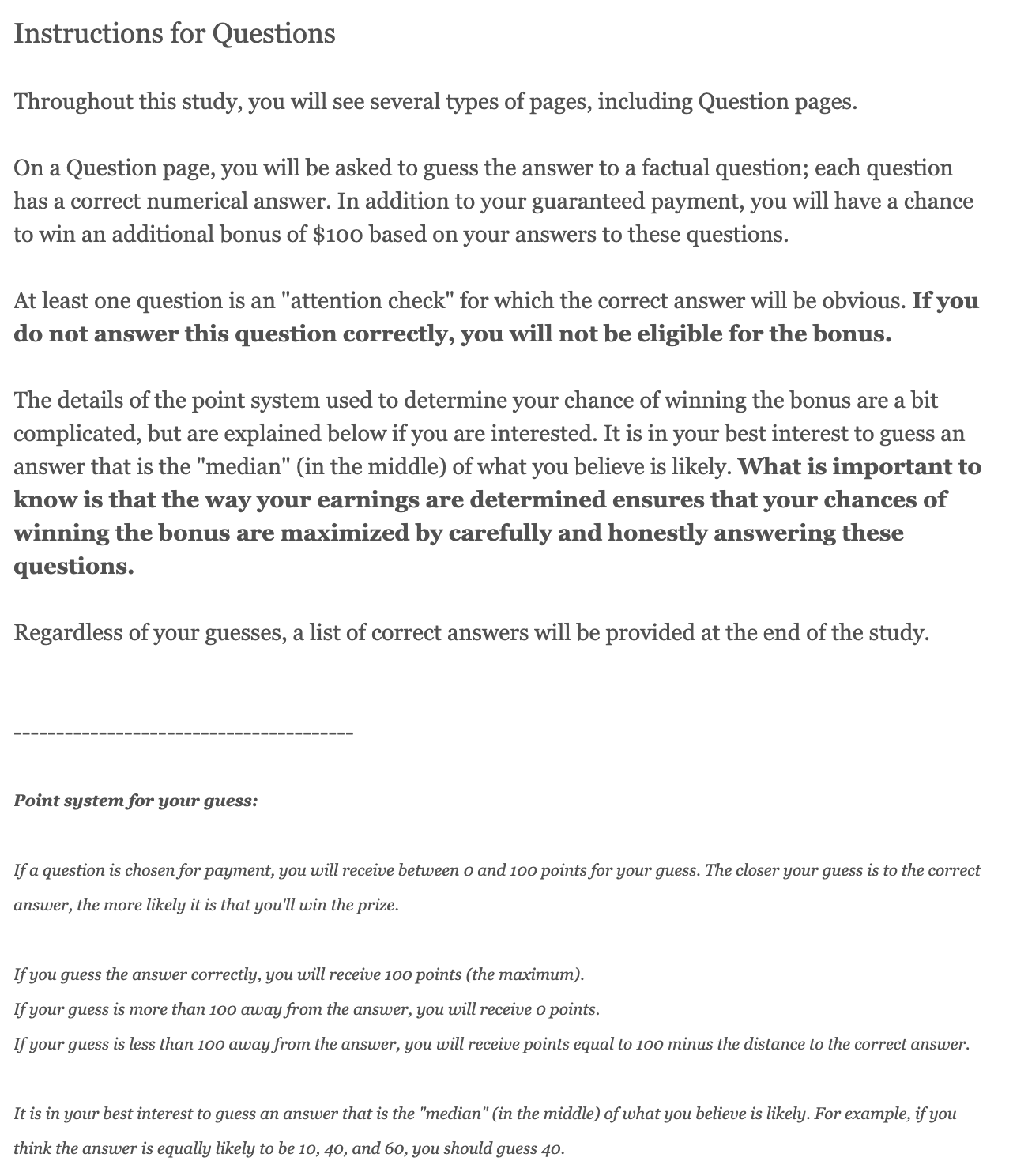}
\end{center}
\end{figure}

\clearpage

\begin{figure}
\caption{Instructions for news ratings}
\begin{center}
\includegraphics[width = \textwidth]{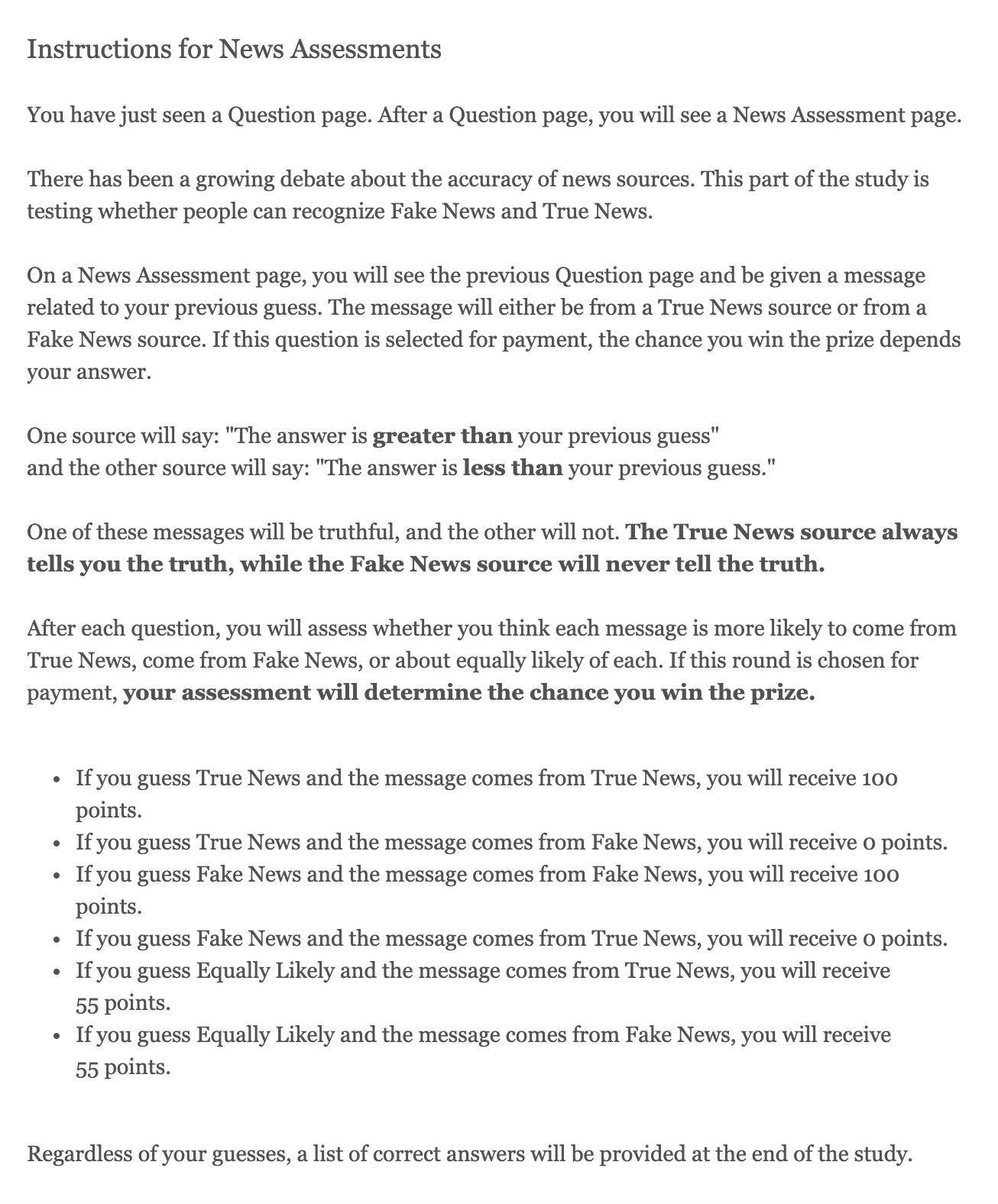}
\end{center}
\end{figure}

\clearpage

\begin{figure}
\caption{Receiver: Question page}
\begin{center}
\includegraphics[width = \textwidth]{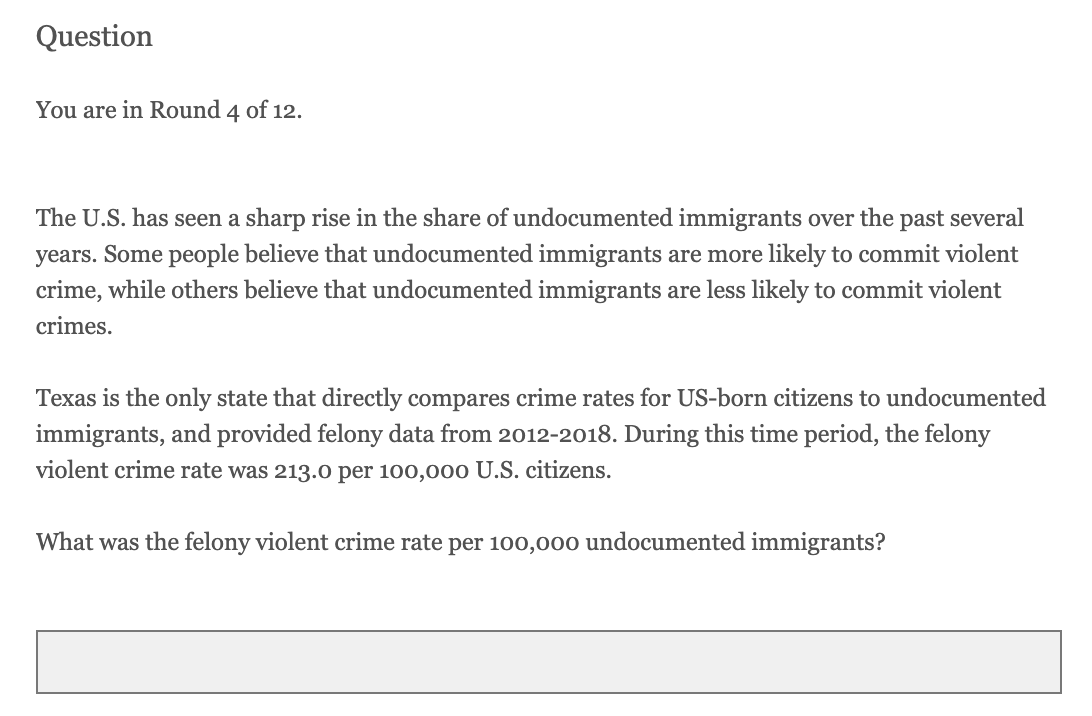}
\end{center}
\end{figure}

\clearpage

\begin{figure}
\caption{Receiver: News page}
\begin{center}
\includegraphics[width = .93\textwidth]{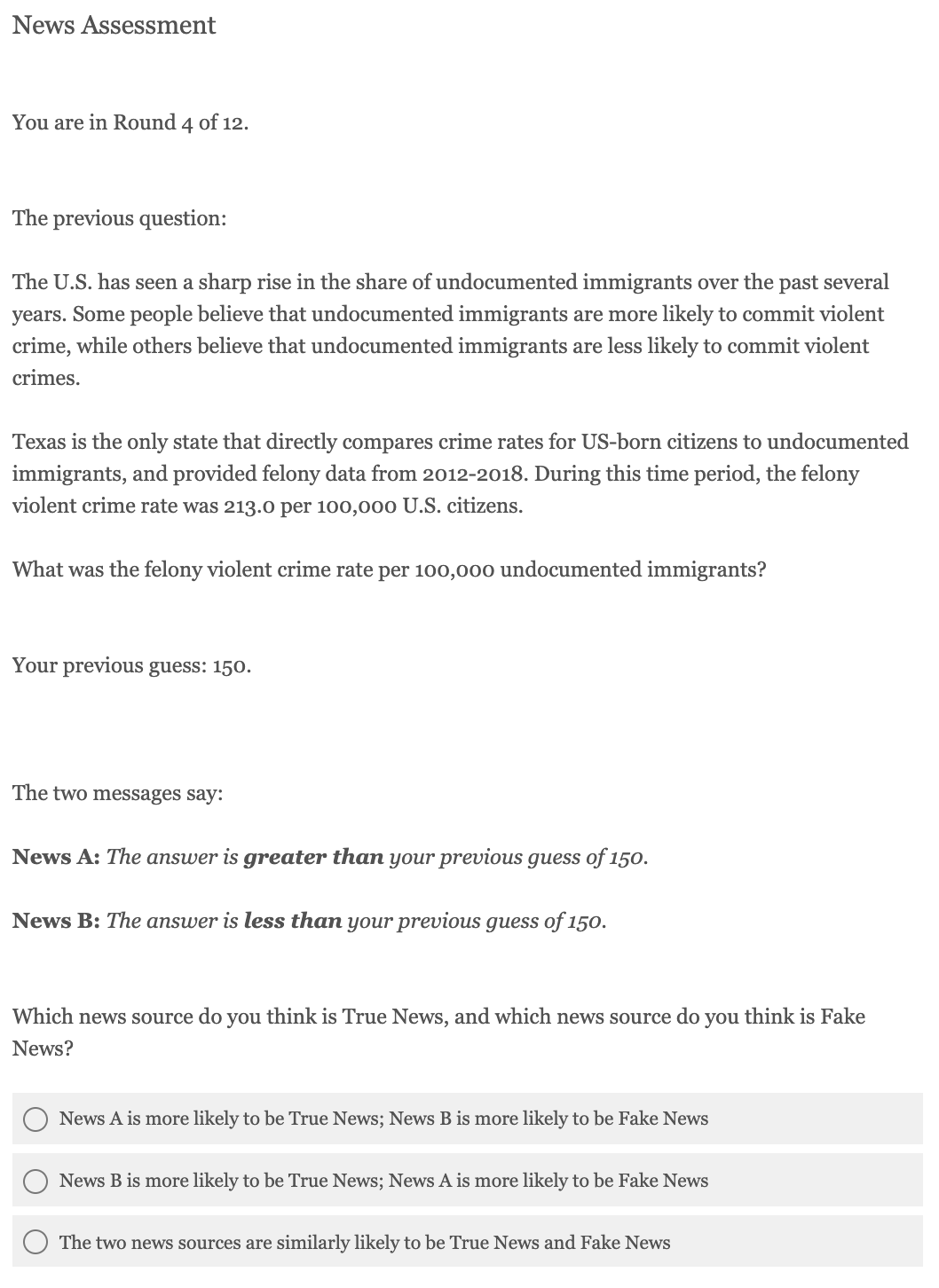}
\end{center}
\end{figure}





\clearpage

\begin{figure}
\caption{Instructions for choosing messages}
\begin{center}
\includegraphics[width = .95\textwidth]{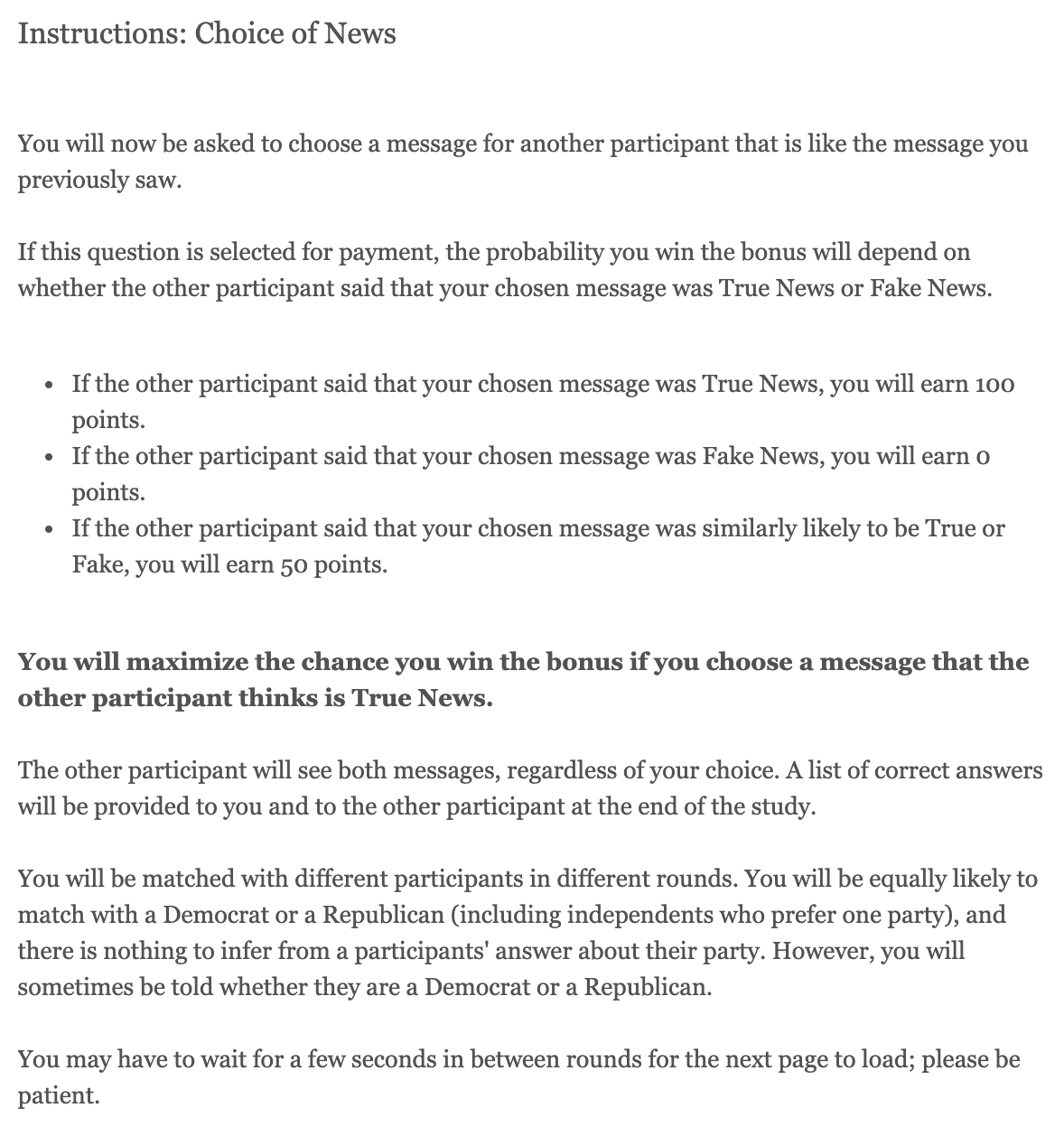}
\end{center}
\vspace{4mm}
\hrule
\vspace{3mm}

\noindent \footnotesize{Senders in the unincentivized treatment do not see the sentences about the bonus, and instead see: \textit{\textbf{``These questions, unless otherwise specified, will not affect your chance to win a bonus payment.''}}}
\end{figure}

\clearpage

\begin{figure}
\caption{Sender: Choosing messages page}
\begin{center}
\includegraphics[width = .95\textwidth]{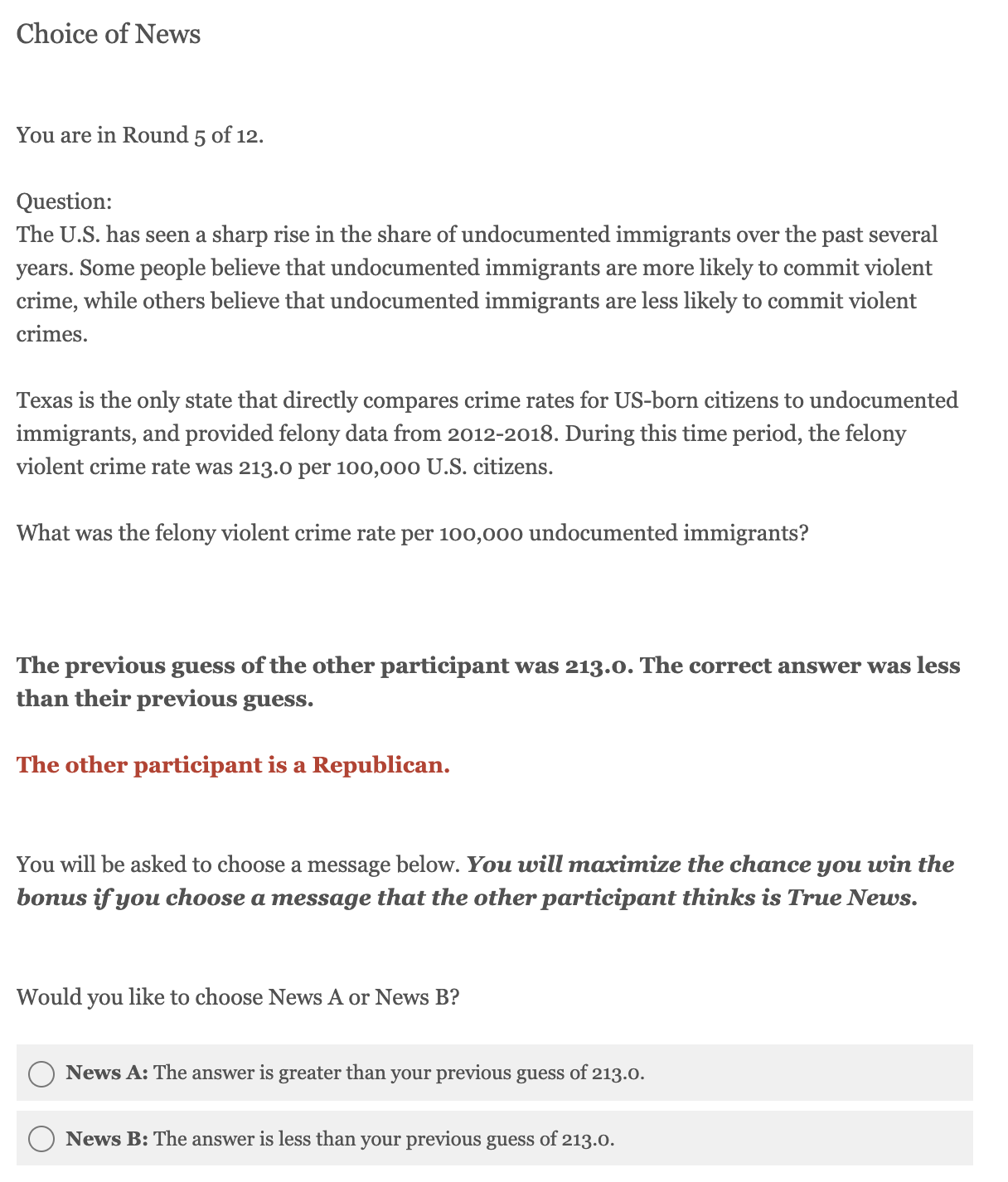}
\end{center}
\vspace{4mm}
\hrule
\vspace{3mm}

\noindent \footnotesize{Senders in the unincentivized treatment do not see the sentence about the bonus.}
\end{figure}

\clearpage

\begin{figure}
\caption{Sender instructions: Demand for information}
\begin{center}
\includegraphics[width = \textwidth]{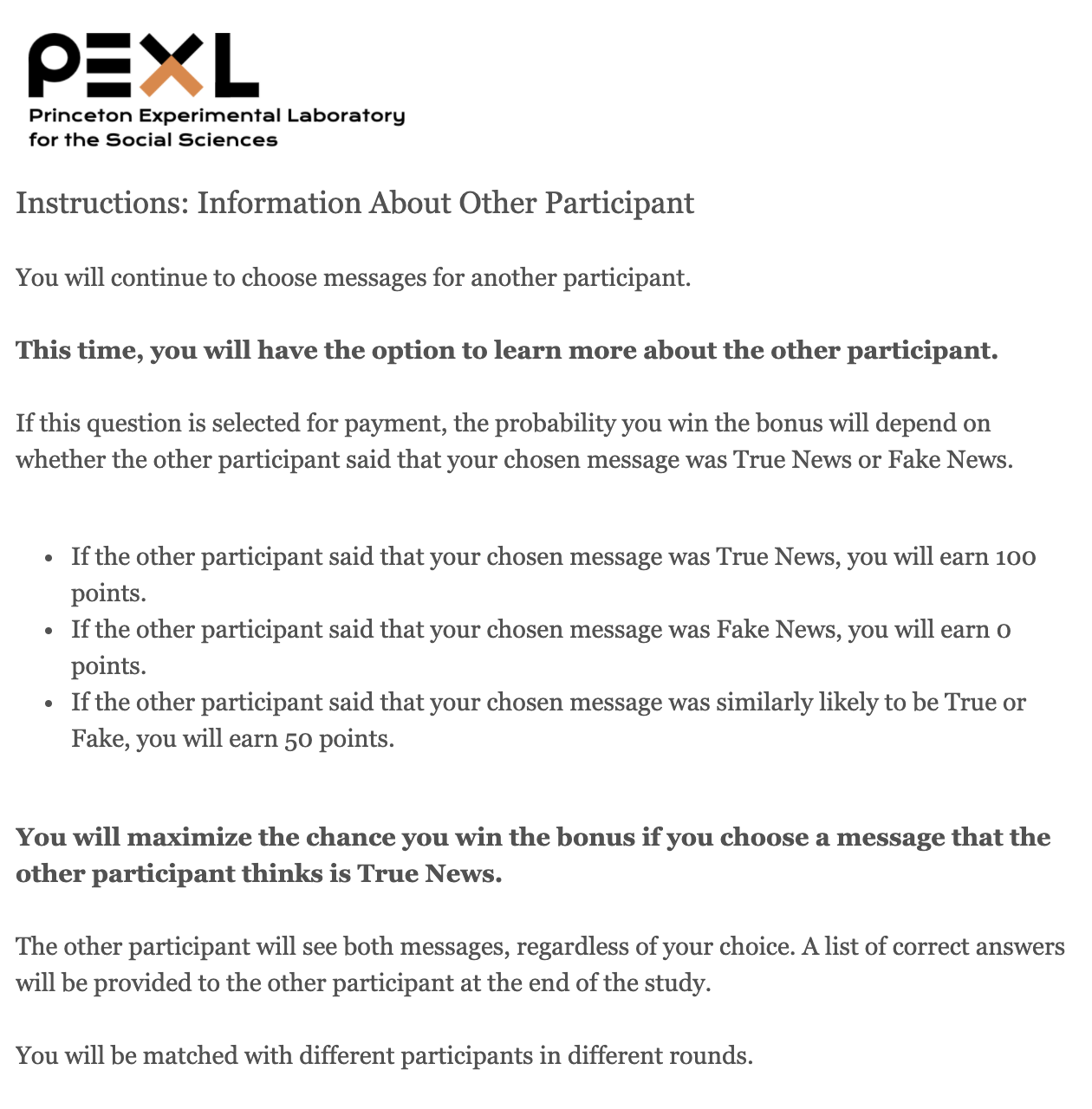}
\end{center}
\end{figure}



\clearpage

\begin{figure}
\caption{Sender: Information choice page}
\begin{center}
\includegraphics[width = \textwidth]{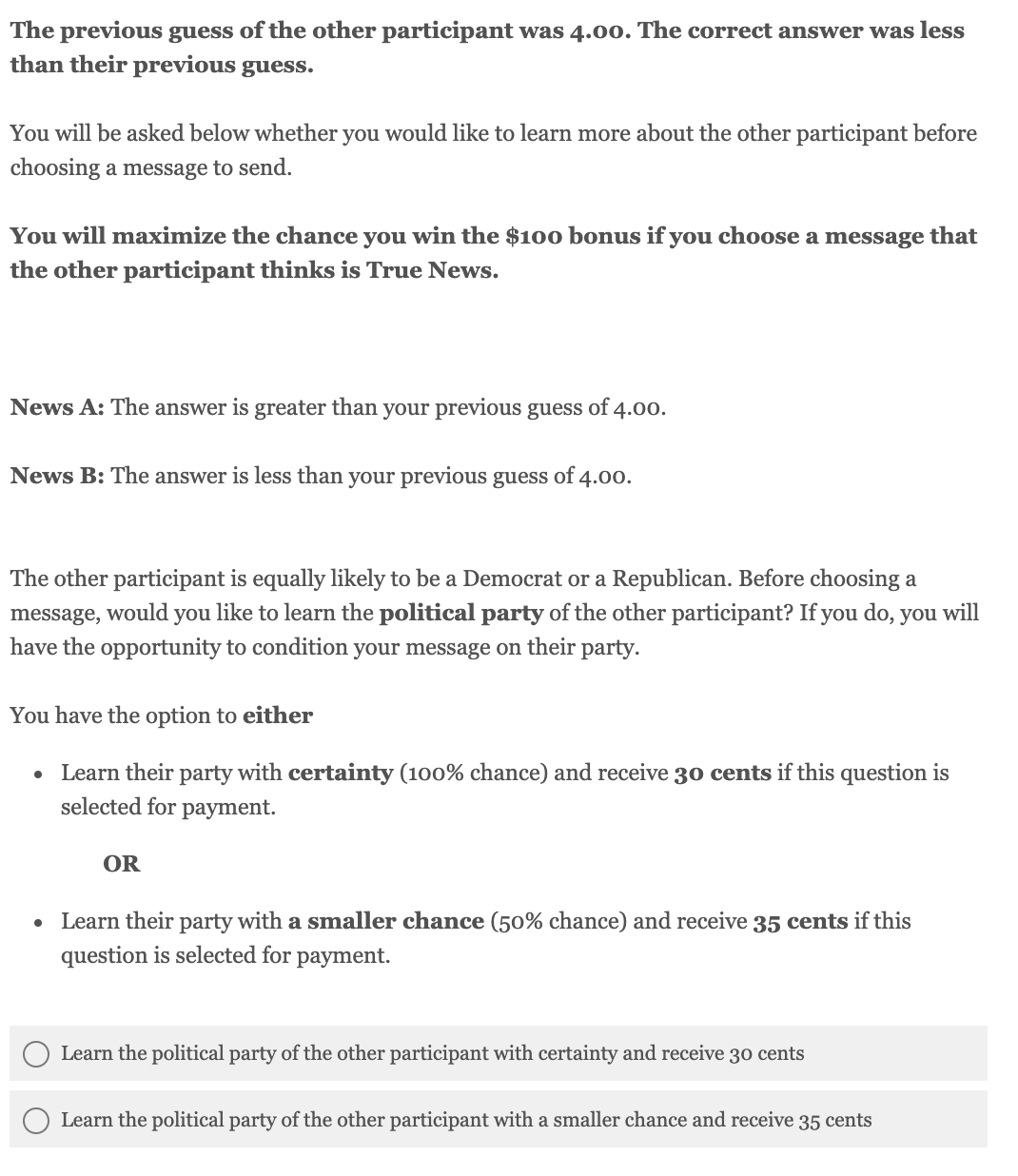}
\end{center}
\end{figure}

\clearpage

\begin{figure}
\caption{Sender: Can condition on receiver's party}
\begin{center}
\includegraphics[width = \textwidth]{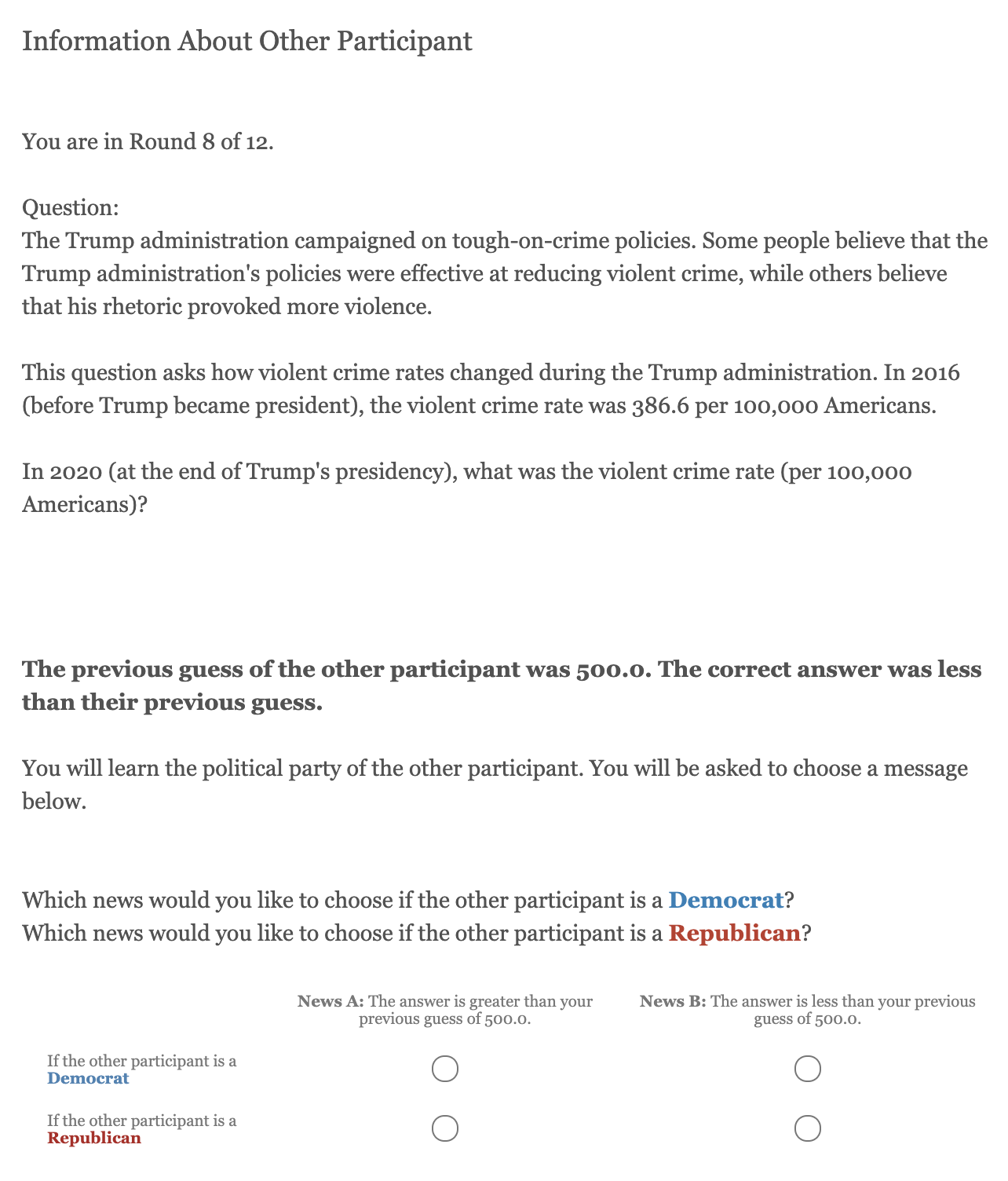}
\end{center}
\end{figure}

\clearpage

\begin{figure}
\caption{Sender: Cannot condition on receiver's party}
\begin{center}
\includegraphics[width = \textwidth]{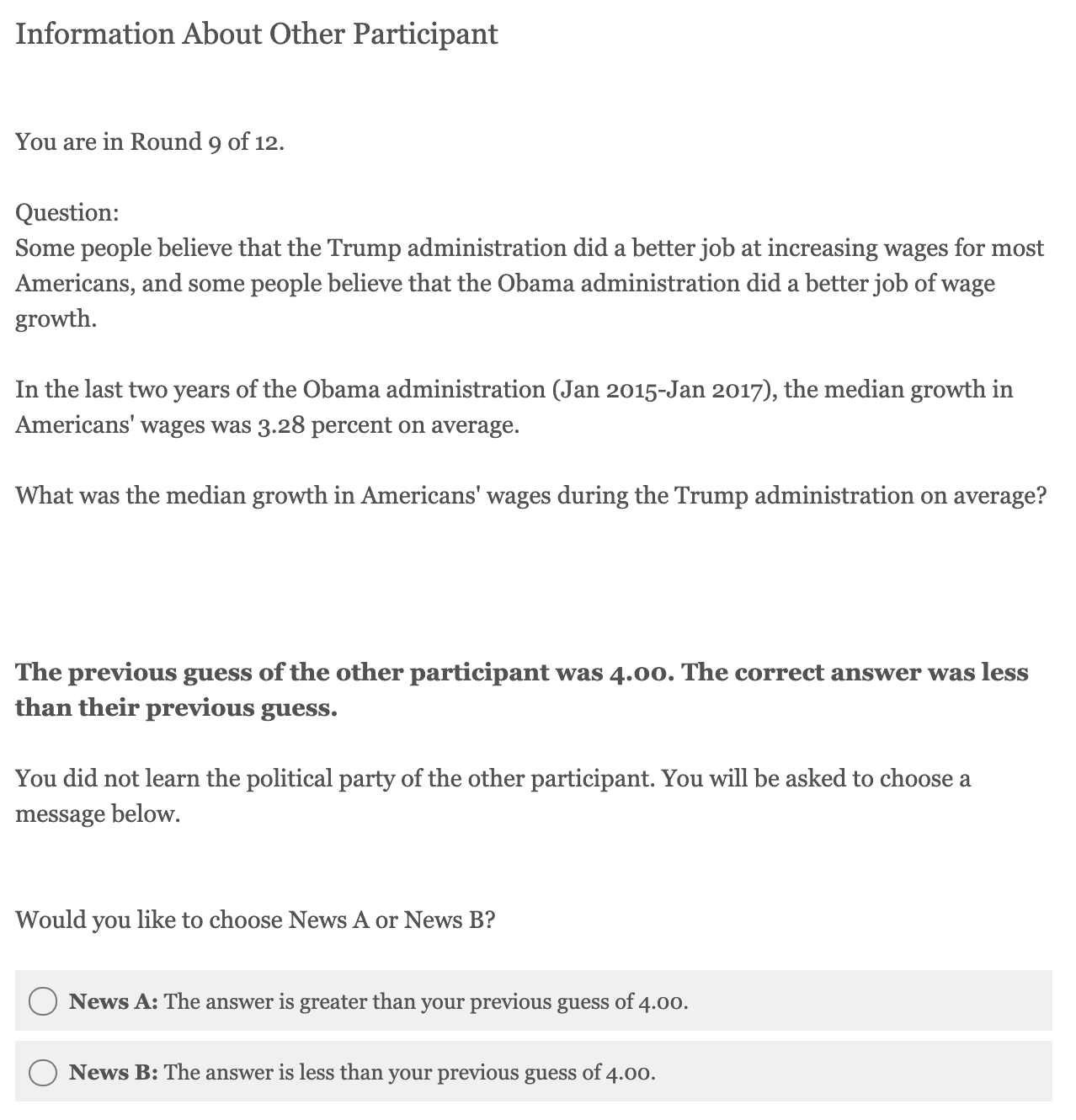}
\end{center}
\end{figure}

\clearpage

\begin{figure}
\caption{Attention check}
\begin{center}
\includegraphics[width = \textwidth]{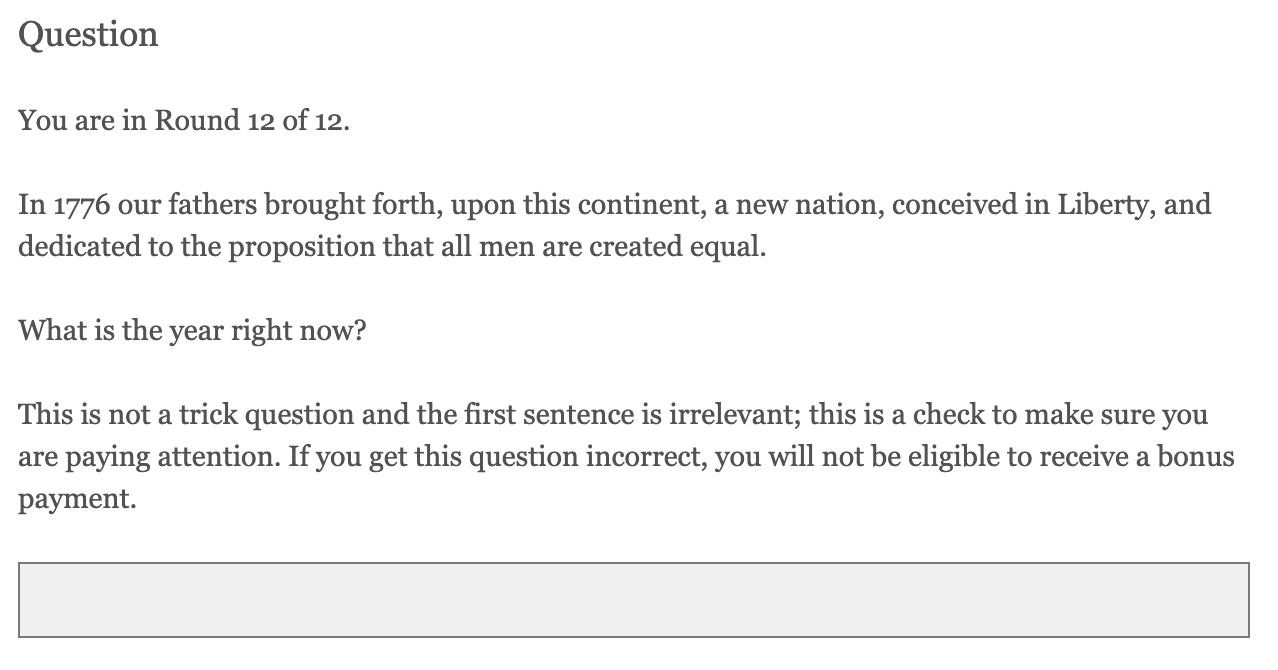}
\end{center}
\end{figure}

\clearpage

\begin{figure}
\caption{Sender: Beliefs about receivers}
\begin{center}
\includegraphics[width = \textwidth]{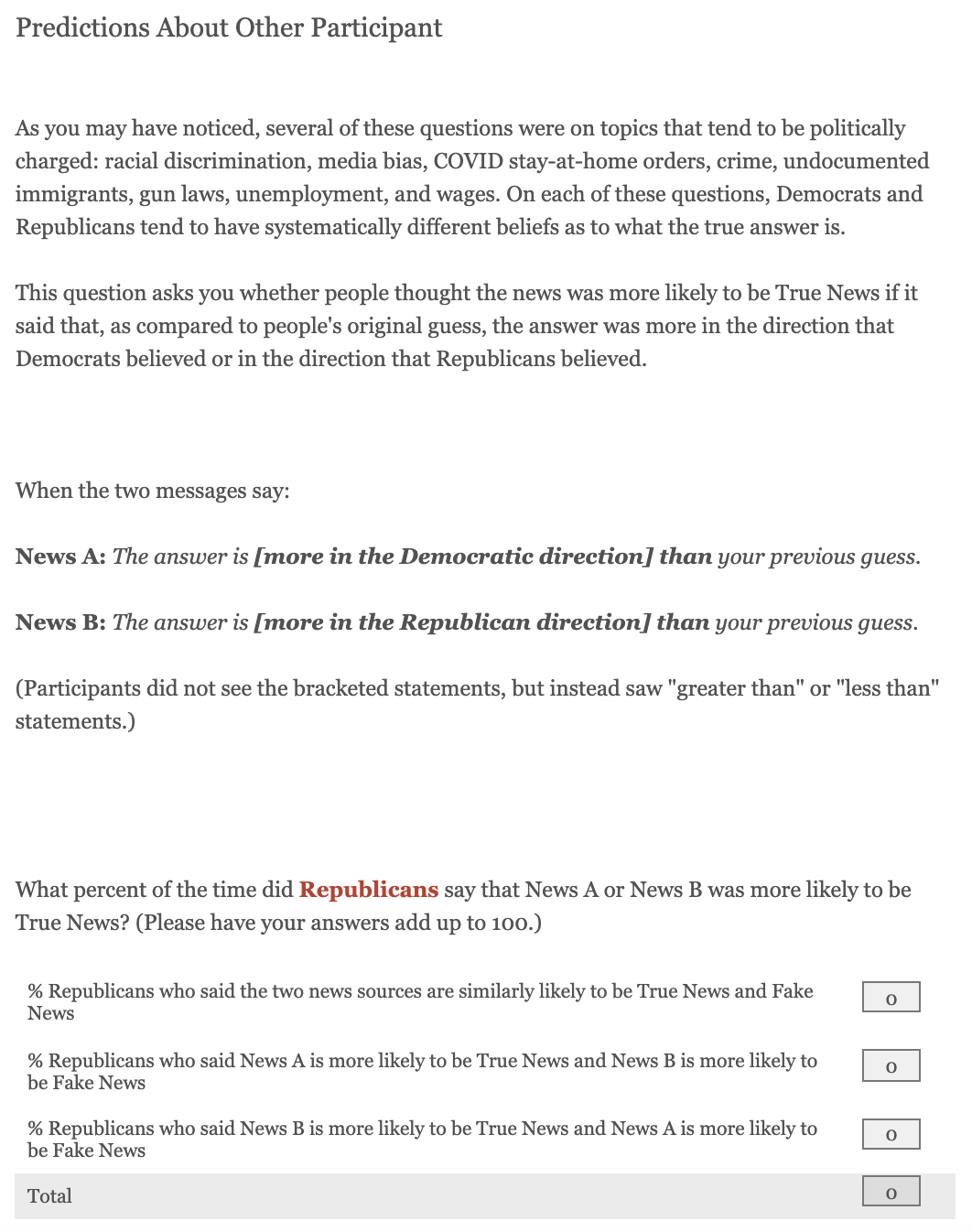}
\end{center}
\end{figure}

\clearpage

\begin{center}
\includegraphics[width = \textwidth]{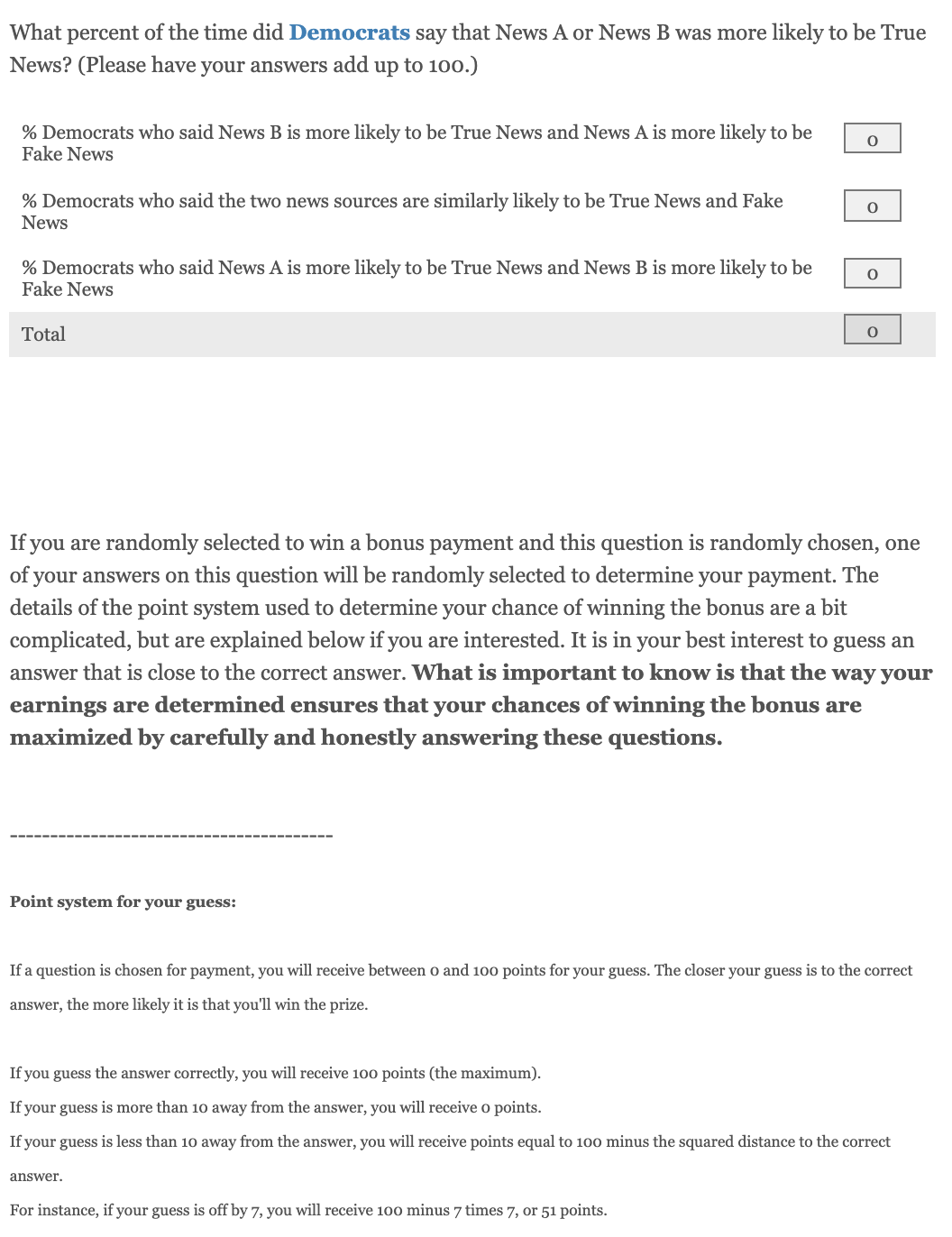}
\end{center}













\end{document}